\newif\iftr\trtrue
\documentclass[10pt]{sigplanconf}

\def\OPTIONConf{1}         %
\def\OPTIONLoudLabels{0}         %

\usepackage{main}
\usepackage{alltt}
\usepackage{balance}
\usepackage[firstpage]{draftwatermark}
\SetWatermarkText{\hspace*{8in}\raisebox{7.0in}{\includegraphics[scale=0.1]{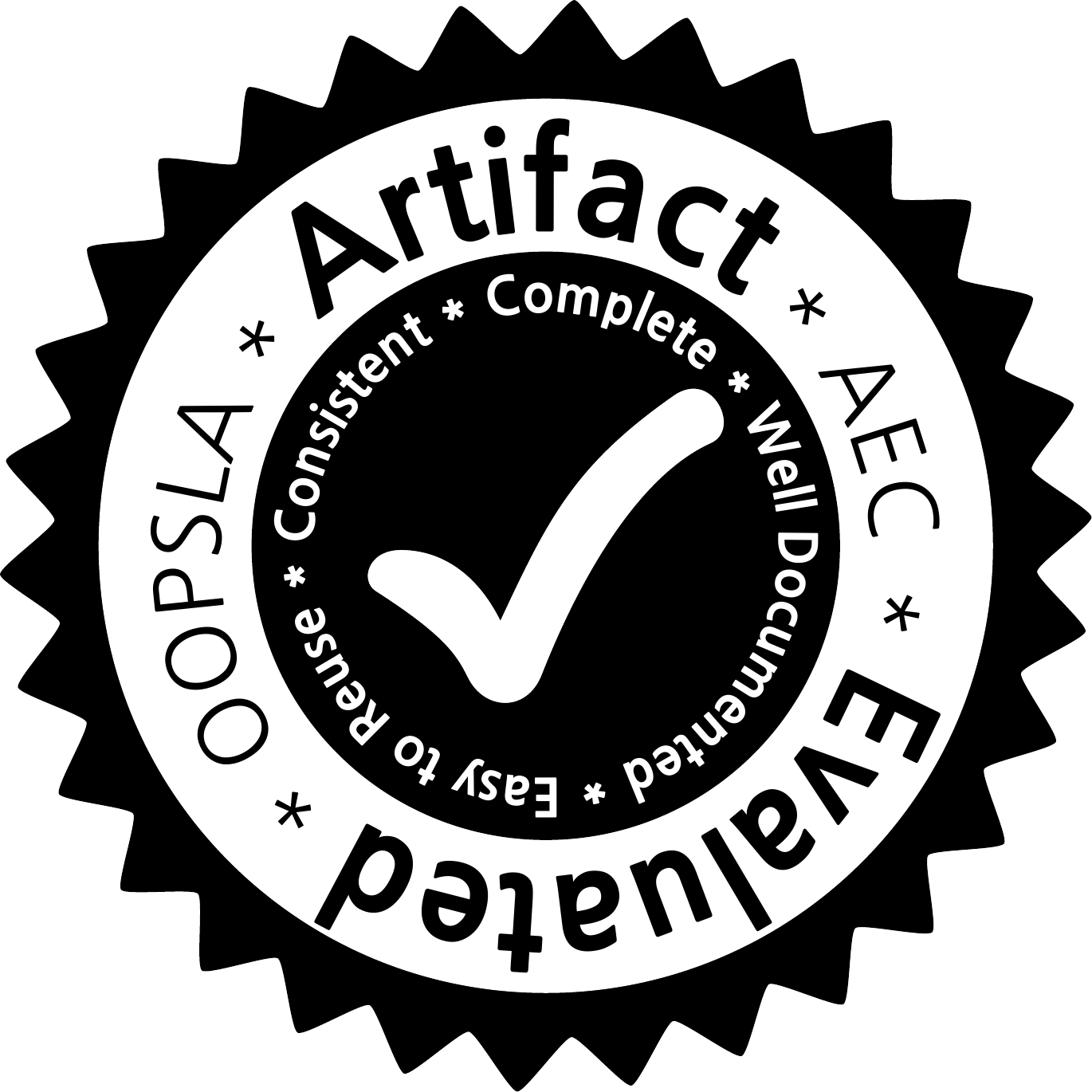}}}
\SetWatermarkAngle{0}

\usepackage{srcltx}
\usepackage{latexsym,amsthm,amsmath,amsfonts,amssymb,stmaryrd}
\usepackage{arcs}
\usepackage{euler}

\usepackage{xspace}
\usepackage{enumerate}
\usepackage[breaklinks]{hyperref}
\usepackage{breakurl}
\usepackage[hyperref,svgnames,table]{xcolor}   % pdftex option gratuitously prevents targeting dvi
\usepackage{graphicx,color}
\usepackage{color}
\usepackage{comment}
\usepackage{listings}
\usepackage{supertabular}
\usepackage{stackrel}
\usepackage{marvosym}
\usepackage{pgf,tikz}
\usepackage{tikz-qtree}
\usepackage{subcaption}
\usepackage{url}
\usepackage{pifont}

\bibpunct{(}{)}{;}{a}{}{,}

\usepackage{mathpartir}
\usepackage{etoolbox}
\usepackage{color,soul} % Highlighting

\usepackage{jdproof} % Thanks J!

\newcommand{\arrayenvc}[1]{\renewcommand{\arraystretch}{1} \begin{array}[c]{@{}c@{}}#1\end{array}}
\newcommand{\arrayenvcl}[1]{\renewcommand{\arraystretch}{1} \begin{array}[c]{@{}l@{}}#1\end{array}}
\newcommand{\arrayenvr}[1]{\renewcommand{\arraystretch}{1} \begin{array}[t]{@{}r@{}}#1\end{array}}

\newcommand{\arrayenvl}[1]{\renewcommand{\arraystretch}{1} \begin{array}[t]{@{}l@{}}#1\end{array}}
 
\newcommand{\arrayenvbl}[1]{\renewcommand{\arraystretch}{1}  \begin{array}[b]{@{}l@{}}#1\end{array}}

\newcommand{\tabularenvl}[1]{\renewcommand{\arraystretch}{1} \!\begin{tabular}[t]{@{}l@{}}#1\end{tabular}}

\usepackage[implicitLineBreakHack]{ottalt}

\newcommand{\invisiclues}[1]{{}}

\newenvironment{arraybl}%
    {\begin{array}[b]{@{}l@{}}}%
    {\end{array}}

\newdimen\zzfontsz
\newcommand{\fontsz}[2]{\zzfontsz=#1%
  {\fontsize{\zzfontsz}{1.2\zzfontsz}\selectfont{#2}}}

\newcommand{\runonfontsz}[1]{\zzfontsz=#1%
\fontsize{\zzfontsz}{1.2\zzfontsz}\selectfont}

\newcommand{\mathsz}[2]{\text{\fontsz{#1}{$#2$}}}

\newcommand{\runonmathsz}[1]{\runonfontsz{#1}}

\newcommand{\ruleinnerprelude}{\normalfont\sf}
\newcommand{\ruleunderscorereplacement}{-}

\newcommand{\rulename}[1]{\text{\ruleinnerprelude #1}}

\newcommand{\judgboxfontsize}[1]{%
    \ifnum\OPTIONConf=1%
        \fontsz{11pt}{\strut#1}%
    \else%
        \fontsz{13pt}{\strut#1}%
    \fi}

\newcommand{\judgboxmathsize}[1]{%
    \ifnum\OPTIONConf=1%
        \mathsz{11pt}{\mathstrut#1}%
    \else%
        \mathsz{13pt}{\mathstrut#1}%
    \fi}

\newcommand{\lesscaptionspacing}{\vspace{-1pt}}

\newcommand{\ottdrule}[4][]{{\displaystyle\frac{\begin{array}{l}#2\end{array}}{#3}\quad\ottdrulename{#4}}}
\newcommand{\ottusedrule}[1]{\[#1\]}
\newcommand{\ottpremise}[1]{ #1 \\}
\newenvironment{ottdefnblock}[3][]{ \framebox{\mbox{#2}} \quad #3 \\[0pt]}{}
\newenvironment{ottfundefnblock}[3][]{ \framebox{\mbox{#2}} \quad #3 \\[0pt]\begin{displaymath}\begin{array}{l}}{\end{array}\end{displaymath}}

\newcommand{\ottmv}[1]{\mathit{#1}}
\newcommand{\ottkw}[1]{\mathbf{#1}}
\newcommand{\ottsym}[1]{#1}
\newcommand{\ottcom}[1]{\text{#1}}
\newcommand{\ottdrulename}[1]{\textsc{#1}}

\newcommand{\ottprodline}[6]{& & $#1$ & $#2$ & $#3 #4$ & $#5$ & $#6$}

\newcommand{\ottafterlastrule}{\\}

\newcommand{\ottdruleGrwfXXemp}[1]{\ottdrule[#1]{%
}{
 G   \vdash    \varepsilon  \; {\txtsf{wf} } }{%
{\ottdrulename{Grwf\_emp}}{}%
}}

\newcommand{\ottdruleGrwfXXval}[1]{\ottdrule[#1]{%
\ottpremise{ G   \vdash   H \; {\txtsf{wf} } }%
\ottpremise{ G ( p )   \ottsym{=}  v}%
}{
 G   \vdash   H  \ottsym{,}   p {:} v  \; {\txtsf{wf} } }{%
{\ottdrulename{Grwf\_val}}{}%
}}

\newcommand{\ottdruleGrwfXXthunk}[1]{\ottdrule[#1]{%
\ottpremise{ G   \vdash   H \; {\txtsf{wf} } }%
\ottpremise{{\begin{arraybl}
{ G ( p )   \ottsym{=}  e}%
\end{arraybl}}}
}{
 G   \vdash   H  \ottsym{,}   p {:} e  \; {\txtsf{wf} } }{%
{\ottdrulename{Grwf\_thunk}}{}%
}}

\newcommand{\ottdruleGrwfXXthunkCache}[1]{\ottdrule[#1]{%
\ottpremise{ G   \vdash   H \; {\txtsf{wf} } }%
\ottpremise{{\begin{arraybl}
{ G ( p ) = ( e , \mathrm{t} ) }%
\\
{ \text{if~}  \txtsf{all-clean-out}( G , p )  \text{~then~} \ottsym{(}   \D  \derives   G   \vdash ^{ p }_{ \omega }  e   \Downarrow   G  ;  \mathrm{t}    \ottsym{)} }%
\\
{ \text{if~}  \text{not~}  \txtsf{all-clean-out}( G , p )   \text{~then~}  \txtsf{all-dirty-in}( G , p )  }%
\end{arraybl}}}
}{
 G   \vdash   H  \ottsym{,}   p {:}( e , \mathrm{t} )  \; {\txtsf{wf} } }{%
{\ottdrulename{Grwf\_thunkCache}}{}%
}}

\newcommand{\ottdruleGrwfXXdirtyEdge}[1]{\ottdrule[#1]{%
\ottpremise{ G   \vdash   H \; {\txtsf{wf} } }%
\ottpremise{{\begin{arraybl}
{ \ottsym{(}  p  \ottsym{,}  a  \ottsym{,}   \txtsf{dirty}   \ottsym{,}  q  \ottsym{)}  \in  G }%
\\
{ q  \in   \txtsf{dom}( H )  }%
\\
{ \text{if~}  \ottsym{(}  p_{{\mathrm{0}}}  \ottsym{,}  a_{{\mathrm{0}}}  \ottsym{,}  b_{{\mathrm{0}}}  \ottsym{,}  p  \ottsym{)}  \in  G  \text{~then~} b_{{\mathrm{0}}}  \ottsym{=}   \txtsf{dirty}  }%
\end{arraybl}}}
}{
 G   \vdash   H  \ottsym{,}  \ottsym{(}  p  \ottsym{,}  a  \ottsym{,}   \txtsf{dirty}   \ottsym{,}  q  \ottsym{)} \; {\txtsf{wf} } }{%
{\ottdrulename{Grwf\_dirtyEdge}}{}%
}}

\newcommand{\ottdruleGrwfXXcleanEdge}[1]{\ottdrule[#1]{%
\ottpremise{ G   \vdash   H \; {\txtsf{wf} } }%
\ottpremise{{\begin{arraybl}
{ \ottsym{(}  p  \ottsym{,}  a  \ottsym{,}   \txtsf{clean}   \ottsym{,}  q  \ottsym{)}  \in  G }%
\\
{ \txtsf{consistent-action} ( H ,  a ,  q  ) }%
\\
{ \txtsf{all-clean-out}( G , q ) }%
\end{arraybl}}}
}{
 G   \vdash   H  \ottsym{,}  \ottsym{(}  p  \ottsym{,}  a  \ottsym{,}   \txtsf{clean}   \ottsym{,}  q  \ottsym{)} \; {\txtsf{wf} } }{%
{\ottdrulename{Grwf\_cleanEdge}}{}%
}}

\newcommand{\ottdefngrwf}[1]{\begin{ottdefnblock}[#1]{$ G   \vdash   H \; {\txtsf{wf} } $}{}
\ottusedrule{\ottdruleGrwfXXemp{}}
\ottusedrule{\ottdruleGrwfXXval{}}
\ottusedrule{\ottdruleGrwfXXthunk{}}
\ottusedrule{\ottdruleGrwfXXthunkCache{}}
\ottusedrule{\ottdruleGrwfXXdirtyEdge{}}
\ottusedrule{\ottdruleGrwfXXcleanEdge{}}
\end{ottdefnblock}}

\newcommand{\ottdefnsGrwf}{
\ottdefngrwf{}}

\newcommand{\ottdruleEvalXXterm}[1]{\ottdrule[#1]{%
}{
 G   \vdash ^{ p }_{ \omega }  \mathrm{t}   \Downarrow   G  ;  \mathrm{t} }{%
{\ottdrulename{Eval\_term}}{}%
}}

\newcommand{\ottdruleEvalXXapp}[1]{\ottdrule[#1]{%
\ottpremise{{\begin{arraybl}
{ G_{{\mathrm{1}}}   \vdash ^{ p }_{ \omega }  e_{{\mathrm{1}}}   \Downarrow   G_{{\mathrm{2}}}  ;  \lambda  x  \ottsym{.}  e_{{\mathrm{2}}} }%
\\
{ G_{{\mathrm{2}}}   \vdash ^{ p }_{ \omega }  \ottsym{[}  v  \ottsym{/}  x  \ottsym{]}  e_{{\mathrm{2}}}   \Downarrow   G_{{\mathrm{3}}}  ;  \mathrm{t} }%
\end{arraybl}}}
}{
 G_{{\mathrm{1}}}   \vdash ^{ p }_{ \omega }  e_{{\mathrm{1}}} \, v   \Downarrow   G_{{\mathrm{3}}}  ;  \mathrm{t} }{%
{\ottdrulename{Eval\_app}}{}%
}}

\newcommand{\ottdruleEvalXXfix}[1]{\ottdrule[#1]{%
\ottpremise{ G_{{\mathrm{1}}}   \vdash ^{ p }_{ \omega }  \ottsym{[}  \ottsym{(}  \ottkw{fix} \, f  \ottsym{.}  e  \ottsym{)}  \ottsym{/}  f  \ottsym{]}  e   \Downarrow   G_{{\mathrm{2}}}  ;  \mathrm{t} }%
}{
 G_{{\mathrm{1}}}   \vdash ^{ p }_{ \omega }  \ottkw{fix} \, f  \ottsym{.}  e   \Downarrow   G_{{\mathrm{2}}}  ;  \mathrm{t} }{%
{\ottdrulename{Eval\_fix}}{}%
}}

\newcommand{\ottdruleEvalXXbind}[1]{\ottdrule[#1]{%
\ottpremise{{\begin{arraybl}
{ G_{{\mathrm{1}}}   \vdash ^{ p }_{ \omega }  e_{{\mathrm{1}}}   \Downarrow   G_{{\mathrm{2}}}  ;  \ottkw{ret} \, v }%
\\
{ G_{{\mathrm{2}}}   \vdash ^{ p }_{ \omega }  \ottsym{[}  v  \ottsym{/}  x  \ottsym{]}  e_{{\mathrm{2}}}   \Downarrow   G_{{\mathrm{3}}}  ;  \mathrm{t} }%
\end{arraybl}}}
}{
 G_{{\mathrm{1}}}   \vdash ^{ p }_{ \omega }   \textbf{let}\, x \,{\leftarrow}\, e_{{\mathrm{1}}} \, \ottkw{in} \, e_{{\mathrm{2}}}    \Downarrow   G_{{\mathrm{3}}}  ;  \mathrm{t} }{%
{\ottdrulename{Eval\_bind}}{}%
}}

\newcommand{\ottdruleEvalXXcase}[1]{\ottdrule[#1]{%
\ottpremise{ G_{{\mathrm{1}}}   \vdash ^{ p }_{ \omega }  \ottsym{[}  v  \ottsym{/}  x_{\ottmv{i}}  \ottsym{]}  e_{\ottmv{i}}   \Downarrow   G_{{\mathrm{2}}}  ;  \mathrm{t} }%
}{
 G_{{\mathrm{1}}}   \vdash ^{ p }_{ \omega }  \ottkw{case} \, \ottsym{(}   \ottkw{inj} _{  i  }~ v   \ottsym{,}  x_{{\mathrm{1}}}  \ottsym{.}  e_{{\mathrm{1}}}  \ottsym{,}  x_{{\mathrm{2}}}  \ottsym{.}  e_{{\mathrm{2}}}  \ottsym{)}   \Downarrow   G_{{\mathrm{2}}}  ;  \mathrm{t} }{%
{\ottdrulename{Eval\_case}}{}%
}}

\newcommand{\ottdruleEvalXXsplit}[1]{\ottdrule[#1]{%
\ottpremise{ G_{{\mathrm{1}}}   \vdash ^{ p }_{ \omega }  \ottsym{[}  v_{{\mathrm{1}}}  \ottsym{/}  x_{{\mathrm{1}}}  \ottsym{]}  \ottsym{[}  v_{{\mathrm{2}}}  \ottsym{/}  x_{{\mathrm{2}}}  \ottsym{]}  e   \Downarrow   G_{{\mathrm{2}}}  ;  \mathrm{t} }%
}{
 G_{{\mathrm{1}}}   \vdash ^{ p }_{ \omega }  \ottkw{split} \, \ottsym{(}  \ottsym{(}  v_{{\mathrm{1}}}  \ottsym{,}  v_{{\mathrm{2}}}  \ottsym{)}  \ottsym{,}  x_{{\mathrm{1}}}  \ottsym{.}  x_{{\mathrm{2}}}  \ottsym{.}  e  \ottsym{)}   \Downarrow   G_{{\mathrm{2}}}  ;  \mathrm{t} }{%
{\ottdrulename{Eval\_split}}{}%
}}

\newcommand{\ottdruleEvalXXfork}[1]{\ottdrule[#1]{%
}{
 G   \vdash ^{ p }_{ \omega }   \Grn{\textbf{fork}(  \Grn{ \ottkw{nm} \, k }  )}    \Downarrow   G  ;  \ottkw{ret} \, \ottsym{(}   \Grn{ \ottkw{nm} \,  { k }{\cdot}{ \ottsym{1} }  }   \ottsym{,}   \Grn{ \ottkw{nm} \,  { k }{\cdot}{ \ottsym{2} }  }   \ottsym{)} }{%
{\ottdrulename{Eval\_fork}}{}%
}}

\newcommand{\ottdruleEvalXXrefPlain}[1]{\ottdrule[#1]{%
\ottpremise{{\begin{arraybl}
{q  \ottsym{=}   k  @  \omega }%
\\
{ \Grn{ q } \notin   \txtsf{dom}( G_{{\mathrm{1}}} )  }%
\end{arraybl}}}
\ottpremise{ G_{{\mathrm{1}}} \{ q {\mapsto} v \}   \ottsym{=}  G_{{\mathrm{2}}}}%
}{
 G_{{\mathrm{1}}}   \vdash ^{ p }_{ \omega }   \textbf{ref}(\Grn{  \Grn{ \ottkw{nm} \, k }  }, v )    \Downarrow   G_{{\mathrm{2}}}  ;  \ottkw{ret} \,  \ottkw{ref} \,\Grn{ q }  }{%
{\ottdrulename{Eval\_refPlain}}{}%
}}

\newcommand{\ottdruleEvalXXrefClean}[1]{\ottdrule[#1]{%
\ottpremise{q  \ottsym{=}   k  @  \omega }%
\ottpremise{ G ( q )   \ottsym{=}  v}%
}{
 G   \vdash ^{ p }_{ \omega }   \textbf{ref}(\Grn{  \Grn{ \ottkw{nm} \, k }  }, v )    \Downarrow   G  \ottsym{,}  \ottsym{(}  p  \ottsym{,}  \ottkw{alloc} \, v  \ottsym{,}   \txtsf{clean}   \ottsym{,}  q  \ottsym{)}  ;  \ottkw{ret} \,  \ottkw{ref} \,\Grn{ q }  }{%
{\ottdrulename{Eval\_refClean}}{}%
}}

\newcommand{\ottdruleEvalXXrefDirty}[1]{\ottdrule[#1]{%
\ottpremise{q  \ottsym{=}   k  @  \omega }%
\ottpremise{ G_{{\mathrm{1}}} \{ q {\mapsto} v \}   \ottsym{=}  G_{{\mathrm{2}}}}%
\ottpremise{ \txtsf{dirty-paths-in} ( G_{{\mathrm{2}}} , q )   \ottsym{=}  G_{{\mathrm{3}}}}%
}{
 G_{{\mathrm{1}}}   \vdash ^{ p }_{ \omega }   \textbf{ref}(\Grn{  \Grn{ \ottkw{nm} \, k }  }, v )    \Downarrow   G_{{\mathrm{3}}}  \ottsym{,}  \ottsym{(}  p  \ottsym{,}  \ottkw{alloc} \, v  \ottsym{,}   \txtsf{clean}   \ottsym{,}  q  \ottsym{)}  ;  \ottkw{ret} \,  \ottkw{ref} \,\Grn{ q }  }{%
{\ottdrulename{Eval\_refDirty}}{}%
}}

\newcommand{\ottdruleEvalXXgetPlain}[1]{\ottdrule[#1]{%
\ottpremise{ G ( q )   \ottsym{=}  v}%
}{
 G   \vdash ^{ p }_{ \omega }  \ottkw{get} \, \ottsym{(}   \ottkw{ref} \,\Grn{ q }   \ottsym{)}   \Downarrow   G  ;  \ottkw{ret} \, v }{%
{\ottdrulename{Eval\_getPlain}}{}%
}}

\newcommand{\ottdruleEvalXXgetClean}[1]{\ottdrule[#1]{%
\ottpremise{ G ( q )   \ottsym{=}  v}%
}{
 G   \vdash ^{ p }_{ \omega }  \ottkw{get} \, \ottsym{(}   \ottkw{ref} \,\Grn{ q }   \ottsym{)}   \Downarrow   G  \ottsym{,}  \ottsym{(}  p  \ottsym{,}  \ottkw{obs} \, v  \ottsym{,}   \txtsf{clean}   \ottsym{,}  q  \ottsym{)}  ;  \ottkw{ret} \, v }{%
{\ottdrulename{Eval\_getClean}}{}%
}}

\newcommand{\ottdruleEvalXXthunkPlain}[1]{\ottdrule[#1]{%
\ottpremise{{\begin{arraybl}
{q  \ottsym{=}   k  @  \omega }%
\\
{ \Grn{ q } \notin   \txtsf{dom}( G_{{\mathrm{1}}} )  }%
\end{arraybl}}}
\ottpremise{ G_{{\mathrm{1}}} \{ q {\mapsto} e \}   \ottsym{=}  G_{{\mathrm{2}}}}%
}{
 G_{{\mathrm{1}}}   \vdash ^{ p }_{ \omega }   \textbf{thunk}(\Grn{  \Grn{ \ottkw{nm} \, k }  }, e )    \Downarrow   G_{{\mathrm{2}}}  ;  \ottkw{ret} \, \ottsym{(}   \ottkw{thk} \,\Grn{ q }   \ottsym{)} }{%
{\ottdrulename{Eval\_thunkPlain}}{}%
}}

\newcommand{\ottdruleEvalXXthunkDirty}[1]{\ottdrule[#1]{%
\ottpremise{q  \ottsym{=}   k  @  \omega }%
\ottpremise{ G_{{\mathrm{1}}} \{ q {\mapsto} e \}   \ottsym{=}  G_{{\mathrm{2}}}}%
\ottpremise{ \txtsf{dirty-paths-in} ( G_{{\mathrm{2}}} , q )   \ottsym{=}  G_{{\mathrm{3}}}}%
}{
 G_{{\mathrm{1}}}   \vdash ^{ p }_{ \omega }   \textbf{thunk}(\Grn{  \Grn{ \ottkw{nm} \, k }  }, e )    \Downarrow   G_{{\mathrm{3}}}  \ottsym{,}  \ottsym{(}  p  \ottsym{,}  \ottkw{alloc} \, e  \ottsym{,}   \txtsf{clean}   \ottsym{,}  q  \ottsym{)}  ;  \ottkw{ret} \, \ottsym{(}   \ottkw{thk} \,\Grn{ q }   \ottsym{)} }{%
{\ottdrulename{Eval\_thunkDirty}}{}%
}}

\newcommand{\ottdruleEvalXXthunkClean}[1]{\ottdrule[#1]{%
\ottpremise{q  \ottsym{=}   k  @  \omega }%
\ottpremise{ \txtsf{exp}( G ,  q )   \ottsym{=}  e}%
}{
 G   \vdash ^{ p }_{ \omega }   \textbf{thunk}(\Grn{  \Grn{ \ottkw{nm} \, k }  }, e )    \Downarrow   G  \ottsym{,}  \ottsym{(}  p  \ottsym{,}  \ottkw{alloc} \, e  \ottsym{,}   \txtsf{clean}   \ottsym{,}  q  \ottsym{)}  ;  \ottkw{ret} \, \ottsym{(}   \ottkw{thk} \,\Grn{ q }   \ottsym{)} }{%
{\ottdrulename{Eval\_thunkClean}}{}%
}}

\newcommand{\ottdruleEvalXXforcePlain}[1]{\ottdrule[#1]{%
\ottpremise{{\begin{arraybl}
{ G_{{\mathrm{1}}} ( q )   \ottsym{=}  e}%
\\
{ G_{{\mathrm{1}}}   \vdash ^{ q }_{  \txtsf{namespace}(  q  )  }  e   \Downarrow   G_{{\mathrm{2}}}  ;  \mathrm{t} }%
\end{arraybl}}}
}{
 G_{{\mathrm{1}}}   \vdash ^{ p }_{ \omega }  \ottkw{force} \, \ottsym{(}   \ottkw{thk} \,\Grn{ q }   \ottsym{)}   \Downarrow   G_{{\mathrm{2}}}  ;  \mathrm{t} }{%
{\ottdrulename{Eval\_forcePlain}}{}%
}}

\newcommand{\ottdruleEvalXXforceClean}[1]{\ottdrule[#1]{%
\ottpremise{ G ( q ) = ( e , \mathrm{t} ) }%
\ottpremise{ \txtsf{all-clean-out}( G , q ) }%
}{
 G   \vdash ^{ p }_{ \omega }  \ottkw{force} \, \ottsym{(}   \ottkw{thk} \,\Grn{ q }   \ottsym{)}   \Downarrow   G  \ottsym{,}  \ottsym{(}  p  \ottsym{,}  \ottkw{obs} \, \mathrm{t}  \ottsym{,}   \txtsf{clean}   \ottsym{,}  q  \ottsym{)}  ;  \mathrm{t} }{%
{\ottdrulename{Eval\_forceClean}}{}%
}}

\newcommand{\ottdruleEvalXXcomputeDep}[1]{\ottdrule[#1]{%
\ottpremise{{\begin{arraybl}
{ \txtsf{exp}( G_{{\mathrm{1}}} ,  q )   \ottsym{=}  e'}%
\\
{ \txtsf{del-edges-out} (  G_{{\mathrm{1}}} \{ q {\mapsto} e' \}  , q )   \ottsym{=}  G'_{{\mathrm{1}}}}%
\\
{ G'_{{\mathrm{1}}}   \vdash ^{ q }_{  \txtsf{namespace}(  q  )  }  e'   \Downarrow   G_{{\mathrm{2}}}  ;  \mathrm{t}' }%
\end{arraybl}}}
\ottpremise{{\begin{arraybl}
{ G_{{\mathrm{2}}} \{ q {\mapsto}( e' , \mathrm{t}' )\}   \ottsym{=}  G'_{{\mathrm{2}}}}%
\\
{ \txtsf{all-clean-out}( G'_{{\mathrm{2}}} , q ) }%
\\
{ G'_{{\mathrm{2}}}   \vdash ^{ p }_{ \omega }  \ottkw{force} \, \ottsym{(}   \ottkw{thk} \,\Grn{ p_{{\mathrm{0}}} }   \ottsym{)}   \Downarrow   G_{{\mathrm{3}}}  ;  \mathrm{t} }%
\end{arraybl}}}
}{
 G_{{\mathrm{1}}}   \vdash ^{ p }_{ \omega }  \ottkw{force} \, \ottsym{(}   \ottkw{thk} \,\Grn{ p_{{\mathrm{0}}} }   \ottsym{)}   \Downarrow   G_{{\mathrm{3}}}  ;  \mathrm{t} }{%
{\ottdrulename{Eval\_computeDep}}{}%
}}

\newcommand{\ottdruleEvalXXscrubEdge}[1]{\ottdrule[#1]{%
\ottpremise{{\begin{arraybl}
{ \txtsf{all-clean-out}( \ottsym{(}  G_{{\mathrm{1}}}  \ottsym{,}  G_{{\mathrm{2}}}  \ottsym{)} , q_{{\mathrm{2}}} ) }%
\\
{ \txtsf{consistent-action} ( \ottsym{(}  G_{{\mathrm{1}}}  \ottsym{,}  G_{{\mathrm{2}}}  \ottsym{)} ,  a ,  q_{{\mathrm{2}}}  ) }%
\\
{ G_{{\mathrm{1}}}  \ottsym{,}  \ottsym{(}  q_{{\mathrm{1}}}  \ottsym{,}  a  \ottsym{,}   \txtsf{clean}   \ottsym{,}  q_{{\mathrm{2}}}  \ottsym{)}  \ottsym{,}  G_{{\mathrm{2}}}   \vdash ^{ p }_{ \omega }  \ottkw{force} \, \ottsym{(}   \ottkw{thk} \,\Grn{ p_{{\mathrm{0}}} }   \ottsym{)}   \Downarrow   G_{{\mathrm{3}}}  ;  \mathrm{t} }%
\end{arraybl}}}
}{
 G_{{\mathrm{1}}}  \ottsym{,}  \ottsym{(}  q_{{\mathrm{1}}}  \ottsym{,}  a  \ottsym{,}   \txtsf{dirty}   \ottsym{,}  q_{{\mathrm{2}}}  \ottsym{)}  \ottsym{,}  G_{{\mathrm{2}}}   \vdash ^{ p }_{ \omega }  \ottkw{force} \, \ottsym{(}   \ottkw{thk} \,\Grn{ p_{{\mathrm{0}}} }   \ottsym{)}   \Downarrow   G_{{\mathrm{3}}}  ;  \mathrm{t} }{%
{\ottdrulename{Eval\_scrubEdge}}{}%
}}

\newcommand{\ottdruleEvalXXnamespace}[1]{\ottdrule[#1]{%
\ottpremise{ G_{{\mathrm{1}}}   \vdash ^{ p }_{ \omega }  \ottsym{[}   \Grn{\textbf{ns}\,  \omega . k  }   \ottsym{/}  x  \ottsym{]}  e   \Downarrow   G_{{\mathrm{2}}}  ;  \mathrm{t} }%
}{
 G_{{\mathrm{1}}}   \vdash ^{ p }_{ \omega }   \Grn{\textbf{ns}\,(  \Grn{ \ottkw{nm} \, k }  , x . e )}    \Downarrow   G_{{\mathrm{2}}}  ;  \mathrm{t} }{%
{\ottdrulename{Eval\_namespace}}{}%
}}

\newcommand{\ottdruleEvalXXnest}[1]{\ottdrule[#1]{%
\ottpremise{ G_{{\mathrm{1}}}   \vdash ^{ p }_{ \mu }  e_{{\mathrm{1}}}   \Downarrow   G_{{\mathrm{2}}}  ;  \ottkw{ret} \, v }%
\ottpremise{ G_{{\mathrm{2}}}   \vdash ^{ p }_{ \omega }  \ottsym{[}  v  \ottsym{/}  x  \ottsym{]}  e_{{\mathrm{2}}}   \Downarrow   G_{{\mathrm{3}}}  ;  \mathrm{t} }%
}{
 G_{{\mathrm{1}}}   \vdash ^{ p }_{ \omega }   \Grn{\textbf{nest}(  \Grn{\textbf{ns}\, \mu }  , e_{{\mathrm{1}}} , x . e_{{\mathrm{2}}} )}    \Downarrow   G_{{\mathrm{3}}}  ;  \mathrm{t} }{%
{\ottdrulename{Eval\_nest}}{}%
}}

  \renewottcommands[ott]

\renewcommand{\ottprodline}[6]{\ifthenelse{\equal{#3}{} }{ $\mid$~\hbox{$#2$}}{}}

\newcommand{\ottendprodlines}{}

\renewcommand{\ottafterlastrule}{\ottendprodlines}

\renewcommand{\ottkw}[1]{\textbf{#1}}

\newcommand{\derives}{\ensuremath{\mathrel{::}}}

\newcommand{\mytodo}[1]%
{\begin{spacing}{0.5}\todo{#1}\end{spacing}\xspace}%

\renewcommand{\topfraction}{0.99}	% max fraction of floats at top
\renewcommand{\floatpagefraction}{0.7} % require fuller float pages
\definecolor{lightgreen}{rgb}{.85,.95,.85}
\definecolor{lightblue}{rgb}{.85,.90,1}
\definecolor{lightred}{rgb}{.95,.85,.85}
\definecolor{lightgrey}{rgb}{.95,.95,.95}
\sethlcolor{lightred}

\definecolor{drkyellow}{rgb}{0.4,0.3,0}
\definecolor{drkred}{rgb}{0.5,0,0}
\definecolor{drkgreen}{rgb}{0,0.5,0}

\definecolor{medgreen}{rgb}{0.2,0.60,0.2}

\definecolor{dkred}{rgb}{0.5,0,0}
\definecolor{gray}{rgb}{0.5,0.5,0.5}
\definecolor{mauve}{rgb}{0.58,0,0.82}

\definecolor{hickw}{rgb}{0.2,0.2,0.8}
\definecolor{type}{rgb}{0.6,0.2,0.6}
\definecolor{typebg}{rgb}{0.1,0.0,0.1}
\definecolor{punc}{rgb}{0.2,0.2,0.5}

\definecolor{highlightcolor}{rgb}{0.1,0,0.1}

\newcommand{\secref}[1]{Section~\ref{sec:#1}}

\newcommand{\secreftwo}[2]{Sections~\ref{sec:#1} and~\ref{sec:#2}}

\newcommand{\figref}[1]{Figure~\ref{fig:#1}}

\newtoggle{figoflisting} \togglefalse{figoflisting}
\newtoggle{usectable}    \toggletrue{usectable}

\newcommand{\codeLineL}[1]{}

\newcommand{\myinline}[1]{\lstinline!#1!}

\newcommand{\keywordcolor}{dBlue}

\newcommand{\cod}[1]{\textsf{#1}}

\newcommand{\kw}[1]{\textcolor{\keywordcolor}{\textsf{#1}}}

\newtheoremstyle{slplain}%
        {3pt}% <Space above>
        {3pt}% <Space below>
        {\slshape}% <Body font>
        {}% <Indent amount>
        {\bf}% <Theorem head font>
        {.}% <Punctuation after theorem head>
        {.2em}% <Space after theorem head>
        {}%
\theoremstyle{slplain}
\newtheorem{thm}{Theorem}[section]
\newtheorem*{thm*}{Theorem}
\newtheorem{lem}[thm]{Lemma}
\newtheorem{conj}[thm]{Conjecture}

\makeatletter
\renewenvironment{proof}[1][\proofname]{\par
  \pushQED{\qed}%
  \normalfont
  \topsep2pt \partopsep2pt % no space before
  \trivlist
  \item[\hskip\labelsep
        \itshape
    #1\@addpunct{.}]\ignorespaces
}{%
  \popQED\endtrivlist\@endpefalse
  \addvspace{2pt plus 2pt} % some space after
}

\renewenvironment{proof}[1][\proofname]{\par%
  \vspace{-4pt}%
  \pushQED{\qed}%
  \normalfont%\topsep6\p@\@plus6\p@\relax
  \trivlist
  \item[\hskip\labelsep
        \itshape
    #1\@addpunct{.}]\ignorespaces%
}{%
  \popQED\endtrivlist\@endpefalse
}
\makeatother

\theoremstyle{definition}
\newtheorem{defn}[thm]{Definition}

\theoremstyle{remark}

\newcounter{cons}

\newcommand{\textvtt}[1]{{\normalfont\fontfamily{cmvtt}\selectfont {#1}}}

\newcommand{\marginnote}[1]{\marginpar{\raggedright\scriptsize{#1}}}

\newcommand{\LoudLabel}[1]{\label{#1}%
\ifnum\OPTIONLoudLabels=1%
  \marginnote{\tiny\textvtt{#1}}%
\fi%
}

\ifnum\OPTIONLoudLabels=1%
\newcommand{\Label}[1]{\LoudLabel{#1}}%
\newcommand{\FLabel}[1]{\label{#1}%
{\tt\scriptsize{#1}}}%
\else%
\newcommand{\Label}[1]{\label{#1}}%
\newcommand{\FLabel}[1]{\label{#1}}%
\fi

\newcommand{\textgraybox}[1]{\boxed{#1}}
\newcommand{\loudjudgbox}[3]{%
      {\ifx\\#1\\%empty
        \else %not empty
           \ifnum\OPTIONConf=0
               \judgboxfontsize{\text{\textcolor{dBlue}{\textvtt{\Backslash{#1}}}}}
           \fi
        \fi
        \raggedright \textgraybox{\ensuremath{\judgboxmathsize{#2}}}\!%
        \fontsz{9pt}{\begin{tabular}[c]{l} #3 \end{tabular}} %
}}

\newcommand{\txtsf}[1]{{\normalfont\textsf{#1}}}
\newcommand{\union}{\mathrel{\cup}}

\newcommand{\monotonic}{~\txtsf{monotonic}}
\newcommand{\satisfactory}{~\txtsf{satisfactory}}

\renewcommand{\topfraction}{0.999999999999999}

\renewcommand{\floatpagefraction}{0.99999999999}

\definecolor{dRed}{rgb}{0.65, 0.0, 0.0}
\definecolor{dGreen}{rgb}{0.05, 0.30, 0.20}
\definecolor{DRED}{rgb}{0.65, 0.0, 0.0}

\newcommand{\authnote}[2]{{\footnotesize\color{dRed}[{#1}: #2]}}

\ifdefined\nocomments
\newcommand{\showindraft}[1]{}
\else
\newcommand{\showindraft}[1]{#1}
\fi

\newcommand{\jeff}[1]{\showindraft{\authnote{Jeff}{#1}}}
\newcommand{\jana}[1]{\showindraft{\authnote{Jana}{#1}}}

\newcommand{\mathcolorbox}[2]{\colorbox{#1}{\hspace{-1pt}\ensuremath{#2}\hspace{-1pt}}}

\newcommand{\Grn}[1]{\mathcolorbox{green!20}{#1}}

\newtoggle{HlInsane}
\iftoggle{HlInsane}{

}%else
{

}

\definecolor{dBlue}{rgb}{0.0, 0.0, 0.65}

\newcommand{\NominalAdapton}{\textsc{Nominal Adapton}\xspace}

\newcommand{\mergesym}{\mathrel{\union}}  % non-disjoint union
\newcommand{\joinsym}{\mathrel{*}}  % disjoint union

\newenvironment{bnfarray}{\begin{tabular}{c}\begin{array}[t]{@{}r@{~~}c@{~~}l@{~~}l@{~~}l}}{\end{array}\end{tabular}}
\newenvironment{xbnfarray}{\begin{tabular}{c}\begin{array}[t]{@{}r@{~~}c@{~}l@{~~}ll}}{\end{array}\end{tabular}}

\newcommand{\bnfheader}[1]{\multicolumn{4}{@{}l}{\text{#1}}}
\newcommand{\bnfas}{\mathrel{::=}}
\newcommand{\bnfalt}{\mathrel{\mid}}
\newcommand{\BnfaltBRK}{ \\ & & & \!\!\!\bnfalt}
\newcommand{\bnfaltBR}{ & & \!\!\!\bnfalt}
\newcommand{\bnfaltBRK}{ \\ \bnfaltBR}

\renewcommand{\mathrm}[1]{\ensuremath{#1}}

\mprset{andskip=3ex}

\begin{document}
\exclusivelicense

\conferenceinfo{OOPSLA '15}{October 25--30, 2015, Pittsburgh, PA, USA}
\copyrightyear{2015}
\copyrightdata{978-1-4503-3689-5/15/10}
\doi{2814270.2814305}

\newcommand{\mytitle}{Incremental Computation with Names}
\title{\mytitle}
\iftr
\subtitle{Extended Version}
\else\fi

\fboxsep1.5pt

\authorinfo{Matthew~A.~Hammer$^{1,2}$ \and
  Jana~Dunfield$^3$ \and
  Kyle~Headley$^{1,2}$ \and
  Nicholas~Labich$^2$ \and
  Jeffrey~S.~Foster$^2$ \and
  Michael~Hicks$^2$ \and
  David~Van~Horn$^2$}{
  \begin{minipage}{2in}
    \begin{center}
      $^1$ University of Colorado \\
      Boulder, USA
    \end{center}
  \end{minipage}
  \and
  \begin{minipage}{2in}
    \begin{center}
      $^2$ University of Maryland \\
      College Park, USA
    \end{center}
  \end{minipage}
  \and
  \begin{minipage}{2in}
    \begin{center}
      $^3$ University of British Columbia \\
      Vancouver, Canada
    \end{center}
  \end{minipage}%
\iftr
~\\
\vspace*{-22pt} %
\fi}
           {\relax}

\maketitle

%
% do not remove these comment lines
%
\begin{abstract}
  Over the past thirty years, there has been significant progress in
  developing general-purpose, language-based approaches to
  \emph{incremental computation}, which aims to efficiently update the
  result of a computation when an input is changed. A key design
  challenge in such approaches is how to provide efficient incremental
  support for a broad range of programs.  In this paper, we argue that
  first-class \emph{names} are a critical linguistic feature for
  efficient incremental computation. Names identify computations to be
  reused across differing runs of a program, and making them first
  class gives programmers a high level of control over reuse. We
  demonstrate the benefits of names by presenting \NominalAdapton, an
  ML-like language for incremental computation with names. We describe
  how to use \NominalAdapton to efficiently incrementalize several
  standard programming patterns---including maps, folds, and
  unfolds---and show how to build efficient, incremental probabilistic
  trees and tries. Since \NominalAdapton's implementation is subtle,
  we formalize it as a core calculus and prove it is
  \emph{from-scratch consistent}, meaning it always produces the same
  answer as simply re-running the computation. Finally, we demonstrate
  that \NominalAdapton can provide large speedups over both
  from-scratch computation and \Adapton, a previous state-of-the-art
  incremental computation system.
  %%
  %% \Adapton is a recently developed incremental computation (IC) system
  %% that aims to perform efficient recomputation given a small change to
  %% the input.  \Adapton's particular target is interactive programs, and in
  %% this setting it can provide orders of magnitude speedups compared to
  %% prior IC systems. However, in some cases \Adapton's performance
  %% suffers because it cannot determine that a previous function call is
  %% actually ``the same'' as a new function call and should be reused.
  %% To address this problem, we present \NominalAdapton, an extension to
  %% \Adapton that lets programmers use \emph{names} and
  %% \emph{namespaces} to directly control the process of matching a
  %% current call to a prior call. As a result, \NominalAdapton can
  %% greatly improve the run-time performance and memo table size
  %% compared to (classic) \Adapton in certain cases. We formalize
  %% \NominalAdapton and prove it is \emph{from-scratch consistent} for
  %% recomputation. We also describe an OCaml implementation of
  %% \NominalAdapton, which fixes a subtle bug with \Adapton's
  %% interaction with the GC\@. Performance experiments on standard IC
  %% benchmarks show that while \NominalAdapton introduces some overhead
  %% compared to \Adapton on simple programs that do not exhibit name
  %% mismatches, on more complicated programs \NominalAdapton performs
  %% better by up to 53$\times$.
\end{abstract}

% Local Variables:
% TeX-master: "alias"
% End:

{
}


\section{Introduction}
\Label{sec:intro}

\emph{Memoization} is a widely used technique to speed up running time
by caching and reusing prior results~\citep{Michie68}. The idea is
simple---the first time we call a pure function $f$ on immutable
inputs $\vec{x}$, we store the result $r$
in a \emph{memo table} mapping $\vec{x}$ to $r$. Then on subsequent
calls $f(\vec{y})$, we can return $r$ immediately if $\vec{x}$ and
$\vec{y}$ match. \emph{Incremental computation}
(IC)~\citep{PughThesis} takes this idea a step further, aiming to
reuse prior computations even if there is a small change in the
input. %
Recent forms of IC as exemplified by self-adjusting computation
(SAC)~\citep{AcarThesis} and \Adapton~\citep{Adapton2014} support 
\emph{mutable} inputs, meaning that two calls to $f(\vec{x})$
might produce different results because values reachable from the
same arguments $\vec{x}$ have been mutated. As such, before reusing a
memoized result $r$, any inconsistencies are repaired via a process
called \emph{change propagation}.

An important goal of an IC system is to minimize the work performed in
support of change propagation, and thus improve overall
performance. Matching---the task of determining whether a call's
arguments are ``the same'' as those of a memoized call---turns out to
play a key role, as we show in this paper.
The most common mechanism is \emph{structural matching}, which
traverses an input's structure to check whether each of its components match.
To make it fast,
implementations use variants of hash-consing (e.g.,
see~\cite{filliatre2006type}) which, in essence, associates with a
pointer a hash of its contents and compares pointers by hash.

Structural matching works well when memoizing pure computations over
immutable data because such computations always produce the same
result. But for IC involving mutable references, structural matching
can be too specific and therefore too fragile.
For example, suppose we map a function~\ml{f} over a mutable list \ml{input} ${=}$ $[0,1,3]$
producing \ml{output} $=$ \ml{map f input} $=$ $[f~0, f~1, f~3]$.
Next, suppose we mutate \ml{input}
by inserting \ml{2}, so \ml{input} becomes $[0,1,\mathbf{2},3]$.
Finally, suppose we recompute \ml{map f input}, now~$[f~0,f~1,\mathbf{f~2},f~3]$, attempting
to reuse as much of the prior computation as possible.
Structural matching will successfully identify and reuse the
recursive sub-call \ml{map f [3]}, reusing the result $[f~3]$. 
More generally, it will reuse
the mapped suffix after the inserted element, since the computation of
this output is independent of the mutated prefix. However, structural
matching will not match the sub-computations that map 0 to $f~0$ and 1
to $f~1$ because these sub-computations' outputs transitively include
the newly inserted value of $f~2$ (via their tail pointers). As a result, an IC system that uses
structural matching will rerun those sub-computations from scratch,
recomputing $f~0$ and $f~1$ and allocating new list cells to hold the
results. (Section~\ref{sec:overview} covers this example in detail.)

The key takeaway is that structural matching is too
conservative---it was designed for immutable inputs, in which case a
structural match produces a correct memoized result. But with IC using
mutable inputs, a match need not return a correct result; rather, our
aim should be to return a result that requires only a little work to
repair.  For our example, an ideal IC system
would be able to memoize the prefix, repairing it by mutating the old
output cell containing $f~1$ to insert $f~2$.

In this paper, we propose to overcome the deficiencies of structural
matching for IC by employing an alternative matching strategy that
involves \emph{names}. We implement our solution in
\NominalAdapton, an extension to the \Adapton IC
framework.
Our new \textit{nominal matching} strategy permits the programmer to
explicitly associate a name---a first-class (but abstract)
value---with a pointer such that pointers match when
their names are equal.
A program produces names by generating them
from existing names or other seed values. 
Returning to our example, we can add names to list cells such that an
output cell's name is derived from the corresponding input cell's
name. With this change, insertion into a list does not affect output
cells' names, and hence we can successfully reuse the computation of \ml{map} on the prefix before
the inserted element.
The particular naming strategy is explained in detail in  
Section~\ref{sec:overview}, which also gives an overview of the \NominalAdapton
programming model and shows how it improves performance on the
\ml{map} example compared to prior structural approaches.

Nominal matching is strictly more powerful than structural matching: The
programmer can choose names however they wish, including mimicking
structural matching by using hashing. 
That said, there is a risk that names could be used ambiguously,
associating the same name with distinct pointers.
In \NominalAdapton, if the programmer makes a mistake and uses a name ambiguously,
efficient run-time checks detect this misuse and raise an exception.
The use-once restriction can be limiting, however, so in 
addition to supporting first-class names, \NominalAdapton provides first-class
\emph{namespaces}---the same name can be reused as long as each use is
in a separate namespace. For example, we can safely map different
functions over the same list by wrapping the computations in
separate namespaces. 
Section~\ref{sec:patterns} illustrates several use cases of the \NominalAdapton
programming model, presenting naming design patterns for incremental
lists and trees and common computation patterns over them. We also
propose a fundamental data structure for probabilistically balanced
trees that works in a variety of applications.

We have formalized 
\NominalAdapton in a core calculus \calc and proved its incremental
recomputation is \emph{from-scratch consistent}, meaning it produces
the same answer as would a recomputation from scratch. As such,
mistakes from the programmer will never produce incorrect
results. Section~\ref{sec:corecalc} presents our formalism and
theorem.

We have implemented \NominalAdapton in OCaml as an extension to
\Adapton (Section~\ref{sec:implementation}).
We evaluated our implementation by comparing it to \Adapton on a set of
subject programs commonly evaluated in the IC literature, including
\cod{map}, \cod{filter}, \cod{reduce}, \cod{reverse}, \cod{median}, \cod{mergesort}, and \cod{quickhull}.
As a more involved example, we implemented an interpreter for an
imperative programming language (IMP with arrays), showing that
interpreted programs enjoy incrementality by virtue of using
\NominalAdapton as the meta-language.
Across our benchmarks, we find that \adapton{} is nearly always slower
than \NominalAdapton (sometimes orders of magnitude slower),
and is sometimes orders of magnitude slower than from-scratch
computation. By contrast, \NominalAdapton uniformly enjoys
speedups over from-scratch computation (up to $10900\times$) as well
as classic \adapton{} (up
to $21000\times$). (Section~\ref{sec:experiments} describes our
experiments.)

The idea of names has come up in prior incremental computation
systems, but only in an informal way. For example, \citet{AcarLW09}
includes a paragraph describing the idea of named references (there
called ``keys'') in the DeltaML implementation of SAC\@. To our
knowledge, our work is the first to formalize a notion of named
computations in IC and prove their usage correct. We are also the
first to empirically evaluate the costs and benefits of
programmer-named references and thunks. Finally, the notion of
first-class namespaces, with the same determinization benefits as
named thunks and references, is also new.  (Section~\ref{sec:related}
discusses SAC and other related work in more detail.)

\section{Overview}
\Label{sec:overview}

In this section we present \NominalAdapton and its programming model,
illustrating how names can be used to improve opportunities for
reuse. We start by introducing \Adapton's approach to incremental
computation, highlighting how \NominalAdapton extends its programming
model with support for names. Next we use an example, mapping over a
list, to show how names can be used to improve incremental performance.

\subsection{\Adapton and \NominalAdapton}

\Adapton aims to reuse prior computations as much as possible after a
change to the input. \Adapton achieves this by memoizing a
function call's arguments and results, reusing memoized results
when the arguments match (via structural matching).
In this section, we write \ml{memo}$(e)$ to
indicate that the programmer wishes $e$ to be memoized.\footnote{Programmers
actually have more flexibility thanks to \Adapton's support for
laziness, but laziness is orthogonal to names, which we focus on in this section.
We discuss laziness in Section \ref{sec:corecalc}.}

\Adapton provides mutable references: $\kw{ref}~e$ allocates a memory
location $p$ which it initializes to the result of evaluating $e$, and
$\kw{!}~p$ retrieves the contents of that cell. Changes to inputs
are expressed via reference cell mutations; \Adapton propagates
the effect of such changes to update previous results. Like
many approaches to incremental computation, \Adapton distinguishes
two layers of computation. Computations in the \emph{inner layer} are
incremental, but can only read and allocate references,
while computations in the \emph{outer layer} can change reference
values (necessitating change propagation for the affected inner-layer
computations) but are not themselves incremental.
This works by having the initial incremental run produce a
\emph{demanded computation graph} (DCG),
which stores values of memoized computations and tracks dependencies
between those computations and references.
Changes to mutable state ``dirty'' this graph, and change propagation
``cleans'' it, making its results consistent.

\subsection{Incremental Computation in \Adapton}

As a running example, consider incrementalizing a program that maps
over a list's elements. To support this, we define a list data structure that allows the
tail to be imperatively modified by the outer context:
\begin{OCaml}
    type 'a list = Nil | Cons of 'a * ('a list) ref
    let rec map f xs = memo( match xs with
      | Nil -> Nil
      | Cons (x, xs) -> Cons (f x, ref (map f !xs)) )
\end{OCaml}
This is a standard map function, except for two
twists: the function body is memoized via \cod{memo}, and the
input and output \cod{Cons} tails are reference cells.  The use of
\cod{memo} here records function calls to \cod{map}, identifying prior
calls using the function~\cod{f} and input list~\cod{xs}.  In turn,
\cod{xs} is either \cod{Nil} or is identified by a value of type \cod{'a} and
a reference cell.  Hence, reusing the identity of references is critical
to reusing calls to \cod{map} via \cod{memo}.
Now we can create a list (in the outer layer) and map over it (in the
inner layer):

\begin{OCaml}
    let l3 = Cons(3, ref Nil)  (* l3 = [3] *)
    let l1 = Cons(1, ref l3)   (* l1 = [1; 3] *)
    let l0 = Cons(0, ref l1)   (* l0 = [0; 1; 3] *)
    let m0 = (* inner start *) map f l0 (* inner end *)
        (* m0 = [f 0; f 1; f 3] *)
\end{OCaml}
Suppose we change the input to \cod{map} by
inserting an element:
\begin{OCaml}
    (tl l1) := Cons(2, ref l3)  (* l0 = [0; 1; 2; 3] *)
\end{OCaml}
Here, \cod{tl} returns the tail of its list argument.  After this change,
\cod{m0} will be updated to \cod{[f 0; f 1; f 2; f 3]}. In the best case,
computing \cod{m0} should only require applying \cod{f 2}
and inserting the result into the original output.  However,
\Adapton performs much more work for the above code. Specifically,
\Adapton will recompute \cod{f 0} and \cod{f 1}; if the change were
in the middle of a longer list, it would recompute the
\emph{entire} prefix of the list before the change. In contrast,
\NominalAdapton will only redo the minimal amount of work.

\begin{figure}[t]
   \begin{subfigure}{\columnwidth}
     \centering
    \includegraphics[width=2.3in]{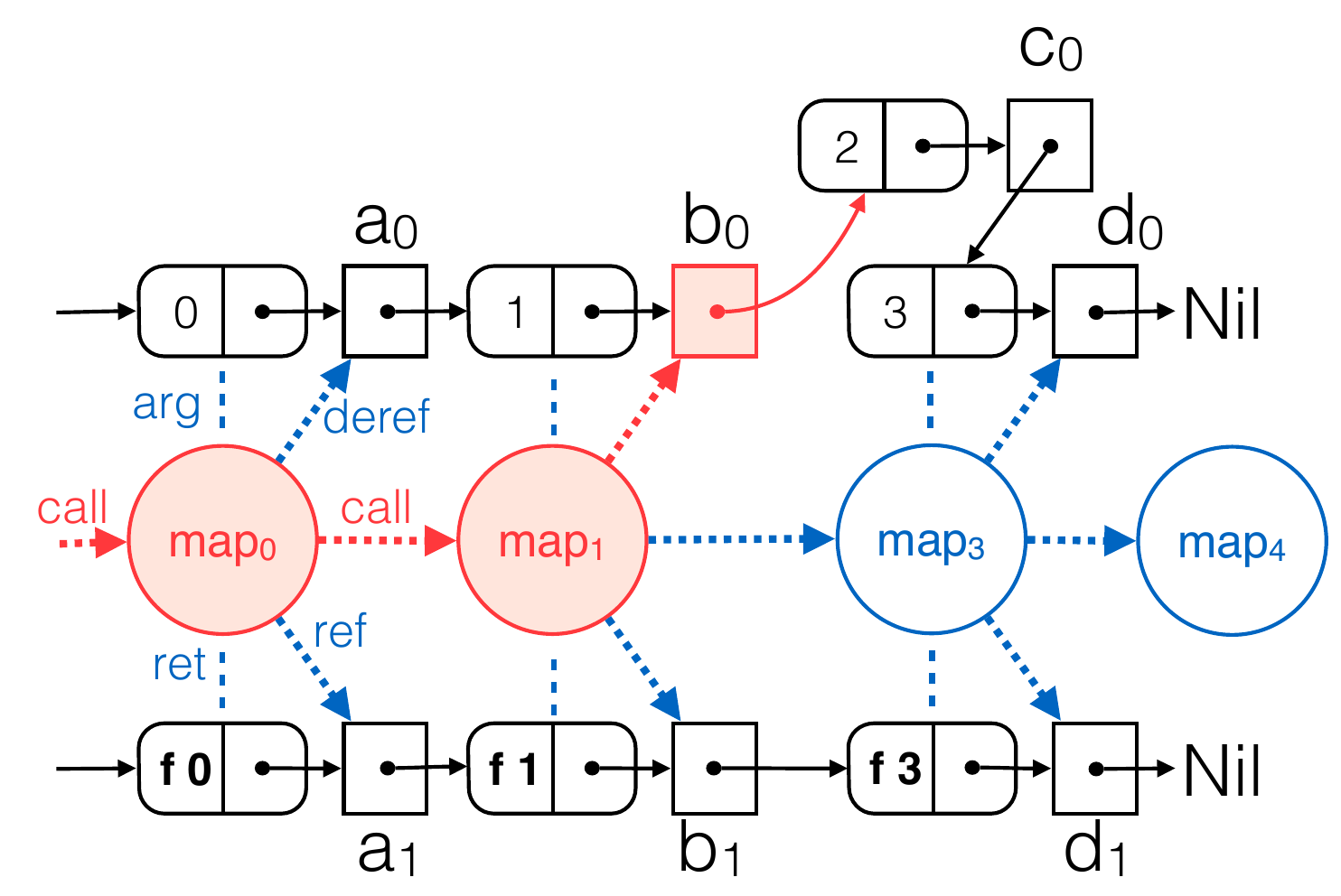}
     \caption{\Adapton, after insertion update}
     \label{fig:structbefore}
   \end{subfigure}

\medskip

   \begin{subfigure}{\columnwidth}
     \centering
    \includegraphics[width=2.9in]{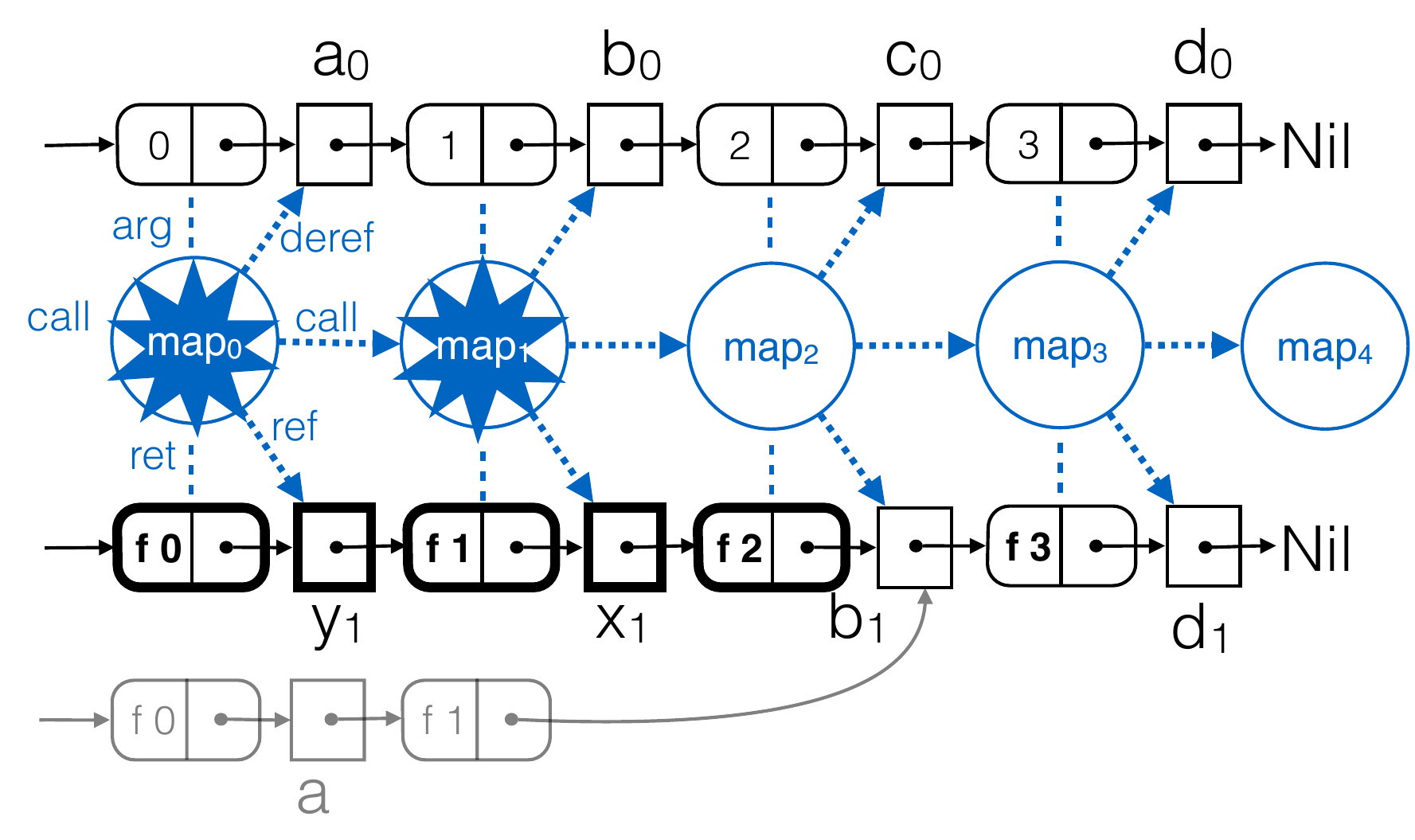}
     \caption{\Adapton, after change propagation}
     \label{fig:structafter}
  \end{subfigure}
  \caption{Incremental computation of \cod{map} in \Adapton}
  \label{fig:structexample}
\end{figure}

To understand why, consider
Figure~\ref{fig:structbefore}, which illustrates what
happens after the list update.  In this figure, the initial input and
output lists are shown in black at the top and the bottom of the
figure, respectively.
The middle of the figure shows the \emph{demanded computation graph}
(DCG), which records each recursive call of \cod{map} and its dynamic
dependencies. Here, nodes \cod{map}$_0$, \cod{map}$_1$, \cod{map}$_3$ and
\cod{map}$_4$ correspond to the four calls to \cod{map}. For
each call, the DCG records the arguments, the result, and the
computation's effects. Here, the effects are: dereferencing a pointer;
making a recursive call; and allocating a \cod{ref} cell in the
output list.  We label the arrows/lines of the first node only, to avoid
clutter; the same pattern holds for \cod{map}$_1$ and \cod{map}$_3$.

In the input and output, the tail of each \cod{Cons}~cell (a rounded
box) consists of a reference (a square box).  The input list is labeled
$\textsf{a}_0,\textsf{b}_0,\textsf{d}_0$, and the output list is labeled
$\textsf{a}_1,\textsf{b}_1,\textsf{d}_1$.
In \Adapton, structural matching determines the labels chosen by the
inner layer, and in particular, whenever the inner layer allocates a
reference cell to hold a value that is already contained in an existing
reference cell, it reuses this first cell and its label.

After the list is updated, \Adapton dirties all the computations that
transitively depend on the changed reference cell, $\textsf{b}_0$. The
dirtied elements are shaded in \textcolor{red}{red}. Dirtying is how
\Adapton knows that previously memoized results cannot be reused
safely without further processing.
\Adapton processes dirty parts of the DCG into clean computations
on demand, when they are demanded by the outer program.  To do so, it
either re-executes the dirty computations, or
verifies that they are unaffected by the original set of changes.

Figure~\ref{fig:structafter} shows the result of recomputing the output following
the insertion change, using this mechanism.  When
the outer program recomputes \cod{map f l0} (shown as \cod{map}$_0$),
\Adapton will clean (either recompute or reuse) the dirty nodes of the
DCG.
First, it re-executes computation \cod{map}$_1$, because it is the
first dirty computation to be affected by the changed reference
cell. We indicate re-executed computations with stars in their DCG
nodes.
Upon re-execution, \cod{map}$_1$ dereferences $\textsf{b}_0$ and
calls \cod{map} on the inserted \cod{Cons} cell holding \cod{2}.
This new call \cod{map}$_2$ calls \cod{f 2} (not shown),
dereferences $\textsf{c}_0$, calls \cod{map}$_3$'s computation
\cod{map f (Cons(3,d$_0$))}, allocates
$\textsf{b}_1$ to hold its result and returns \cod{Cons(f 3, d$_1$)}.

The recomputation of \cod{map}$_2$ exploits two instances of reuse. First, when
it calls $\cod{map f (Cons(3,d$_0$))}$, \Adapton reuses this portion of the DCG and
the result it computed in the first run.  \Adapton knows that the
prior result of $\cod{map f (Cons(3,d$_0$))}$ is unchanged because \cod{map}$_3$
is not dirty.  Notice that even if the tail of the list were much
longer, the prior computation of \cod{map}$_3$ could still be reused, since
the insertion does not affect it.

Second, when \cod{map}$_2$ allocates the reference cell to hold
the result of
$\cod{map f (Cons(3,d$_0$)}$, i.e.\ \cod{Cons(f 3, d$_1$)}, it reuses and shares
the existing reference cell \textsf{b}$_1$ that already holds this
content. Notice that this maintains our labeling invariant, so we
can continue to perform structural matching by comparing
labels.
As a side effect, it also improves performance by avoiding allocation
of an (isomorphic) copy of the output.\footnote{Note that the outer
  layer may cause two cells with the same contents to be labeled
  differently: If two cells are initially allocated with different
  contents, but then the outer layer mutates one cell to contain the
  contents of the other, the cells will still have different
  labels and hence will not match in \Adapton. This is a practical
  implementation choice, since otherwise \Adapton might
  need to do complicated heap operations to merge the identities of
  two cells. At worst, this choice causes \Adapton to miss some minor opportunities
  for reuse.
  However, this choice is consistent with the standard implementation of hash-consing,
  which only aims to share immutable structures~\citep{Allen1978,filliatre2006type}.
}

So far, \Adapton has successfully reused subcomputations, but
consider what happens next.
When the call \cod{map}$_2$ completes, \cod{map}$_1$
resumes control and allocates a reference cell to hold \cod{Cons(f 2,
  b$_1$)}.  Since no reference exists with this content, it
allocates a fresh reference $\textsf{x}_1$.  New references are shown
in bold.  The computation then returns \cod{Cons(f 1,
  x$_1$)}, which does not match its prior return value \cod{Cons(f 1,
  b$_1$)}.
Since this return value has changed from the prior run, \Adapton
re-runs \cod{map}$_1$'s caller, \cod{map}$_0$.  This consists of
re-running \cod{f 0}, reusing the (just recomputed, hence no longer dirty) computation
\cod{map}$_1$, and allocating a reference to hold its result
\cod{Cons(f 1, x$_1$)}.  Again, no reference exists yet with this
content (\cod{x$_1$} is a fresh tail pointer), so \Adapton allocates
a fresh cell~\cod{y$_1$}.
Finally, the call to \cod{map}$_0$ completes, returning a new list
prefix with the \emph{same content} (\cod{f 0} and \cod{f 1}) as in
the first run, but with new reference cell identities (\cod{y$_1$} and
\cod{x$_1$}).

In this example, small changes cascade into larger changes
because \Adapton identifies reference cells \emph{structurally}
based on their contents.  Thus, the entire prefix of
the output list before the insertion is reallocated and recomputed,
which is much more work than should be necessary.

\subsection{The Nominal Approach}
\label{sec:ovw-nominal}

We can solve the problems with structural matching by giving the
programmer explicit \emph{names} to control reuse. In this particular
case, we aim to avoid re-computing \cod{map}$_0$
and any
preceding computations.  In particular, we wish to re-compute only
\cod{map}$_1$ (since it reads a changed pointer)
and to compute
\cod{map}$_2$, the mapping for the inserted \cod{Cons} cell.

The first step is to augment mutable lists with names provided by
\NominalAdapton:
\begin{OCaml}
  type list = Nil | Cons of int * name * (list ref)
\end{OCaml}
Globally fresh names are generated either
non-\-deter\-min\-istic\-ally via \cod{new}, or from an
existing name via \cod{fork}. In particular, \cod{fork n} returns a
pair of distinct names based on the name~\cod{n} with the property
that it always returns the \emph{same} pair of names given the same
name \cod{n}.  In this way, the inner layer can
\emph{deterministically} generate additional names from a given one to
enable better reuse.
Finally, when the programmer allocates a reference cell, she explicitly
indicates which name to use, e.g.\ \cod{ref(n,1)}
instead of \cod{ref(1)}.

Now, when the list is created, the outer layer calls \cod{new} to
generate fresh, globally distinct names for each \cod{Cons} cell:
\begin{OCaml}
    let l3 = Cons(3, new, ref(new, Nil))  (* l3 = [3] *)
    let l1 = Cons(1, new, ref(new, l3 ))  (* l1 = [1; 3] *)
    let l0 = Cons(0, new, ref(new, l1 ))  (* l0 = [0; 1; 3] *)
\end{OCaml}

When the inner layer computes with the list, it uses the names in
each \cod{Cons} cell to indicate dependencies between the inputs and
outputs of the computation. In particular, we rewrite \cod{map} as
follows:
\clearpage
\begin{OCaml}
    let rec map f xs = memo( match xs with
      | Nil -> r
      | Cons(x, nm, xs) ->
          let nm1, nm2 = fork nm in
            Cons(f x, nm1, ref (nm2, map f (!xs))) )
\end{OCaml}
Unlike the outer program, which chooses reference names using \cod{new},
the inner program uses \cod{fork} to relate the
names and references in the output list to the names in the
input list.

\begin{figure}[t]
   \begin{subfigure}{\columnwidth}
     \centering
    \includegraphics[width=3in]{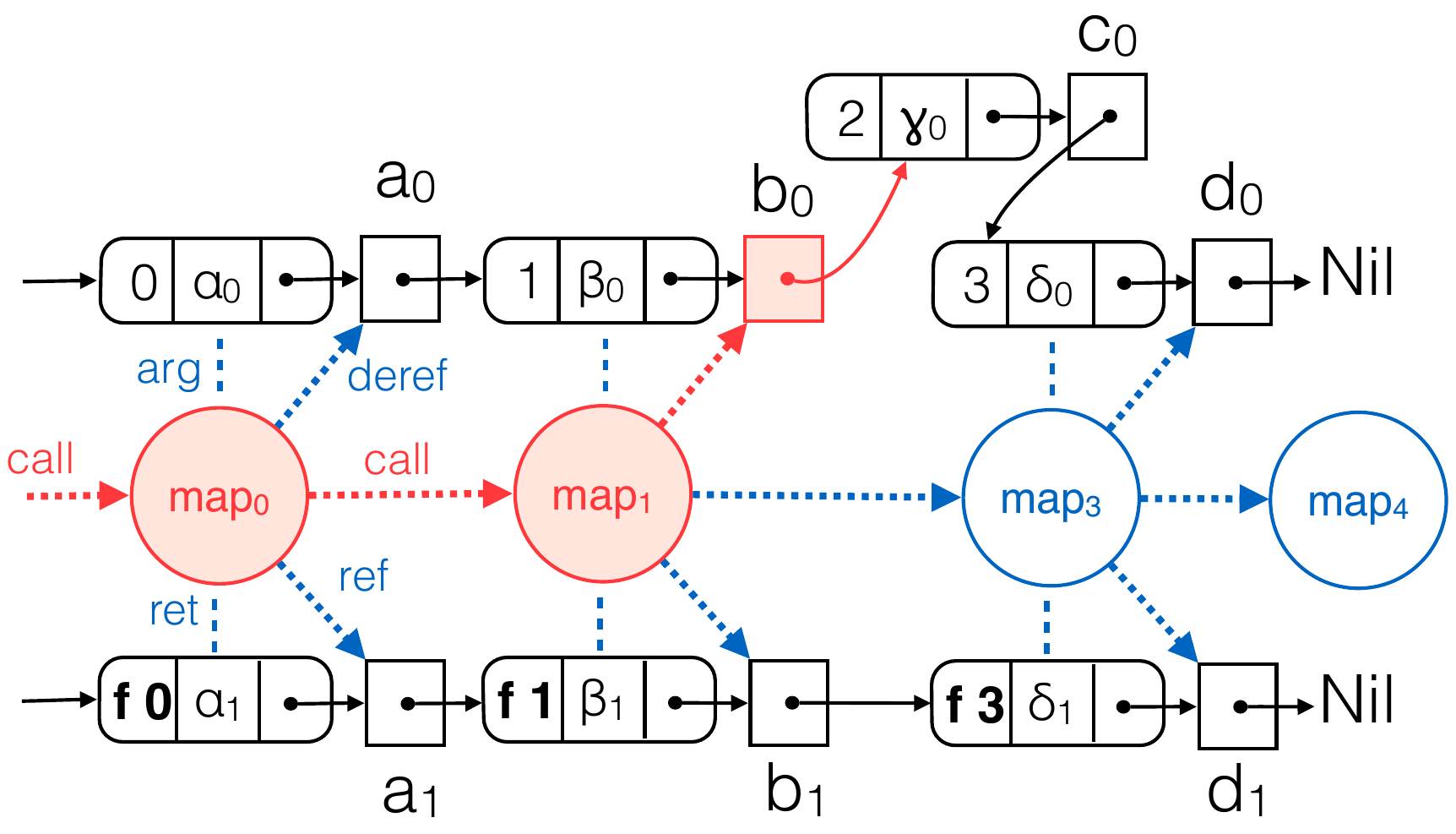}
     \caption{\NominalAdapton, after insertion update}
     \label{fig:nominalbefore}
  \end{subfigure}

\medskip

   \begin{subfigure}{\columnwidth}
     \centering
    \includegraphics[width=3.2in]{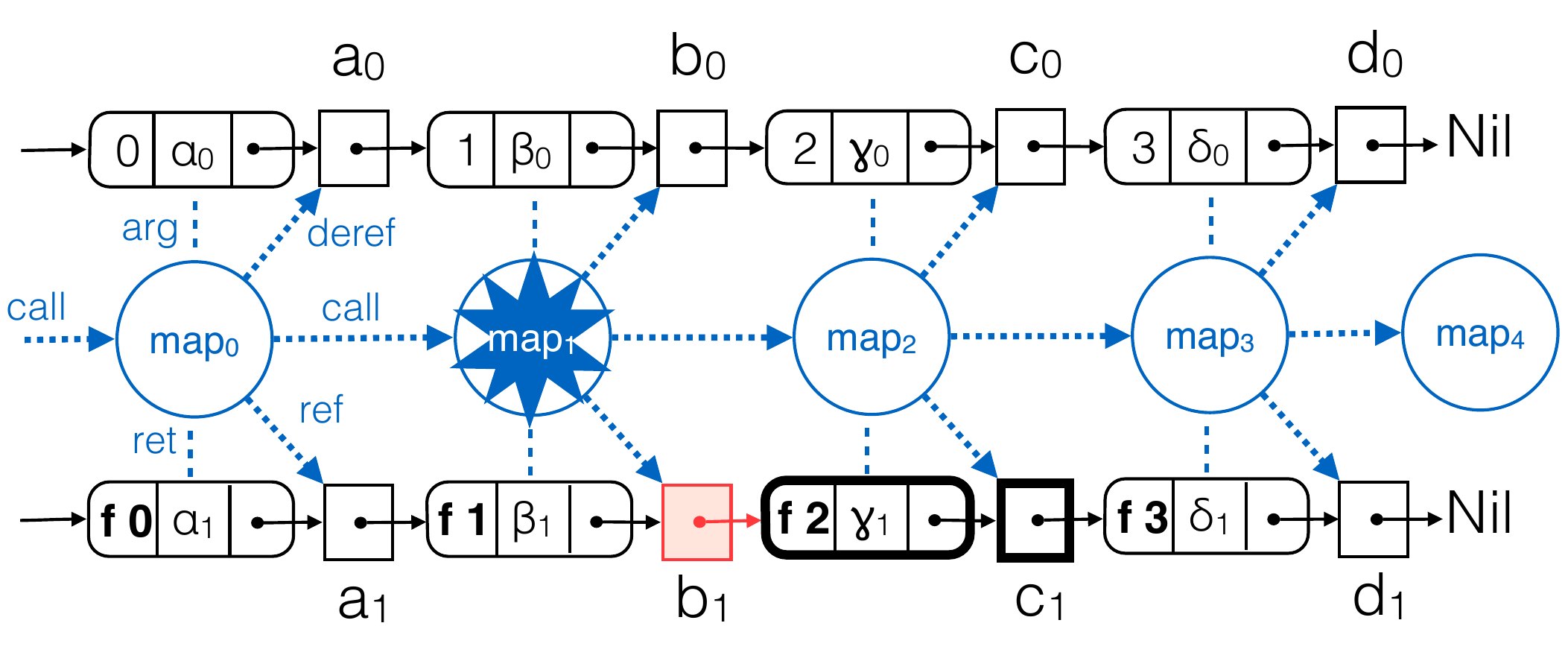}
    \vspace{-0.5em}
     \caption{\NominalAdapton, after change propagation}
     \label{fig:nominalafter}
  \end{subfigure}
  \caption{IC of \cod{map} with \NominalAdapton}
  \label{fig:overviewnominal}
\end{figure}

Now consider applying this function and making the same change as
above:
\begin{OCaml}
    let m0 = map f l0 in
    (tl l1) := Cons(2, new, ref l3)
\end{OCaml}
Figure~\ref{fig:overviewnominal} shows what happens.  The initial
picture in Figure~\ref{fig:nominalbefore} is similar to
the structural case in Figure~\ref{fig:structbefore}, except the input and output lists
additionally contain names $\alpha_0,\beta_0,\gamma_0$ and $\delta_0$.
The first part of the recomputation is the same: \NominalAdapton recomputes
\cod{map}$_1$, which reads the mutated reference~\cod{b}$_0$.
In turn, it recomputes \cod{map}$_2$, which reuses the call to \cod{map}$_3$
 to compute \cod{Cons(f 2, $\gamma_1$, c$_1$)}.
The recomputation of \cod{map}$_2$
returns a different value than in the prior run, with the new head
value \cod{f 2}.

At this point, the critical difference occurs.  Even though the result of
\cod{map}$_2$ is distinct from any list in the prior run, the call
\cod{map}$_1$ allocates the \emph{same} ref cell~\cod{b$_0$} as
before, because the name it uses for this allocation, $\beta_1$, is
the same as before. 
In the figure, \cod{fork $\beta_0$}~$\mapsto$ ($\beta_1$,$\beta_2$),
where $\beta_1$ becomes the name in the output list and
$\cod{b}_1 = \kw{ref}~\beta_2$ identifies the reference cell in its tail.
\NominalAdapton dirties the reference \cod{b}$_1$, to ensure that any
dependent computations will be cleaned before their results are reused. Due to this reuse, the
result of the call \cod{map}$_1$ is identical to its prior result: The
value of \cod{f 1} is unaffected, and the tail pointer \cod{b}$_1$ was
reused exactly.
Next, \NominalAdapton examines \cod{map}$_0$ and all prior calls. Because
the return value of \cod{map}$_1$ did not change, \NominalAdapton
simply marks the DCG node for its caller, \cod{map}$_0$, as clean; no
more re-evaluation occurs. This cleaning step breaks the cascade of
changes that occurred under \Adapton.  Prior computations are now
clean, because they only depend on clean nodes.

As a result of this difference in behavior, \NominalAdapton is able to
reuse all but two calls to function~\cod{f} for an insertion at any
index~$i$, while \Adapton will generally re-execute all $i-1$ calls to
function~\cod{f} that precede the inserted cell.
Moreover, \Adapton allocates a new copy of the output prefix (from 0
to $i$), while \NominalAdapton reuses all prior allocations.
Our experiments (\secref{experiments}) confirm that these differences
make \Adapton over 10$\times$ slower than \NominalAdapton, even for
medium-sized lists (10k elements) and cheap instances of~\cod{f}~(integer
arithmetic).

\subsection{Enforcing that Nominal Matching is Correct}
\label{sec:correctnaming}

Putting the task of naming in the programmer's hands can
significantly improve performance, but opens the possibility of
mistakes that lead to correctness problems. In particular, a
programmer could use the same name for two different objects:
\begin{OCaml}
    let y = ref n 1
    let z = ref n false
\end{OCaml}
Double use leads to problems since the variables \cod{y} and \cod{z} have
distinct types, yet they actually reference the same nominal object,
the number or boolean named~\cod{n}.
Consequently, the dereferenced values of~\cod{y} and~\cod{z} are sensitive to
the order of the allocations above, where the last allocation
``wins.''  This imperative behavior is undesirable because it is
inconsistent with our desired from-scratch semantics, where allocation
always constructs new objects.

To forbid double use errors,
our implementation uses an efficient dynamic check.
As the DCG evolves during program execution, \NominalAdapton maintains
a stack of DCG nodes, its \emph{force stack}, which consists of those
DCG nodes currently being forced (evaluated, re-evaluated, or reused).
When nominal matching re-associates a name with a different value or computation than a prior usage,
it overwrites information in the DCG, and
it dirties the old use and its transitive dependencies in the DCG.
To check that a nominal match is unambiguous, we exploit a key
invariant: A name is used ambiguously by a nominal match if and
  only if one or more DCG nodes on the force stack are dirtied when
said nominal match occurs.
If no such node exists, then the name is unambiguous.
\NominalAdapton implements this check %
by maintaining, for each DCG node, a bit
that is set and unset when the node is pushed and popped from
the force stack, respectively. This implementation is very efficient,
with $O(1)$ overhead.

Returning to the example above, the allocation of \cod{z} nominally
matches the allocation on the prior line, for \cod{y}.  Since \cod{1}
and \cod{false} are not equal, the nominal match dirties the DCG node
that allocates \cod{1}, which is also the ``current'' DCG node, and
thus is on the top of the force stack.  Hence, \NominalAdapton
raises an exception at the allocation of \cod{z}, indicating that
\cod{n} is used ambiguously. %

Note that because this check is dynamic and based on the DCG, it works
even when ambiguous name uses are separated across function or module
boundaries. This is important since, in our experience, most name
reuse errors are not nearly as localized as the example above. This
dynamic check was essential for our own development process; without
it, nominal mistakes were easy to make and nearly impossible to
diagnose.

\subsection{Namespaces}

Unfortunately, forbidding multiple uses of names altogether prevents
many reasonable coding patterns. For example, suppose we want to map
an input list twice: 
\begin{OCaml}
    let ys = map f input_list
    let zs = map g input_list
\end{OCaml}
Recall that in the \cod{Cons} case of \cod{map} we use each
name in the input list to create a corresponding name in the output
list. As such, the two calls to \cod{map} result in \cod{ys} and
\cod{zs} having cells with the same names, which is forbidden by our
dynamic check.

We can address this problem by creating distinct \emph{namespaces} for
the distinct functions (\cod{f} versus \cod{g}), where the same names
in two different namespaces are treated as distinct. A modified version of
\cod{map} using namespaces would be written thus:
\begin{OCaml}
  let map' n h xs = nest(ns(n), map h xs)
\end{OCaml}
The code~\cod{nest(s,e)} performs the nested computation \cod{e} in
namespace~\cod{s}, and the code \cod{ns(n)} creates a name\-space
from a given name \cod{n}. Just as with references,
we must be careful about how namespaces correspond across
different incremental runs, and thus we seed a namespace with a given
name. Now, distinct callers can safely call \cod{map'} with distinct
names: 
\begin{OCaml}
    let n1, n2 = (new, new) in
    let xs = map' n1 f input_list in
    let ys = map' n2 g input_list in
\end{OCaml}
The result is that each name in the input list is used only once per
namespace: Names in \cod{map f} will be associated with the first
namespace (named by \cod{n1}), and names in \cod{map g} will associate
with the second name\-space (named by \cod{n2}).

\paragraph{Section summary.}
The use of names allows the programmer to control (1) how
mutable reference names are chosen the first time, and (2) how to
selectively reuse and overwrite these references to account for
incremental input changes from the outer layer.
These names are transferred from input to output through the data
structures that they help identify (the input and output lists here),
by the programs that process them (such as \cod{map}).
Sometimes we want to use the same name more than once, in different
program contexts (e.g., \cod{map f $\cdot$} versus \cod{map g
  $\cdot$}); we distinguish these program contexts using
namespaces.

\section{Programming with Names}
\Label{sec:patterns}

While \NominalAdapton's names are a powerful tool for improving
incremental reuse, they create more work for the programmer.
In our experience so far, effective name reuse follows
easy-to-understand patterns. Section~\ref{sec:fpn} shows how to
augment standard data structures---lists and trees---and operations
over them---maps, folds, and unfolds---to incorporate names in a way
that supports effective reuse. Section~\ref{sec:sets}
describes \emph{probabilistic tries}, a nominal data structure we
developed that efficiently implements incremental maps and
sets. Finally, Section~\ref{sec:imp} describes our implementation of
an incremental IMP
interpreter that takes advantage of these
data structures to support incremental evaluation of its imperative
input programs. The benefits of these patterns and structures are
measured precisely in  Section~\ref{sec:experiments}.

\subsection{General Programming Patterns}
\Label{sec:fpn}

Practical functional programs use a wide variety of programming
patterns; three particularly popular ones are mapping, folding,
and unfolding. We consider them here in the context of lists and
trees.

\paragraph*{Mapping.}

Maps traverse a list (as in Section~\ref{sec:overview}) or tree and produce
an output structure that has a
one-to-one correspondence with the input structure.  
We have
already seen how to incrementalize mapping by associating a name with
each element of the input list and using \cod{fork} to derive a
corresponding name for each element in the output list, thereby
avoiding spurious recomputation of whole list prefixes
on a change.

\paragraph*{Folding.}
Folds traverse a list or tree and reduce subcomputations to provide a
final result. Examples are summing list elements or finding the
minimum element in a tree.

If we implement folding in a straightforward way in \NominalAdapton,
the resulting program tends to perform poorly.  The problem is that
every step in a list-based reduction uses an accumulator or result
that induces a global dependency on all prior steps---i.e., every step
depends on the \emph{entire prefix} or the \emph{entire suffix} of the
sequence, meaning that any change therein necessitates recomputing the
step. 

The solution is to use an approach from parallel programming: Use
trees to structure the input data, rather than lists, to permit
expressing independence between sub-problems.
Consider the following code, which
defines a type \cod{tree} for trees of integers:
\newcommand{\baltreeoflist}{\cod{baltree\_of\_list}\xspace}
\newcommand{\baltreerec}{\cod{baltree\_rec}\xspace}
\newcommand{\treemin}{\cod{tree\_min}\xspace}
\begin{OCaml}
  type tree = Leaf | Bin of name * int * (tree ref) * (tree ref)
\end{OCaml}
Like lists, these trees use \cod{ref}s to permit their recursive
structure to change incrementally, and each tree node includes a name.
We can reduce over this tree in standard functional style:
\begin{OCaml}
  let tree_min tree = memo(match tree with
    | Leaf         -> max_int
    | Bin(n,x,l,r) -> min3(x, tree_min(!l), tree_min(!r)))
\end{OCaml}
If we later update the tree and recompute, we can reuse subtree minimum computations,
because the names are stable in the tree.

Below, we show the original input tree alongside two 
illustrations (also depicted as trees) of which element of each
subtree is the minimum element, 
before and after the replacement of element~$1$ with the new
element~$9$:

\begin{center}
\vspace{1mm}
\begin{tabular}{|c|c|c|}
  Original tree
  &
  \parbox{0.3\columnwidth}{\centering Minimums\\(pre-change)}
  &
  \parbox{0.3\columnwidth}{\centering Minimums\\(post-change)}
  \\[4mm]
  \hline
  \begin{tikzpicture} %
    \tikzset{level distance=20pt,sibling distance=0pt}
    \tikzset{grow'=up} %
    \tikzset{every tree node/.style={anchor=base west}} %
    \Tree [.$3$ [.$2$ [.$1$ ] [.$4$ ] ] [.$5$ {$\cdot$} [.$6$ ] ] ]
\end{tikzpicture}
&
\begin{tikzpicture} %
    \tikzset{level distance=20pt,sibling distance=0pt}
    \tikzset{grow'=up} %
    \tikzset{every tree node/.style={anchor=base west}} %
    \Tree [.$1$ [.$1$ [.$1$ ] [.$4$ ] ] [.$5$ {$\cdot$} [.$6$ ] ] ]
\end{tikzpicture}
&
\begin{tikzpicture} %
    \tikzset{level distance=20pt,sibling distance=0pt}
    \tikzset{grow'=up} %
    \tikzset{every tree node/.style={anchor=base west}} %
    \Tree [.$2$ [.$2$ [.$9$ ] [.$4$ ] ] [.$5$ {$\cdot$} [.$6$ ] ] ]
\end{tikzpicture}
\end{tabular}
\end{center}

Notice that while the first element~$1$ changed to~$9$, this only
affects the minimum result along one path in the tree: the path from
the root to the changed element.
In contrast, if we folded the sequence naively as a list, all the
intermediate computations of the minimum could be affected by a change
to the first element (or last element, depending on the fold
direction).
By contrast, the balanced tree structure (with expected logarithmic
depth) overcomes this problem by better isolating independent
subcomputations.  

Pleasingly, as first shown by \cite{PughTe89}, we can efficiently
build a tree probabilistically from an input list, and thus transfer
the benefits of incremental tree reuse to list-style
computations. Building such a tree is an example of unfolding,
described next.

\paragraph{Unfolding.}

Unfolds iteratively generate lists or trees using a ``step function''
with internal state. As just mentioned, one example is building a
balanced tree from a list. 
Unfortunately, if we implement unfolding in a straightforward way,
incremental computation suffers. In particular, we want similar
lists (related by small edits) to lead to similar trees, with many
common subtrees; meanwhile, textbook algorithms for building balanced
trees, such as AVL trees and splay trees, are too sensitive to changes
to individual list elements.

\newcommand{\ptr}{%
  \mathrel{\bullet\!\!\!\to}%
}

The solution to this problem is to construct a \emph{probabilistically
  balanced tree}, with expected $O(\log n)$ height for input list $n$.
The height of each element in the resulting tree is determined by a
function that counts the number of trailing
zero bits in a hash of the given integer.
For example, given input
elements $[a,b,c,d,e,f]$ with heights $[0,1,0,2,1,0]$, respectively,
then our tree-construction function will produce the binary
tree shown below:

\begin{center}
\vspace{1mm}
\begin{tikzpicture} %
    \tikzset{level distance=20pt,sibling distance=2pt}
    \tikzset{grow'=up} %
    \tikzset{every tree node/.style={anchor=base west}} %
    \Tree [.$d$ [.$b$ [.$a$ ] [.$c$ ] ] [.$e$ {$\cdot$} [.$f$ ] ] ]
\end{tikzpicture}
\end{center}

\citet{PughTe89} showed that this procedure induces a
probabilistically balanced tree, with similar lists inducing
similar trees, as desired.
Further, each distinct list of elements maps to exactly one tree
structure.  This property is useful in \NominalAdapton, since a
canonical structure is more likely to be reused than one that can
exhibit more structural variation. While past work has
  considered incremental computations over such balanced trees, in
  this work we find that the \emph{construction} of the tree from a
  mutable, changing sequence can also be efficiently incrementalized
  (Pugh's work focused only on a \emph{pure} outer program).

\subsection{Probabilistic Tries}
\Label{sec:sets}
\Label{sec:maps}

Inspired by probabilistic trees, we developed efficient, incremental
\emph{probabilistic tries}, which
use a different naming pattern in which certain names are \emph{external} to
the data structure.

We define tries as binary trees whose nodes hold a name and two
children (in reference cells), and whose leaves store data.  Here we
use integers for simplicity, but in general nodes would hold arbitrary
data (e.g., for maps, they would hold key-value pairs):
\begin{OCaml}
    type trie = Nil
                 | Leaf of int
                 | Bin of name * (trie ref) * (trie ref)
\end{OCaml}
The key idea of a probabilistic trie is to use a bit string to
identify two things at once: the element stored in the structure (via
its hash) and the path to retrieve that element, if it is present.  To
keep it simple, the code below assumes that all data elements have a
unique hash, and that the input trie is complete, meaning that all
paths are defined and either terminate in a \cod{Nil} or a
\cod{Leaf}. Our actual implementation of tries makes
neither assumption.

The first operation of a trie is \cod{find}, which returns either
\cod{Some} data element or \cod{None}, depending on whether data with
the given hash (a list of \cod{bool}s) is present in the trie.
\begin{OCaml}
    val find : trie -> bool list -> int option
    let rec find trie bits =
      match bits, trie with
      | [], Leaf(x) -> Some x
      | [], Nil     -> None
      | true::bits,  Bin(_,left,right) -> find (!left)  bits
      | false::bits, Bin(_,left,right) -> find (!right) bits
\end{OCaml}
The other operation on tries is \cod{extend n t b d} which, given an
input trie \cod{t}, a data element \cod{d} and its hash \cod{b},
produces a new trie with \cod{d} added to it:
\begin{OCaml}
   val extend : name -> trie -> bool list -> int -> trie
   let rec extend nm trie bits data = memo (
     match bits, trie with
     | [],           (Leaf _ | Nil)    -> Leaf(x)
     | false::bits, Bin(_,left,right) -> (* symmetric to below *)
     | true::bits,  Bin(_,left,right) ->
          let n1, n2, n3, n4 = fork4 nm in
          let right' = extend n1 (!right) bits data in
            Bin(n2, ref(n3, !left), ref(n4, right'))  )
\end{OCaml}

Critically, the first argument to \cod{extend} is an externally
provided \cod{nm} that is used to derive names (using \cod{fork})
for each \cod{ref} in the new path.
Thus, the identity of the trie returned by \cod{extend n t b d} only
depends directly on the name~\cod{n} and the inserted data, and
\emph{not} on the names or other content of the input trie~\cod{t}.

Any incremental program that sequences multiple trie extensions makes
critical use of this independence (e.g., the interpreter discussed
below).  To see how, consider two incremental runs of such a program
with two similar sequences of extensions, $[1,2,3,4,5,6]$ versus
$[2,3,4,5,6]$, with the same sequence of five names for common
elements~$[2,\ldots,6]$.  Using the name-based extension above, the
tries in both runs will use exactly the same reference cells. By
contrast, the structural approach will build entirely new tries in the
second run, since the second sequence is missing the leading~$1$ (a
different initial structure).  Similarly, using the names \emph{in the
  trie} to extend it will also fail here, since it will effectively
identify each extension by a global count, which in this case shifts
by one for every extension in the second run.
Using external names for each extension overcomes both of these
problematic behaviors, and gives maximal reuse.

\subsection{Interpreter for \textsc{IMP}}
\Label{sec:imp}
\Label{sec:interpreter}

Finally, as a challenge problem, we used \NominalAdapton to build an
interpreter for \textsc{IMP}~\citep{winskel}, a simple imperative
programming language. Since our
interpreter is incremental, we can efficiently recompute an
interpreted program's output after a change to the program code
itself. While Adapton requires incrementalized computations to be fully functional, implementing a purely functional interpreter for IMP allows the imperative object language to inherit the incrementality of the meta-language.

The core of our interpreter is a simple, big-step \cod{eval} function
that recursively evaluates an IMP command in some store (that is, a heap)
and environment, returning the final store and environment.

Commands include while loops, sequences, conditionals, assignments (of
arithmetic and boolean expressions) to variables, and array operations
(allocation, reading, and writing).  All program values are either
booleans or integers.  As in C, an integer also doubles as a
``pointer'' to the store.
The interpreter uses finite maps to represent its environment of
type~\cod{env} (a mapping from variables to \cod{int}s) and its
\cod{store} of array content (a mapping from \cod{int}s to
\cod{int}s).

For incremental efficiency, the interpreter makes critical use of the
programming patterns we have seen so far.
We use probabilistic tries to represent the finite maps for stores and
environments. Each variable assignment or array update in \textsc{IMP}
is implemented as a call to \cod{extend} for the appropriate trie.
Names have two uses inside the interpreter. Each call to
\cod{extend} requires a name, as does each recursive call to
\cod{eval}. Classic \adapton{} would identify each with the full
structural content of its input. For the IMP language, we require far
less information to disambiguate one program state from another.

Consider first the IMP language without while loops. Each subcommand
is interpreted at most once; as such, each program state and each
path created in the environment or store can be identified with the
particular program position of the subcommand being interpreted.

With the addition of while loops, we may interpret a single program
position multiple times. To disambiguate these, we thread loop counts
through the interpreter represented as a list of integers.
The loop count $[3, 4]$, for example, tells us we are inside two
while loops: the third iteration through the inner loop, within the
fourth iteration through the outer loop. We extend the
\NominalAdapton API to allow creating names not just by forking, but
also by adding in a count---here, names are created using the loop
count paired with the name of the program position. As such, the names
sufficiently distinguish recursive calls to the interpreter and paths
inside the environment and store.

\smallskip 

\begin{OCaml}
  val eval : nm * int list * store * env * cmd -> store * env
\end{OCaml}

In our experiments, we show that the combination of this naming
strategy and probabalistically balanced tries yields an efficient
incremental interpreter. Moreover, the interpreter's design provides
evidence that \NominalAdapton's programming patterns
are \emph{compositional}, allowing us to separately choose
how to use names for different parts of the program.

\section{Formal Development}
\Label{sec:corecalc}

The interaction between memoization, names, and the demanded computation
graph is subtle.  For this reason, we have distilled \NominalAdapton
to a core calculus called \calc, which represents the essence of the
\NominalAdapton implementation and formalizes the key algorithms
behind incremental computation with names.  We prove the fundamental
correctness property of incremental computation, dubbed
\emph{from-scratch consistency}, which states that incrementally
evaluating a program produces the same answer as re-evaluating it
from scratch.  This theorem also establishes that (mis)use of names
cannot interfere with the meaning of programs; while a poor use of
names may have negative impact on performance, it will not cause a
program to compute the wrong result.

We present the full theory but only sketch the consistency result here;
the details and proofs are in the extended version of this paper~\citep{arxiv-version}.

The formal semantics presented here differs considerably from that of
the original \adapton system~\citep{Adapton2014}.
In particular, the original semantics modeled the DCG as an idealized
cache of tree-shaped traces that remember all input--output
relationships, forever.
In contrast, the one developed here models the DCG as it is actually
implemented in both \adapton and \NominalAdapton, as a changing graph
whose nodes are memoized refs and memoized computations (thunks),
and whose edges are their dependencies.
Further, via names, this semantics permits more kinds of DCG mutation,
since the computations and values of a node can be overwritten via
nominal matching.
In summary, compared to the prior semantics, our theory more accurately
models the implementation, and our metatheory proves correctness in a
more expressive setting.

\begin{figure}[thbp]
~\hspace{-2ex}
{
\small
\begin{bnfarray}
\text{Pointers} &
p, q
      & \bnfas
      &
          k  @  \omega 
            \bnfalt
          \txtsf{root} 
       \\
\text{Names}  &
 k
      & \bnfas
      &
           \bullet 
              \bnfalt
           { k }{\cdot}{ \ottsym{1} } 
              \bnfalt
           { k }{\cdot}{ \ottsym{2} } 
       \\
\text{Namespaces}  &
\omega, \mu
      & \bnfas
      &
         \top
                \bnfalt
          \omega . k 
\\
\text{Values} &
v
      & \bnfas
      &
          x  %
              \bnfalt
          \ottsym{(}  v_{{\mathrm{1}}}  \ottsym{,}  v_{{\mathrm{2}}}  \ottsym{)}  %
    \BnfaltBRK
           \ottkw{inj} _{ i }~ v   %
    \BnfaltBRK
           \Grn{ \ottkw{nm} \, k }   %
          & \!\!\!\!\!\!\!\!\!\text{name}
    \BnfaltBRK
           \ottkw{ref} \,\Grn{ p }   %
          & \!\!\!\!\!\!\!\!\!\text{pointer to value (reference cell)}
    \BnfaltBRK
           \ottkw{thk} \,\Grn{ p }   %
          & \!\!\!\!\!\!\!\!\!\text{pointer to computation}
    \BnfaltBRK
           \Grn{\textbf{ns}\, \omega } 
          & \!\!\!\!\!\!\!\!\!\text{namespace identifier}
\end{bnfarray}

\begin{xbnfarray}
\bnfheader{Results} \\
 \mathrm{t}& \bnfas
     &
       \lambda  x  \ottsym{.}  e
           \bnfalt
       \ottkw{ret} \, v
\\[1ex]
\bnfheader{Computations}\\
e
     & \bnfas
     & \mathrm{t}
       \bnfalt
       e \, v
    \bnfalt
      f %
    \bnfalt
      \ottkw{fix} \, f  \ottsym{.}  e %
   \bnfalt
       \textbf{let}\, x \,{\leftarrow}\, e_{{\mathrm{1}}} \, \ottkw{in} \, e_{{\mathrm{2}}}  %
   \hspace{-18ex}~
   \bnfaltBRK
       \ottkw{case} \, \ottsym{(}  v  \ottsym{,}  x_{{\mathrm{1}}}  \ottsym{.}  e_{{\mathrm{1}}}  \ottsym{,}  x_{{\mathrm{2}}}  \ottsym{.}  e_{{\mathrm{2}}}  \ottsym{)}
       \bnfalt
       \ottkw{split} \, \ottsym{(}  v  \ottsym{,}  x_{{\mathrm{1}}}  \ottsym{.}  x_{{\mathrm{2}}}  \ottsym{.}  e  \ottsym{)}
   \hspace{-18ex}~
   \bnfaltBRK
       \textbf{thunk}(\Grn{ v }, e )  %
      & \text{create thunk $e$ at name $v$}
   \bnfaltBRK
      \ottkw{force} \, \ottsym{(}  v  \ottsym{)} %
      & \text{force thunk that $v$ points to}
   \bnfaltBRK
       \Grn{\textbf{fork}( v )}  %
      & \text{fork name $v$ into two halves $v{\cdot}1$, $v{\cdot}2$}
   \bnfaltBRK
       \textbf{ref}(\Grn{ v_{{\mathrm{1}}} }, v_{{\mathrm{2}}} )  %
       & \text{allocate $v_2$ at name $v_1$}
   \bnfaltBRK
      \ottkw{get} \, \ottsym{(}  v  \ottsym{)} %
          & \text{get value stored at pointer $v$}
   \bnfaltBRK
       \Grn{\textbf{ns}\,( v , x . e )}  %
           & \text{bind $x$ to a new namespace $v$}
   \bnfaltBRK
       \Grn{\textbf{nest}( v , e_{{\mathrm{1}}} , x . e_{{\mathrm{2}}} )} \!\!\!\! %
          & \text{evaluate $e$ in namespace $v$} \\
&&& \text{~~and bind $x$ to the result}
\end{xbnfarray}
}
\caption{Syntax}
\FLabel{fig:syntax}
\end{figure}

\begin{figure}[thbp]
{
\small
\begin{xbnfarray}
\bnfheader{Graphs} \\
 G  \ottsym{,}  H\! &
    \bnfas &
    \varepsilon
        & \text{empty graph}
        \bnfaltBRK
    G  \ottsym{,}   p {:} v 
        & \text{$p$ points to value $v$}
        \bnfaltBRK
    G  \ottsym{,}   p {:} e 
        & \text{$p$ points to thunk $e$ (no cached result)}
        \bnfaltBRK
    G  \ottsym{,}   p {:}( e , \mathrm{t} ) 
        & \text{$p$ points to thunk $e$ with cached result $t$}
       \bnfaltBRK
    G  \ottsym{,}  \ottsym{(}  p  \ottsym{,}  a  \ottsym{,}  b  \ottsym{,}  q  \ottsym{)}
        & \text{\small$p$ depends on $q$ due to action $a$,}
        \\ & & & \text{~~with status $b$}
\end{xbnfarray}

\begin{bnfarray}
\bnfheader{Edge actions} \\
a &
   \bnfas
   && \text{For edge $\ottsym{(}  p  \ottsym{,}  a  \ottsym{,}  b  \ottsym{,}  q  \ottsym{)}$, the computation at $p$\dots}
   \\
   &&
       \ottkw{alloc} \, v
          & \text{\dots created reference $v$ at $q$}
   \bnfaltBRK
      \ottkw{alloc} \, e
          & \text{\dots created thunk $e$ at $q$}
   \bnfaltBRK
      \ottkw{obs} \, v
          & \text{\dots read $q$'s value, which was $v$}
   \bnfaltBRK
      \ottkw{obs} \, \mathrm{t}
          & \text{\dots forced thunk $q$, which returned $t$}
\\[1ex]
\bnfheader{Edge statuses} \\
b &
   \bnfas &
       \txtsf{clean} 
      & \text{value or computation at sink is out of date}
   \bnfaltBRK
       \txtsf{dirty} 
      & \text{value or computation at sink is up to date}
\end{bnfarray}
}
\caption{Graphs}
\FLabel{fig:graph}
\end{figure}

\subsection{Syntax}

The syntax of \calc is defined in the style of
the call-by-push-value (CBPV) calculus~\citep{Levy99subsuming}, a
standard variant of the lambda calculus with an explicit thunk
mechanism. \Figureref{fig:syntax} gives the syntax of the language.
The non-highlighted features are standard, and
the \Grn{\text{highlighted}} forms are new in \calc.

\NominalAdapton follows \Adapton in supporting \emph{demand-driven}
incremental computation using a lazy programming model.  In \Adapton,
programmers can write \kw{thunk}$(e)$ to create a suspended
computation, or \emph{thunk}.
The thunk $v$'s value is computed only
when it is \emph{forced}, using syntax \kw{force}$(v)$.  Thunks also
serve as \Adapton's (and \NominalAdapton's) unit of incremental reuse:
if we want to reuse a computation, we must make a thunk out of it.  
The syntax
\cod{\kw{memo}($e$)} we used earlier is shorthand for
\cod{\kw{force}(\kw{thunk}($e, e$))}, where we abuse notation and treat $e$~as a name that identifies itself.\footnote{Notice that \NominalAdapton
thunks are also named, which provides greater control over reuse.}  This construction introduces a
\kw{thunk} that the program immediately \kw{force}s, eliminating
laziness, but supporting memoization.

CBPV distinguishes values, results (or \emph{terminal computations}),
and computations.  A computation $e$ can be turned into a value by
thunking it via $ \textbf{thunk}(\Grn{ v }, e ) $.  The first argument, the name $v$, is particular to
\NominalAdapton; ordinary CBPV does not explicitly name thunks.
Conversely, a value $v$ can play the role of a result $t$ via
$\ottkw{ret} \, v$; results $t$ are a subclass of computations $e$.

Functions $\lambda  x  \ottsym{.}  e$ are terminal computations; $e \, v$ evaluates
$e$ to a function $\lambda  x  \ottsym{.}  e'$ and substitutes $v$ for $e'$.  Note that the
function argument in $e \, v$ is a value, not a computation.

Let-expressions $ \textbf{let}\, x \,{\leftarrow}\, e_{{\mathrm{1}}} \, \ottkw{in} \, e_{{\mathrm{2}}} $ evaluate the computation $e_1$ first.
  The usual
  $\lambda$-calculus application $e_1\,e_2$ can be simulated by
  $ \textbf{let}\, x \,{\leftarrow}\, e_{{\mathrm{2}}} \, \ottkw{in} \, \ottsym{(}  e_{{\mathrm{1}}} \,  x   \ottsym{)} $.
Fixed points $\ottkw{fix} \, f  \ottsym{.}  e$ are computations, and so are fixed point variables $f$.

Given an injection into a disjoint union, $ \ottkw{inj} _{ i }~ v $,
  the \textbf{case} computation form eliminates the sum and computes
  the corresponding $e_i$ branch, with $v$ substituted for $x_i$.
Given a pair $(v_1, v_2)$, the computation $\ottkw{split} \, \ottsym{(}  v  \ottsym{,}  x_{{\mathrm{1}}}  \ottsym{.}  x_{{\mathrm{2}}}  \ottsym{.}  e  \ottsym{)}$
computes $e$, first substituting $v_1$ for $x_1$, and $v_2$ for $x_2$.

For more on (non-nominal, non-incremental) formulations of
CBPV, including discussion of \emph{value types} and
\emph{computation types}, see \citet{Levy99subsuming,LevyThesis}.

\paragraph{Graphs, Pointers, and Names.}
Graphs $G$ are defined in \Figureref{fig:graph}.
They represent the mutable
store (references), memo tables (which cache thunk results), and the
DCG\@.  Element $ p {:} v $ says that pointer $p$'s
current value is $v$. Element $ p {:} e $ says that pointer
$p$ is the name of thunk $e$. Element $ p {:}( e , \mathrm{t} ) $ says
that $p$ is the name of thunk $e$ with a previously computed result
$\mathrm{t}$ attached.  Element $\ottsym{(}  p  \ottsym{,}  a  \ottsym{,}  b  \ottsym{,}  q  \ottsym{)}$ is a DCG edge indicating that the thunk
pointed to by $p$ depends on node $q$, where $q$ could
either name another thunk or a reference cell.  Dependency edges also
reflect the \emph{action} $a$ that produced the edge and the edge's
\emph{status} $b$ (whether it is clean or dirty).

Pointers in \calc are represented as pairs $ k  @  \omega $,\footnote{The
  pointer $ \txtsf{root} $ is needed to represent the top-level ``thunk'' in
  the semantics, but will never be mapped to an actual value or
  expression.} where $ \Grn{ \ottkw{nm} \, k } $ was the name given as the
first argument in a call to \textbf{thunk} or \textbf{ref},
and $\omega$ was the namespace in which the call to \textbf{thunk}
or \textbf{ref} took place.
This namespace $\omega$ is either $\top$ for the top-level,
or some $ \mu . k_{{\mathrm{1}}} $
as set by a call to \textbf{nest}.
For the latter case, the program would first construct a
value $ \Grn{\textbf{ns}\,  \mu . k_{{\mathrm{1}}}  } $ by calling \textbf{ns} with
first argument $ \Grn{ \ottkw{nm} \, k_{{\mathrm{1}}} } $ while in namespace $\mu$.
Notice that
pointers and namespaces have similar structure, and similar assurances
of determinism
when creating named thunks, references,
and namespaces. Finally, names $k$ consist of the root name
$ \bullet $, while other names are created by ``forking'' existing
names: invoking
$ \Grn{\textbf{fork}(  \Grn{ \ottkw{nm} \, k }  )} $ produces names $ \Grn{ \ottkw{nm} \,   { k }{\cdot}{ \ottsym{1} }   } $ and %
$ \Grn{ \ottkw{nm} \,   { k }{\cdot}{ \ottsym{2} }   } $.

\newcommand{\nonincrcolorbox}[1]{%
      \colorbox{gray!10}{%
           \ensuremath{#1}\!\!%
      }%
}

\begin{figure*}[p]

\begin{drulepar}{$ G_{{\mathrm{1}}}   \vdash ^{ p }_{ \omega }  e   \Downarrow   G_{{\mathrm{2}}}  ;  \mathrm{t} $}{\ottcom{Under graph
      $G_{{\mathrm{1}}}$, evaluating expression $e$ as part of thunk
      $p$ in namespace $\omega$ yields $G_{{\mathrm{2}}}$ and $\mathrm{t}$.}}
~
\vspace{-3ex} %
\vspace{-0ex} %
~\\
\text{Rules \textbf{common} to the non-incremental and incremental systems:}\hfill
~\\
\runonmathsz{9pt}
~\!\!
    \ottdruleEvalXXterm{}
    \and
    \ottdruleEvalXXapp{}
    \vspace{-0.5ex}
    \\
    \ottdruleEvalXXfix{}
    \and
    \ottdruleEvalXXbind{}
\vspace{-0.5ex}
\\
\ottdruleEvalXXcase{}
\and
\ottdruleEvalXXsplit{}
\vspace{-2.2ex}
\\
\ottdruleEvalXXfork{}
\vspace{-0.5ex}
\\
\ottdruleEvalXXnamespace{}
\and
\ottdruleEvalXXnest{}
\end{drulepar}

\vspace{-9.0ex}
~\\

$~$\hrule$~$ \\[-2.5ex]
$~$\hrule$~$

\vspace{-2.5ex}

\begin{mathpar}
\text{Rules \textbf{specific} to the (\nonincrcolorbox{\!\text{non-incremental}\;} $||$ incremental) systems:}\hfill
\vspace{-0.2ex}
\\
\runonmathsz{9pt}
~\hspace{-1.0ex}
\begin{array}{l||l}
  \arrayenvc{\nonincrcolorbox{\ottdruleEvalXXrefPlain{}}}
  &
  \arrayenvc{
    \ottdruleEvalXXrefDirty{}
    \\[2.5ex]
    \ottdruleEvalXXrefClean{}
  }
\end{array}
\vspace{1ex}
\\
~\hspace{-2.5ex}
{
\renewcommand{\tweakrulename}{\!\!\runonfontsz{8pt}}
        \begin{array}{l@{~}||@{\,}l}
          \arrayenvc{\nonincrcolorbox{\ottdruleEvalXXthunkPlain{}}}
          &
          \arrayenvc{
            \ottdruleEvalXXthunkDirty{}
            \\[2.5ex]
            \ottdruleEvalXXthunkClean{}
          }
        \end{array}
}
\vspace{1ex}
\\
\begin{array}{l||l}
  \nonincrcolorbox{\ottdruleEvalXXgetPlain{}}
  &
  \ottdruleEvalXXgetClean{}
\end{array}
\vspace{1ex}
\\
~\hspace{-2ex}
\begin{array}{l@{~}||@{\;}l}
  \nonincrcolorbox{\arrayenvcl{\ottdruleEvalXXforcePlain{}}}
  &
  \arrayenvcl{
      \ottdruleEvalXXforceClean{}
      \vspace{2.5ex}
      \\
      \ottdruleEvalXXscrubEdge{}
      \vspace{2.5ex}
      \\
      \ottdruleEvalXXcomputeDep{}
  }
\end{array}
\end{mathpar}

\lesscaptionspacing
\caption{Evaluation rules of \calc; vertical bars separate
  non-incremental rules (left, \protect\nonincrcolorbox{\text{shaded\;}}) from incremental rules (right)}
\FLabel{fig:eval}
\end{figure*}

\subsection{Semantics}

We define a big-step operational semantics (\Figureref{fig:eval}) with judgments
$ G_{{\mathrm{1}}}   \vdash ^{ p }_{ \omega }  e   \Downarrow   G_{{\mathrm{2}}}  ;  \mathrm{t} $, which states that under input graph
$G_{{\mathrm{1}}}$ and within namespace $\omega$, evaluating the computation $e$
produces output graph $G_{{\mathrm{2}}}$ and result $\mathrm{t}$,
where $e$'s evaluation was triggered by a previous force of thunk $p$.

\paragraph{Incremental and Reference Systems.}
In order to state and later prove our central meta-theoretic result, from-scratch consistency,
we define two closely-related systems of evaluation rules:
an \emph{incremental system} and a \emph{non-incremental} or \emph{reference}
system. The incremental system models \NominalAdapton programs that transform a
graph whose nodes are store locations (values and thunks) and whose edges represent
dependencies (an edge from $p$ to $q$ means that $q$ depends on $p$); the reference
system models call-by-push-value programs under a plain store (a graph with no edges).
Since it has no IC mechanisms,
everything the reference system does is, by definition, from-scratch consistent.

Rules above the double horizontal line in \Figureref{fig:eval} do not manipulate
the graph, and are common to the incremental and reference systems.  The shaded
rules, to the left of vertical double lines, are non-incremental rules that never
create edges and do not cache results.  The rules to their right create edges, store
cached results, and recompute results that have become invalid.

\subsubsection{Common Rules}

Several of the rules at the top of \Figureref{fig:eval}
are derived from standard CBPV rules.

Rules for standard language features (pairs, sums, functions, \textbf{fix},
and \textbf{let}) straightforwardly adapt the standard rules,
ignoring $p$ and $\omega$ and ``threading through'' input and output
graphs. For example, \txtsf{Eval-app} evaluates a function $e_1$ to get a terminal
computation $\lambda  x  \ottsym{.}  e_{{\mathrm{2}}}$ and substitutes the argument $v$ for $x$, threading
through the graph: evaluating $e_1$ produces $G_2$, which is given as input
to the second premise, resulting in output graph $G_3$.
Rule \txtsf{Eval-case}, applying \textbf{case} to a sum $ \ottkw{inj} _{ i }~ v $,
substitutes $v$ for $x_i$ in the appropriate case arm.

The last three shared rules are not standard: they deal with names and namespaces.

\begin{itemize}
\item \txtsf{Eval-fork} splits a name $k$ into children $k{\cdot}1$ and $k{\cdot}2$.
  Once forked, the name $k$ should not be used to allocate a new reference
  or thunk, nor should $k$ be forked again.

\item Running in namespace $\omega$, \txtsf{Eval-namespace} makes a new namespace
  $ \omega . k $ and substitutes it for $x$ in the body $e$.

\item Running in namespace $\omega$, \txtsf{Eval-nest} runs $e_1$ in a different namespace
  $\mu$ and then returns to $\omega$ to run $e_2$, with $x$ replaced by the result of
  running $e_1$.
\end{itemize}

\subsubsection{Non-Incremental Rules}

These rules cover allocation and use of references and thunks.  Like the incremental
rules (discussed below), they use names and namespaces; however, they do not
cache the results of thunks.  We discuss the non-incremental rules first, because they
are simpler and provide a kind of skeleton for the incremental rules.

\begin{itemize}
\item \txtsf{Eval-refPlain} checks that the pointer described by $q =  k  @  \omega $
   is fresh ($ \Grn{ q } \notin   \txtsf{dom}( G_{{\mathrm{1}}} )  $),
   adds a node $q$ with contents $v$ to the graph ($ G_{{\mathrm{1}}} \{ q {\mapsto} v \}   \ottsym{=}  G_{{\mathrm{2}}})$,
   and returns a reference $ \ottkw{ref} \,\Grn{ q } $.

\item \txtsf{Eval-thunkPlain} is similar to \txtsf{Eval-refPlain}, but creates a node
  with a suspended computation $e$ instead of a value.

\item \txtsf{Eval-getPlain} returns the contents of the pointer $q$.

\item \txtsf{Eval-forcePlain} extracts the computation stored in a thunk
  ($ G_{{\mathrm{1}}} ( q )   \ottsym{=}  e$) and evaluates it under the namespace of $q$;
  that is, if $q  \ottsym{=}   k  @  \mu $, it evaluates it under $\mu$.
  (The rule uses~$ \txtsf{namespace}(   k  @  \mu   )  = \mu$).
\end{itemize}

\begin{figure}[t]
  \runonfontsz{8pt}%
  \begin{tabular}[t]{c}
    $\begin{array}[t]{lll}
       \txtsf{exp}( G ,  p ) 
      =
          e
          ~~~~
          \text{if~}  G ( p )   \ottsym{=}  e%
          \text{~or~}%
           G ( p ) = ( e , \mathrm{t} ) 
      \\[1ex]
       G \{ p {\mapsto} v \} 
      = G'
          \\
          \quad
          \text{where, if $ \Grn{ p } \notin   \txtsf{dom}( G )  $, then $G'  \ottsym{=}  \ottsym{(}  G  \ottsym{,}   p {:} v   \ottsym{)}$}
          \\
          \quad
          \text{otherwise, if $G  \ottsym{=}  \ottsym{(}  G_{{\mathrm{1}}}  \ottsym{,}   p {:} v'   \ottsym{,}  G_{{\mathrm{2}}}  \ottsym{)}$ then $G'  \ottsym{=}  \ottsym{(}  G_{{\mathrm{1}}}  \ottsym{,}   p {:} v   \ottsym{,}  G_{{\mathrm{2}}}  \ottsym{)}$}
      \\[1ex]
       G \{ p {\mapsto} e \} 
      = G'
          \\
          \quad
          \text{where, if $ \Grn{ p } \notin   \txtsf{dom}( G )  $, then $G'  \ottsym{=}  \ottsym{(}  G  \ottsym{,}   p {:} e   \ottsym{)}$}
          \\
          \quad
          \text{otherwise, 
            \tabularenvl{if $G  \ottsym{=}  \ottsym{(}  G_{{\mathrm{1}}}  \ottsym{,}   p {:} e'   \ottsym{,}  G_{{\mathrm{2}}}  \ottsym{)}$
              or $G  \ottsym{=}  \ottsym{(}  G_{{\mathrm{1}}}  \ottsym{,}   p {:}( e' , \mathrm{t}' )   \ottsym{,}  G_{{\mathrm{2}}}  \ottsym{)}$
              \\
              then $G'  \ottsym{=}  \ottsym{(}  G_{{\mathrm{1}}}  \ottsym{,}   p {:} e   \ottsym{,}  G_{{\mathrm{2}}}  \ottsym{)}$}
           }
      \\[1ex]
       G \{ p {\mapsto}( e , \mathrm{t} )\} 
      = G'
          \\
          \quad
          \text{where, if $ \Grn{ p } \notin   \txtsf{dom}( G )  $, then $G'  \ottsym{=}  \ottsym{(}  G  \ottsym{,}   p {:}( e , \mathrm{t} )   \ottsym{)}$}
          \\
          \quad
          \text{otherwise, 
            \tabularenvl{if $G  \ottsym{=}  \ottsym{(}  G_{{\mathrm{1}}}  \ottsym{,}   p {:} e'   \ottsym{,}  G_{{\mathrm{2}}}  \ottsym{)}$
              or $G  \ottsym{=}  \ottsym{(}  G_{{\mathrm{1}}}  \ottsym{,}   p {:}( e' , \mathrm{t}' )   \ottsym{,}  G_{{\mathrm{2}}}  \ottsym{)}$
              \\
              then $G'  \ottsym{=}  \ottsym{(}  G_{{\mathrm{1}}}  \ottsym{,}   p {:}( e , \mathrm{t} )   \ottsym{,}  G_{{\mathrm{2}}}  \ottsym{)}$}
           }
      \vspace{1ex}
      ~\\
      \txtsf{all-clean-out}(G, p) ~\equiv~
          \forall (p, a, b, q) \in G.~ (b = \txtsf{clean})
      \\[1ex]
      \big(\txtsf{del-edges-out}(G_1,p) = G_2\big) \equiv
      \\
      \quad
        \big(\txtsf{nodes}(G_1) = \txtsf{nodes}(G_2)\big)
      \\
      \quad
        \text{and~}
        \forall q \ne p.~ \big((q,a,b,q') \in G_1\big) \Rightarrow \big((q,a,b,q') \in G_2\big)
      \\
      \quad
        \text{and~}
        \nexists a,b,q.~ (p,a,b,q) \in G_2
      \\[2mm]

      \big(\txtsf{dirty-paths-in}(G_1,p) = G_2\big) \equiv
      \\
      \quad
        \big(\txtsf{nodes}(G_1) = \txtsf{nodes}(G_2)\big)
      \\
      \quad
        \text{and~}
        \forall q_1, q_2.~
            \big((q_1,a,b',q_2) \in G_2\big) \Rightarrow \big((q_1,a,b,q_2) \in G_1\big)
      \\
      \quad
        \text{and~}
          \forall q_1, q_2.~
            \big((q_1,a,b,q_2) \in G_1\big) \Rightarrow
      \\
        \quad\quad
             \text{there exists~}(q_1,a,b',q_2) \in G_2~\text{such that}
      \\
        \quad\quad\quad
             \text{if}~\txtsf{path}(q_2,p,G_1)\text{~then~}
      (b' = \txtsf{dirty})
      \text{~else~}
      (b' = b)
      \\[2mm]
      \txtsf{path}(p,q,G) ~\equiv
      \arrayenvl{
                 ~~~~~
                  \big((p,a,b,q) \in G\big)
              \\[0.5ex]
                  \text{or~}\;
                  \big(
                       \exists p'.~ ((p,a,b,p') \in G)
                       \text{~and~}
                       \txtsf{path}(p',q,G)
                  \big)
      }
      \vspace{1ex}
      \\
       \txtsf{consistent-action} ( G ,  a ,  q  ) {:}
      \\
      \quad
           \txtsf{consistent-action} ( G ,  \ottkw{obs} \, v ,  q  ) \text{~when~}G(q) = v
      \\
      \quad
           \txtsf{consistent-action} ( G ,  \ottkw{obs} \, \mathrm{t} ,  q  ) \text{~when~}G(q) = (e,t)
      \\
      \quad
           \txtsf{consistent-action} ( G ,  \ottkw{alloc} \, v ,  q  ) \text{~when~}G(q) = v
      \\
      \quad
           \txtsf{consistent-action} ( G ,  \ottkw{alloc} \, e ,  q  ) \text{~when~}G(q) = e
    \end{array}$
  \end{tabular}
  \vspace{-1ex}
  \caption{Graph predicates and transformations}
  \FLabel{fig:predicates}
\end{figure}

\subsubsection{Incremental Rules}

Each non-incremental rule
corresponds to one or more incremental rules: the incremental semantics
is influenced by the graph edges, which are not present in the non-incremental
system.  For example, \txtsf{Eval-refPlain} is replaced by \txtsf{Eval-refDirty}
and \txtsf{Eval-refClean}.

These rules use some predicates and operations, such as \txtsf{dirty-paths-in},
that we explain informally as we describe the rules; they are fully defined
in \Figureref{fig:predicates}.

Incremental computation arises by making $G_{{\mathrm{1}}}$
a modification of a previously produced graph $G_{{\mathrm{2}}}$, and then
re-running $e$.  A legal modification involves replacing references
$ p {:} v $ with $ p {:} v' $ and dirtying all edges along
paths to $p$ in the DCG:
$ \txtsf{dirty-paths-in} ( G_{{\mathrm{1}}} , p ) $ is the same as $G_{{\mathrm{1}}}$ but with
edges on paths to $p$ marked $ \txtsf{dirty} $.

\paragraph{Creating Thunks.}
The DCG is constructed during evaluation.  The main rule for
creating a thunk is \txtsf{Eval-thunkDirty},
which
converts computation $e$ into a thunk by generating a
pointer $q$ from the provided name $k$, which couples the name
with the current namespace $\omega$. The output graph is updated to map
$q$ to $e$. If $q$ happens to be in the graph
already, all paths to it will be dirtied.  Finally, the rule adds edge
$\ottsym{(}  p  \ottsym{,}  \ottkw{alloc} \, e  \ottsym{,}   \txtsf{clean}   \ottsym{,}  q  \ottsym{)}$ to the output graph, indicating that
the currently evaluating thunk $p$ depends on $q$ and is
currently clean.

\paragraph{Forcing Thunks.}
Forcing a thunk that has not been previously computed, an operation that involves
one rule in the non-incremental system (\txtsf{Eval-forcePlain}), involves at least two
rules in the incremental system: \txtsf{Eval-computeDep} and \txtsf{Eval-forceClean}.

Rule \txtsf{Eval-forceClean} performs memoization:
Given $q$ pointing to $(e, t)$, where $t$ is the cached result of the thunk $e$,
if $q$'s outgoing edges are clean then $t$ is consistent and can be reused.
Thus \txtsf{Eval-forceClean} returns $t$ immediately without reevaluating $e$,
but adds an edge denoting that $p$ has observed the result of $q$ to be $t$.

Rule \txtsf{Eval-computeDep} applies when $e = \ottkw{force} \, \ottsym{(}   \ottkw{thk} \,\Grn{ p_{{\mathrm{0}}} }   \ottsym{)}$.
This rule serves two purposes: it forces thunks for the first time, and it
selectively recomputes until a cached result can be reused.

Its first premise $ \txtsf{exp}( G_{{\mathrm{1}}} ,  q )   \ottsym{=}  e'$
nondeterministically chooses some thunk $q$ whose suspended
expression is $e'$ (whether or not $q$ also has a cached result).
Its second premise $ \txtsf{del-edges-out} (  G_{{\mathrm{1}}} \{ q {\mapsto} e' \}  , q )   \ottsym{=}  G'_{{\mathrm{1}}}$
updates $G_{{\mathrm{1}}}$ so that $q$ points to $e'$ (removing $q$'s cached
result, if any), and deletes outward edges of $q$.
We need to delete the outward edges before evaluating $e'$ because
they represent what a \emph{previous} evaluation of $e'$ depended on.
The third premise recomputes $q$'s expression $e'$,
with $q$ as the current thunk and $q$'s namespace
component as the current namespace. %
In the fourth premise, the
recomputed result $\mathrm{t}'$ is cached, resulting in graph $G_2'$.

The final premise evaluates $\ottkw{force} \, \ottsym{(}   \ottkw{thk} \,\Grn{ p_{{\mathrm{0}}} }   \ottsym{)}$,
the same expression as the conclusion, but under a graph
$G_2'$ containing the result of evaluating $e'$.
In deriving this premise we may again apply \txtsf{Eval-computeDep}
to ``fix up'' other nodes of the graph, but will eventually end up with
the thunk $q$
chosen in the first premise being $p_{{\mathrm{0}}}$ itself.  In this case,
the last premise of \txtsf{Eval-computeDep} will be derived by
\txtsf{Eval-forceClean} (with $q$ instantiated to $p_{{\mathrm{0}}}$).

We skipped the fifth premise $ \txtsf{all-clean-out}( G'_{{\mathrm{2}}} , q ) $,
which demands that all outgoing edges from $q$ in the updated graph are
clean.
This consistency check ensures that the program has
not used the same name for two different thunks or references, e.g.,
by calling $ \textbf{thunk}(\Grn{  \Grn{ \ottkw{nm} \, k }  }, e_{{\mathrm{1}}} ) $ and later $ \textbf{thunk}(\Grn{  \Grn{ \ottkw{nm} \, k }  }, e_{{\mathrm{2}}} ) $ in
the same namespace $\omega$.  If this happens, the graph will first map
$ k  @  \omega $ to $e_{{\mathrm{1}}}$ but will later map it to $e_{{\mathrm{2}}}$.
Without this check, a computation $q$ that depends on both
$e_{{\mathrm{1}}}$ and $e_{{\mathrm{2}}}$ could be incorrect, because (re-)computing one of
them might use cached values that were due to the other. Fortunately,
this potential inconsistency is detected by \txtsf{all-clean-out}:
When a recomputation of
$q$ results in $ k  @  \omega $ being mapped to a different value, all
\emph{existing} paths into $ k  @  \omega $ are dirtied (by the last
premise of \txtsf{Eval-thunkDirty} above).  Since $q$
is one of the dependents, it will detect that fact and can
signal an error.

This fifth premise formalizes the double-use checking
algorithm first described in \secref{correctnaming}.  In
particular, each use of \txtsf{Eval-computeDep} corresponds to the
implementation pushing (and later popping) a node from its force
stack.  By inspecting the outgoing edges upon each pop, it
effectively verifies that each node popped from the stack is clean.

  Note that in the case $q = p_0$---where the ``dependency'' being
  computed is $p_0$ itself---the last premise could be derived
  by \txtsf{Eval-forceClean}, which looks up the result just computed by
  $ G'_{{\mathrm{1}}}   \vdash ^{ q }_{  \txtsf{namespace}(  q  )  }  e'   \Downarrow   G_{{\mathrm{2}}}  ;  \mathrm{t}' $.

\paragraph{Replacing Dirty Edges with Clean Edges.}
\txtsf{Eval-scrubEdge} replaces a dirty edge $\ottsym{(}  q_{{\mathrm{1}}}  \ottsym{,}  a  \ottsym{,}   \txtsf{dirty}   \ottsym{,}  q_{{\mathrm{2}}}  \ottsym{)}$
  with a clean edge $\ottsym{(}  q_{{\mathrm{1}}}  \ottsym{,}  a  \ottsym{,}   \txtsf{clean}   \ottsym{,}  q_{{\mathrm{2}}}  \ottsym{)}$.
  First, it checks that all edges out from $q_2$ are clean; this means
  that the contents of $q_2$ are up-to-date.  Next, it checks that the action $a$
  that represents $q_1$'s dependency on $q_2$ is consistent with
  the contents of $q_2$.  For example, if $q_2$ points to a thunk with
  a cached result $(e_2, t_2)$ and $a = \ottkw{obs} \, \mathrm{t}$, then the ``consistent-action''
  premise checks that the currently cached result $t_2$ matches the
  result $t$ that was previously used by $q_1$.

\paragraph{Creating Reference Nodes.}
 Like \txtsf{Eval-refPlain}, \txtsf{Eval-refDirty} creates a node $q$ with value $v$:
  $ G_{{\mathrm{1}}} \{ q {\mapsto} v \}   \ottsym{=}  G_{{\mathrm{2}}}$.  Unlike \txtsf{Eval-refPlain}, \txtsf{Eval-refDirty} does not check
  that $q$ is not in the graph: if we are recomputing, $q$ may already exist.
  So $ G_{{\mathrm{1}}} \{ q {\mapsto} v \}   \ottsym{=}  G_{{\mathrm{2}}}$ either creates $q$ pointing to $v$, or updates
  $q$ by replacing its value with $v$.  It then marks the edges along all paths
  into $q$ as dirty $ \txtsf{dirty-paths-in} ( G_{{\mathrm{2}}} , q )   \ottsym{=}  G_{{\mathrm{3}}}$; these are the paths from
  nodes that depend on $q$.

\paragraph{Re-creating Clean References.}
\txtsf{Eval-refClean} can be applied only during recomputation, and only
  when $G(q) = v$.  That is, we are evaluating
  $ \textbf{ref}(\Grn{  \Grn{ \ottkw{nm} \, k }  }, v ) $ and allocating the same value as the previous run.
  Since the values are the same, we need not mark any dependency edges
  as dirty, but we do add an edge to remember that $p$ depends on $q$.

\paragraph{Creating Thunks.}
\txtsf{Eval-thunkDirty} corresponds exactly to \txtsf{Eval-refDirty}, but for thunks
  rather than values.

\paragraph{Re-creating Thunks.}
\txtsf{Eval-thunkClean} corresponds to \txtsf{Eval-refClean} and does not
  change the contents of $q$.  Note that the condition $\ottsym{(}   \txtsf{exp}( G ,  q )   \ottsym{)}  \ottsym{=}  e$
  applies whether or not $q$ includes a cached result.
  If a cached result is present, that is, $G(q) = (e,t)$, it remains in the output graph.

\paragraph{Reading References.}
\txtsf{Eval-getClean} is the same as \txtsf{Eval-getPlain}, except that it adds an
  edge representing the dependency created by reading the contents of $q$.

\paragraph{The Rules vs.\ the Implementation.}
Our rules are not intended to be an ``instruction manual'' for building an implementation;
rather, they are intended to model our implementation.
To keep the rules simple, we underspecify two aspects of the implementation.

First, when recomputing an allocation in the case when the allocated value
is equal to the previously allocated value ($G(q) = v$), either \txtsf{Eval-refDirty}
or \txtsf{Eval-refClean} applies.  However, in this situation the implementation always
follows the behavior of \txtsf{Eval-refClean}, since that is the choice that avoids unnecessarily
marking edges as dirty and causing more recomputation.
An analogous choice exists for \txtsf{Eval-thunkClean} versus \txtsf{Eval-thunkDirty}.

Second, similarly to the original \adapton system, 
the timing of dirtying and cleaning is left open:
\txtsf{Eval-scrubEdge} can be applied to dirty edges with no particular connection to
the $p_0$ mentioned in the subject expression $\ottkw{force} \, \ottsym{(}   \ottkw{thk} \,\Grn{ p_{{\mathrm{0}}} }   \ottsym{)}$,
and the timing of recomputation via \txtsf{Eval-computeDep} is left
open as well.
Our implementations of dirtying and re-evaluation fix these open choices,
and they are each analogous to the algorithms found in the original
\adapton work.
For details, we refer the interested reader to Algorithm~1
in \cite{Adapton2014}.

\subsection{From-Scratch Consistency}
We show that an incremental computation modeled by our
evaluation rules has a corresponding non-incremental computation:
given an incremental evaluation of $e$ that produced $t$,
a corresponding non-incremental evaluation %
also produces $t$.  Moreover, the values and
expressions in the incremental output graph match those in
the graph produced by the non-incremental evaluation.

Eliding some details and generalizations, the from-scratch consistency
result is:

\begin{thm*}
  If incremental $\D_i$ derives $ G_{{\mathrm{1}}}   \vdash ^{ p }_{ \omega }  e   \Downarrow   G_{{\mathrm{2}}}  ;  \mathrm{t} $,
  then a non-incremental
  $\D_{ni}$ derives $  \lfloor  G_{{\mathrm{1}}}  \rfloor_{ P_{{\mathrm{1}}} }    \vdash ^{ p }_{ \omega }  e   \Downarrow    \lfloor  G_{{\mathrm{2}}}  \rfloor_{ P_{{\mathrm{2}}} }   ;  \mathrm{t} $
  where
  $  \lfloor  G_{{\mathrm{1}}}  \rfloor_{ P_{{\mathrm{1}}} }   \subseteq   \lfloor  G_{{\mathrm{2}}}  \rfloor_{ P_{{\mathrm{2}}} }  $
  and
  $P_{{\mathrm{2}}}  \ottsym{=}  P_{{\mathrm{1}}} \, \cup \,  \txtsf{dom}( W ) $.
\end{thm*}

Here, $W$ is the set of pointers that $\D_i$ may allocate.
The restriction function $ \lfloor  G_{\ottmv{i}}  \rfloor_{ P_{\ottmv{i}} } $ drops
all edges from $G_{\ottmv{i}}$ and keeps only nodes in the set $P_{\ottmv{i}}$.
It also removes any cached results $t$.
The set $P_{\ottmv{i}}$ corresponds to the nodes in $G_{\ottmv{i}}$ that are present at
this point in the non-incremental derivation, which may differ
from the incremental derivation since \txtsf{Eval-computeDep}
need not compute dependencies in left-to-right order.

The full statement, along with definitions of $W$, the restriction
function, and lemmas, is in the extended version~\citep{arxiv-version} as
\iftr
\Theoremref{thm:fsc}.
\else
Theorem B.13.
\fi

\iffalse
\fi

\section{Implementation}
\label{sec:implementation}

We implemented \NominalAdapton as an OCaml library. In this section,
we describe its programming interface, data structures, and
algorithms.  Additional details about memory management appear in
Appendix~%
\iftr
\ref{sec:spacemanagement}.
\else
A.
\fi
The code for \NominalAdapton is freely available:
\begin{center}
\url{https://github.com/plum-umd/adapton.ocaml}
\end{center}

\subsection{Programming Interface}

\begin{figure}[t]
\begin{OCaml}
type name
val new : unit -> name
val fork : name -> name * name

type ('arg,'res) mfn
type ('arg,'res) mbody = ('arg,'res) mfn -> 'arg -> 'res
val mk_mfn : name -> ('arg,'res) mbody -> ('arg,'res) mfn
val call : ('arg,'res) mfn -> ('arg -> 'res)

type 'res athunk
val thunk : ('arg,'res) mfn -> name -> 'arg -> 'res athunk
val force : 'res athunk -> 'res

type 'a aref
val aref : name -> 'a -> 'a aref
val get : 'a aref -> 'a
val set : 'a aref -> 'a -> unit
\end{OCaml}
\nocaptionrule
\caption{Basic \NominalAdapton API}
\label{fig:adapton-api}
\end{figure}

\figref{adapton-api} shows the basic \NominalAdapton API.  Two
of the data types, \ml{name} and \ml{aref}, correspond exactly to
names and references in \secreftwo{overview}{corecalc}. The other data
types, \ml{mfn} and \ml{athunk}, work a little differently, due to
limitations of OCaml: In OCaml, we cannot type a general-purpose memo
table (containing thunks with non-uniform types), nor can we examine a
thunk's ``arguments'' (that is, the values of the variables in a
closure's environment).\footnote{Recall from the start of the
  previous section that memoized calls are implemented as thunks in
  \NominalAdapton.} 

To overcome these limitations, our implementation creates a tight
coupling between namespaces and memoized functions. The function
\ml{mk_mfn k f} takes a name \ml{k} and a function \ml{f} and returns
a \emph{memoized function} \ml{mfn}. The function \ml{f} must have
type \ml{('arg,'res) mbody}, i.e., it takes an \ml{mfn} and an
\ml{'arg} as arguments, and produces a \ml{'res}. (The \ml{mfn} is for
recursive calls; see the example below.)

Later on, we call \ml{thunk m k arg} to create a thunk of type
\ml{athunk} from the memoized function \ml{m}, with thunk name \ml{k}
(relative to \ml{m}'s namespace) and argument \ml{arg}. The code for
the thunk will be whatever function \ml{m} was created from. In other
words, in our implementation, only function calls can be memoized (not
arbitrary expressions), and each set of thunks that share the same
function body also share the same namespace.

Using this API, we can rewrite the \ml{map} code from
Section~\ref{sec:ovw-nominal} as follows:
\vspace*{1ex}

\begin{OCaml}
  let map f map_f_name =
    let mfn = mk_mfn map_f_name (fun mfn list ->
      match list with
      | Nil -> Nil
      | Cons(hd, n, tl_ref) ->
        let n1, n2 = fork n in
        let tl = get tl_ref in
        Cons(f hd, aref(n1, force(thunk mfn n2 tl))) )
    in fun list -> call mfn list
\end{OCaml}

The code above differs from the earlier version in that the programmer
uses \ml{mk_mfn} with the name \ml{map_f_name} to create a memo table
in a fresh namespace. Moreover, memoization happens directly on the
recursive call, by introducing a thunk (and immediately forcing
it).

\subsection{Implementing Reuse}

Much of the implementation of \NominalAdapton remains unchanged from
classic \Adapton.
Specifically, both systems use DCGs to represent dependency
information among nodes representing thunks and refs, and both systems
traverse their DCGs to dirty dependencies and to
later reuse (and repair) partially inconsistent graph components.
These steps were described in \secreftwo{overview}{corecalc}, and
were detailed further by \citet{Adapton2014}.

The key differences between \NominalAdapton and Classic \Adapton have
to do with memo tables and thunks.

\paragraph{Memo Tables and Thunks.}

\NominalAdapton memo tables are implemented as maps from names to DCG nodes, which
contain the thunks they represent.
When creating memo tables with~\ml{mk_mfn}, the programmer supplies
a name and an~\ml{mbody}. Using the name, the library checks for an
existing table (i.e., a namespace).  If none exists, it creates an empty table, registers
it globally, and returns it as an \ml{mfn}.  If a table
exists, then the library checks that the given~\ml{mbody} is
(physically) equal to the \ml{mbody} component of the existing~\ml{mfn}; it issues a
run-time error if not.%

When the program invokes~\ml{thunk}, it provides an~\ml{mfn},
name, and argument.  The library checks the \ml{mfn}'s memo table
for an existing node with the provided name.
If none exists, it registers a fresh node with the given name in the
memo table and adds an allocation edge to it from the current node
(which is set whenever a thunk is forced).

If a node with the same name already exists, the library checks
whether the argument is equal to the current one. If equal, then the
thunk previously associated with the name is the same as the new
thunk, so the library reuses the
node, returning it as an~\ml{athunk} and adding an allocation edge to
the DCG.
If not equal, then the name has been allocated for a different thunk
either in a prior run, or in \emph{this run}.
The latter case is an error that we detect and signal.
To distinguish these two cases, we use the check described in \secref{correctnaming}.
Assuming no error, the library needs to reset the state associated
with the name: It clears any prior cached result, dirties any incoming
edges (transitively), mutates the argument stored in the node to be
the new one, and adds an allocation edge.
Later, when and if this thunk is forced, the system will run it.
Further, because of the dirtying traversal, any
nodes that (transitively or directly) forced this changed node are
also candidates for reevaluation.

\paragraph{Names.}
A name in \NominalAdapton is implemented as a kind of
list, as follows:
\vspace*{.1in}
\begin{OCaml}
type name = Bullet | One of int * name | Two of int * name
\end{OCaml}
Ignoring the \ml{int} part, this is a direct implementation of \calc's
notion of names. The \ml{int} part is a hash of the next element in
the list (but not beyond it), to speed up disequality
checks---if two \ml{One} or \ml{Two} elements do not share the same
hash field they cannot be equal; if they do, we must compare their
tails (because of hash collisions). Thus, at worst,
establishing equality is linear in the length of the name,
but we can short circuit a full
traversal in many cases.
We note that in our applications, the size of names is either a
constant, or it is proportional to the \emph{depth} (not total size)
of the DCG, which is usually sublinear (e.g., logarithmic) in the
current input's total size.

\section{Experimental Results}
\label{sec:experiments}

This section evaluates \NominalAdapton's performance against
\Adapton and from-scratch recomputation.\footnote{\cite{Adapton2014}
  report that for interactive, lazy usage patterns, \Adapton
substantially outperforms another state-of-the-art incremental
technique, \emph{self-adjusting computation} (SAC), which sometimes
can incur significant slowdowns. We do not compare directly against
SAC here.}
We find that \NominalAdapton is nearly always faster than \Adapton,
which is sometimes \emph{orders of magnitude slower} than from-scratch
computation. \NominalAdapton always enjoys speedups, and sometimes
very dramatic ones (up to $10900\times$).

\subsection{Experimental Setup}
\Label{benchmarks}

Our experiments measure the time taken to recompute the output of a
program after a change to the input, for a variety of different sorts
of changes. We compare \NominalAdapton against classic \Adapton and
from-scratch computation on the changed input; the latter avoids all
IC-related overhead and therefore represents the best from-scratch
time possible.

We evaluate two kinds of subject programs. The first set is drawn from
the IC literature on SAC and Adapton~\citep{Hammer09:ceal,Hammer11:stack-machines,Adapton2014}.
These consist of standard list processing programs: (eager and lazy) \ml{filter},
(eager and lazy) \ml{map}, \ml{reduce(min)}, \ml{reduce(sum)}, \ml{reverse}, \ml{median}, and a
list-based \ml{mergesort} algorithm.  Each program operates over randomly
generated lists. These aim to represent key primitives that are likely
to arise in standard functional programs, and use the patterns
discussed in Section~\ref{sec:fpn}. We also consider an
implementation of \ml{quickhull}~\citep{quickhull}, a
divide-and-conquer method for computing the convex hull of a set of
points in a plane.  Convex hull has a number of applications including
pattern recognition, abstract interpretation, computational geometry,
and statistics.

We also evaluate an incremental IMP interpreter, as discussed in Section~\ref{sec:imp},
measuring its performance on a variety of different IMP
programs. \ml{fact} iteratively computes the factorial of an
integer. \ml{intlog;fact} evaluates the sequence of computing an
integer logarithm followed by factorial. \ml{array max} allocates,
initializes, and destructively computes the maximum value in an
array. \ml{matrix mult} allocates, initializes, and multiplies two
square matrices (implemented as arrays of arrays of integers). These
IMP programs exhibit imperative behavior not otherwise
incrementalizable, except as programs evaluated by a purely
functional, big-step interpreter implemented in an incremental
meta-language.

\begin{table}
\begin{subtable}{\columnwidth}
\footnotesize
\begin{centering}
  \begin{tabular}{|c|c|l||r||r|r|}
  \hline
  \multicolumn{6}{|c|}{Batch-mode comparison (\emph{``demand all''})}
  \\
  \hline
  \textbf{Program} &
  \textbf{$n$} &
  \textbf{Edit} &
  \textsf{FS} (ms) &
  \textsf{A} ($\times$) &
  \textsf{NA} ($\times$)
  \\
  \hline
  eager filter     & 1e4 & insert & 21 & 0.178 & 1.29 \\
             &     & delete & 21 & 0.257 & 1.39\\
             &     & replace & 21 & 0.108 & 1.27\\ \hline

  eager map  & 1e4 & insert & 21.6 & 0.0803 & 1.02 \\
             &     & delete & 21.6 & 0.0920 & 1.01\\
             &     & replace & 21.6 & 0.0841 & 1.09\\ \hline

  min        & 1e5 & insert & 424 & 2790 & 2980 \\
             &     & delete & 424 & 4450 & 4720\\
             &     & replace & 424 & 1850 & 2310\\ \hline

  sum        & 1e5 & insert & 421 & 785 & 833 \\
             &     & delete & 421 & 1140 & 1230\\
             &     & replace & 421 & 727 & 733\\ \hline

  reverse    & 1e5 & insert & 197 & 0.0404 & 1.23 \\
             &     & delete & 197 & 0.764 & 1.19\\
             &     & replace & 197 & 0.0404 & 1.23\\ \hline

  median     & 1e4 & insert & 3010 & 0.747 & 127 \\
             &     & delete & 3010 & 192 & 115\\
             &     & replace & 3010 & 0.755 & 148\\ \hline

  mergesort  & 1e4 & insert & 267 & 0.212 & 12.0 \\
             &     & delete & 267 & 11.0 & 10.1\\
             &     & replace & 267 & 0.205 & 10.5\\ \hline

  quickhull  & 1e4 & insert & 853 & 0.0256* & 3.78\\
             &     & delete & 853 & 0.0270* & 4.11\\
             &     & replace & 853 & 0.0378* & 3.86\\ \hline

   \hline
\end{tabular}
  \caption{Speedups of batch-mode experiments}
\label{tab:summary-demand-all}
\end{centering}
\end{subtable}

\vskip 1em

\begin{subtable}{\columnwidth}
\footnotesize
\begin{centering}
  \begin{tabular}{|c|c|l||r||r|r|}
  \hline
  \multicolumn{6}{|c|}{Demand-driven comparison (\emph{``demand one''})}
  \\
  \hline
  \textbf{Program} &
  \textbf{$n$} &
  \textbf{Edit} &
  \textsf{FS} (ms) &
  \textsf{A} ($\times$) &
  \textsf{NA} ($\times$)
  \\
  \hline
  lazy filter & 1e5 & insert & 0.016 & 3.79 & 3.55\\
              &     & delete & 0.016 & 18.1 & 16.3\\
              &     & replace & 0.016 & 3.55 & 3.20\\ \hline

  lazy map   & 1e5 & insert & 0.016 & 4.08 & 3.79\\
             &     & delete & 0.016 & 18.1 & 20.4\\
             &     & replace & 0.016 & 3.71 & 3.62\\ \hline

  reverse    & 1e5 & insert & 188 & 0.067 & 2130 \\
             &     & delete & 188 & 50.8 & 4540\\
             &     & replace & 188 & 0.068 & 2360\\ \hline

  mergesort  & 1e4 & insert & 63.4 & 96.3 & 369 \\
             &     & delete & 63.4 & 111 & 752\\
             &     & replace & 63.4 & 86.2 & 336\\ \hline

  quickhull  & 1e4 & insert & 509 & 0.0628* & 5.30\\
             &     & delete & 509 & 0.0571* & 5.52\\
             &     & replace & 509 & 0.0856* & 5.23\\

  \hline
\end{tabular}
\caption{Speedups of demand-driven experiments}
\label{tab:summary-demand-one}
\end{centering}
\end{subtable}
\caption{List benchmarks}
\label{tab:summary}
\end{table}

All programs were compiled using OCaml 4.01.0 and run on an 8-core,
2.26 GHz Intel Mac Pro with 16 GB of RAM running Mac OS X 10.6.8.

\subsection{List-Based Experiments}

Table~\ref{tab:summary} contains the results of our list
experiments.  For each program (leftmost column), we consider a
randomly generated input of size $n$ and three kinds of edits to it:
\emph{insert}, \emph{delete}, and \emph{replace}. For the first, we
insert an element in the list; for the second, we delete the inserted
element; for the last, we delete an element and then
re-insert an element with a new value. Rather than consider only one
edit position, we consider
ten positions in the input list, spaced evenly (1/10 through the list,
2/10 through the list, etc.), and perform the edit at those positions,
computing the average time across all ten edits. We report the
median of seven trials of this experiment.

The table reports the time to perform recomputation from scratch, in
milliseconds, in column \textsf{FS}, and then the speed-up (or
slow-down) factor compared to the from-scratch time for both \Adapton,
in column \textsf{A}, and \NominalAdapton, in column
\textsf{NA}. Table~\ref{tab:summary-demand-all} considers
the case when \emph{all} of the program's output is demanded, whereas
Table~\ref{tab:summary-demand-one} considers the case when only one
element of the output is demanded, thus measuring the benefits of both
nominal and classic \Adapton in a lazy setting. Note that in the lazy
setting, \textsf{FS} sometimes also avoids complete
recomputation, since thunks that are created but never forced are not executed.

\paragraph*{Results: Demand All.}

Table~\ref{tab:summary-demand-all} focuses on benchmarks where
\emph{all} of the output is demanded, or when there is only a single
output value (sum and minimum).  In these cases, several patterns
emerge in the results.

First, for eager map (\secref{overview}) and eager filter, \NominalAdapton
gets modest speedup and breaks even, respectively, while
\Adapton gets slowdowns of one to two orders of magnitude.
As \secref{overview} explains, \Adapton recomputes and reallocates a
linear number of output elements for each $O(1)$ input change
(insertion, deletion or replacement). By contrast, \NominalAdapton
need not rebuild the prefix of the output lists.%

Next, the benchmarks minimum and sum use the probabilistically balanced
trees from \secref{fpn} to do an incremental fold where, in
expectation, only a logarithmic number of intermediate computations
are affected by a small change.  Due to this construction, both
\Adapton and \NominalAdapton get large speedups over from-scratch
computation (up to 4720$\times$).  \NominalAdapton tends to get
slightly more speedup, since its use of names leads to less tree
rebuilding.  This is similar to, but not as asymptotically deep as,
the eager map example ($O(\log n)$ here versus $O(n)$ above).

The next four benchmarks (reverse, mergesort, median, quickhull) show
marked contrasts between the times for \NominalAdapton and \Adapton:
In all cases, \NominalAdapton gets a speedup (from about 4$\times$ to
148$\times$), whereas \Adapton nearly always gets a slowdown.  Two
exceptions are the deletion changes that revert a prior insertion. In
these cases, \Adapton reuses the original cache information that it
duplicates (at great expense) after the insertion.
\Adapton gets no speedup for quickhull, our most complex benchmark in
this table.  By contrast, \NominalAdapton performs updates
\emph{orders of magnitude} faster than \Adapton and gets a speedup
over from-scratch; the stars~($\ast$) indicate that we ran quickhull
at one tenth of the listed input size for \Adapton, because otherwise
it used too much memory due to having large memo tables but little reuse
from them.

\paragraph*{Results: Demand One.}

Table~\ref{tab:summary-demand-one} focuses on benchmarks where
\emph{one} (of many possible) outputs are demanded.  In these cases,
two patterns emerge.
First, on simple lazy list benchmarks map and filter, \Adapton and
\NominalAdapton perform roughly the same, with \Adapton getting
slightly higher speedups than \NominalAdapton.  These cases are good fits
for \Adapton's model, and names only add overhead.

Second, on more involved list benchmarks (reverse, mergesort and
quickhull), \NominalAdapton delivers greater speedups (from 5$\times$
to 4540$\times$) than \Adapton, which often delivers slowdowns.  Two
exceptions are mergesort, where \Adapton delivers speedups, but is
still up to 6.7$\times$ slower than \NominalAdapton, and the deletion
changes, which---as in the table above---are fast because of spurious
duplication in the insertion change.

In summary, \NominalAdapton consistently delivers speedups for small
changes, while \Adapton does so to a lesser extent, and much less
reliably.

\subsection{Interpreter Experiments}

\begin{table}
\footnotesize
\begin{centering}
  \begin{tabular}{|c|c|l||r||r|r|}
  \hline
  \multicolumn{6}{|c|}{Batch-mode comparison (\emph{``demand all''})}
  \\
  \hline
  \textbf{Program} &
  \textbf{$n$} &
  \textbf{Edit} &
  \textsf{FS} (ms) &
  \textsf{A} ($\times$) &
  \textsf{NA} ($\times$)
  \\
  \hline

  fact &  5e3  & repl    & 945 & 0.520 & 10900 \\
       &       & swap1   & 947 & 2410 & 4740 \\
       &       & swap2   & 955 & 4740 & 6590 \\
       &       & ext     & 847 & 0.464 & 0.926 \\
   \hline
  intlog;fact & $2^{30}$, 5e3 & swap & 849 & 0.413 & 3.18 \\

  \hline
  array     & $2^{10}$ & repl1 & 191 & 0.323 & 6.52 \\
  max       &         & repl2 & 191 & 0.310 & 7.62 \\

  \hline
  matrix    & $20x20$ & swap1 & 4500 & 0.617 & 1.31 \\
  mult      &         & swap2 & 4500 & 0.756 & 1.17 \\
            & $25x25$ & ext  & 6100 & 1.50    & 1.55 \\

  \hline

\end{tabular}
  \caption{Speedups of IMP interpreter experiments}
\label{tab:summary-imp-interp}
\end{centering}
\end{table}

We tested the incremental behavior of the IMP interpreter with three
basic forms of edits to the input programs: replacing values ({\it
replace}), swapping subexpressions ({\it swap}), and increasing the
size of the input ({\it ext}). These experiments all take the
following form: evaluate an expression, mutate the expression, and then
reevaluate.

\begin{itemize}
\item %
For \ml{fact}, {\it repl} mutates the value of an unused variable;
{\it swap1} reverses the order of two assignments at the start of
the program; {\it swap2} reverses the order of two assignments at the
end; and {\it ext} increases the size of the input.

\item %
For \ml{intlog;fact}, {\it swap} swaps the two
subprograms.

\item %
For \ml{array max}, {\it repl1} replaces a value at the start of
the array, and {\it repl2} moves a value from the start to the end
of the array.

\item %
For \ml{matrix mult}, {\it swap1} reverses the order of the initial
assignments of the outer arrays of the input matrices; {\it swap2}
reverses the order of the while loops that initialize the inner arrays
of the input matrices; and {\it ext} extends the dimensions of
the input arrays.
\end{itemize}

\paragraph*{Results.}

Table~\ref{tab:summary-imp-interp} summarizes the results, presented
the same way as the list benchmarks. We can see that
\NominalAdapton provides a speedup over from-scratch computation in
all but one case, and can provide dramatic speedups. In
addition, \NominalAdapton consistently outperforms classic \adapton{},
in some cases providing a speedup where \Adapton incurs a (sizeable)
slowdown. 

The \ml{fact} program's {\it repl} experiment shows significant
performance improvement due to names. Classic
\adapton{} dirties each intermediary environment
and is forced to recompute. With the naming strategy outlined in
Section~\ref{sec:interpreter}, the environment is identified without
regard to the particular values inserted.  Future computations that
depend on the environment, but not the changed value in particular,
are reused. The \ml{fact} {\it swap} experiments show significant
speedup for both classic \adapton{} and
\NominalAdapton, because the trie map
representation remains unchanged regardless of order of the
assignments.

The remaining results fall into two categories. The edits made
to \ml{intlog;fact}, \ml{array max}, and \ml{matrix mult}'s {\it
 swap}s show speedups between 1.17$\times$ and 7.62$\times$ with \NominalAdapton, while
classic \adapton{} exhibits a slowdown, due to spending much of its
time creating and evaluating new nodes in the DCG\@. \NominalAdapton,
on the other hand, spends its time walking the already-present nodes
and reusing many (from 25\% to as much as 99\%) of them, with the
added benefit of far better memory performance.

The last category includes the {\it ext} tests for \ml{fact}
(increasing the input value) and \ml{matrix mult} (extending the
dimensions of the input matrices). Such changes have pervasive effects
on the rest of the computation and are a challenge to incremental
reuse. Extension for matrix multiplication shows a modest speedup over
the from-scratch running time for both nominal and classic \adapton{}.
\NominalAdapton is able to reuse a third of the nodes created during
the original run, while classic \adapton{} is not able to reuse
any. Increasing the value of the input to factorial causes similar
behavior, though the single, short loop prevents the amount of reuse
from overcoming the from-scratch time.

\section{Related Work}
\Label{sec:related}

Here we survey past approaches to incremental computation, organizing
our discussion
into three categories: static approaches, dynamic approaches, and
specialized approaches.

\paragraph{Static Approaches to IC.}
These approaches transform programs to derive a second program
that can process ``deltas''; the derived program takes as input the last
(full) output and the representation of an input change, and
produces (the representation of) the next output change.
This program derivation is performed a priori, before any dynamic
changes are issued.  As such, static approaches have the advantage of
not incurring dynamic space or time overhead, but also carry
disadvantages that stem from not being dynamic in nature: They cannot
handle programs with general recursion, and cannot take advantage of
cached intermediate results, since by design, there are
none~\citep{LiuTeitelbaum95, Cai2014}.
Other static approaches transform programs into ones that cache and
reuse past results, given a predefined class of input
changes~\citep{LiuStTe98}.  Future work should explore an empirical
comparison between these approaches and comtemporary dynamic
approaches, described below.

\paragraph{Dynamic Approaches to IC.}
In contrast to static approaches, dynamic approaches attempt to trade
space for time savings.
A variety of dynamic approaches to IC have been proposed.
Most early approaches fall into one of two camps: they
either perform function caching of pure
programs~\citep{Bellman57,McCarthy63,Michie68,PughThesis}, or they
support input mutation and employ some form of dynamic dependency
graphs.
However, the programming model advanced by earlier work on dependence
graphs lacked features like general recursion and dynamic allocation,
instead restricting programs to those expressible as \emph{attribute
  grammars} (a language of declarative constraints over tree
structures)~\citep{DemersReTe81,Reps82a,Reps82b,vogt1991}.

Some recent general-purpose approaches to dynamic IC (SAC and
\Adapton) support general-purpose input structures and general
recursion; internally, they use a notion of memoization to
find and reuse portions of existing dependency graphs.
As described in \cite{Adapton2014}, SAC and \Adapton differ greatly in the
programming model they support (SAC is eager/batch-oriented whereas
\Adapton is demand-driven) and in how they represent dependency
graphs.  Consequently, they have different performance
characteristics, with \Adapton excelling at demand-driven and interactive
settings, and SAC doing better in non-interactive, batch-oriented
settings~\citep{Adapton2014}.

The presence of dynamic memory
allocation in SAC poses a reuse problem due to ``fresh'' object
identities, and thus benefits from a mechanism to deterministically
match up identities from prior runs.
Various past work on SAC addresses this problem in some
form~\citep{AcarBlHa04,AcarBlBlHaTa06,HammerAc08,AcarLW09}
describing how to use ``hints'' or ``keys.''
The reuse problem in \NominalAdapton is more general in nature than in
SAC, and thus requires a very different solution. For example,
\NominalAdapton's DCG and more general memo tables do not impose SAC's
total ordering of events, admitting more opportunities for reuse, but
complicating the issue of assuring that names are not used more than once within a
run. The use of thunks, which also need names, adds
a further layer of complication.
This paper addresses name reuse in this (more general) IC setting.
Further, we address other naming issues, such as how to generate
new names from existing ones (via \ml{fork}) and how to determinize
memo table creation (via named namespaces).

Ongoing research in programming languages and systems continues to
generalize memoization.
\citet{Bhatotia2015} extend memoization to parallel C and C{+}{+}
programs written against a traditional UNIX threading API.
\citet{Bhatotia2011} extend memoization to distributed, cloud-based
settings (MapReduce-style computations in particular).
\citet{Chen2014} reduce the (often large) time and space overhead,
which is pervasive in both SAC and in \Adapton.
In particular, they propose coarsening the granularity of dependence
tracking, and report massive reductions (orders of magnitude)
in space as a result.  We believe that their approach
(``probabilistic chunking'') should be immediately applicable to our setting,
as well as to classic \adapton{}.
Indeed, early results for mergesort indicate up to an
order-of-magnitude reduction in overhead.

\paragraph{Specialized Approaches to IC.}

Some recent approaches to IC are not general-purpose, but exploit
domain-specific structure to handle input changes efficiently.
\textsf{DITTO} incrementally checks invariants in Java programs, but
is limited to invariant checking~\citep{DITTO2007}.
\textsf{i3QL} incrementally repairs database views (queries) when the
underlying data changes due to insertions and removals of table
rows~\citep{i3ql2014}.

Finally, reactive programming (especially functional reactive programming or FRP)
shares some
elements with incremental computation: both paradigms offer
programming models for systems that strive to efficiently react to
``outside changes''; internally, they use graph representations to
model dependencies in a program that change over time
\citep{Cooper06embeddingdynamic,DBLP:conf/icfp/KrishnaswamiB11,Czaplicki2013AFR}.
However, the chief aim of FRP is to provide a declarative means of
specifying programs whose values are time-dependent (stored in
signals), whereas the chief aim of IC is to provide time
savings for small input changes (stored in special references).
The different scope and programming model of FRP makes it
hard to imagine using it to write an efficient incremental sorting algorithm,
though it may be possible. On the other hand, IC would seem to be an
appropriate mechanism for implementing an FRP engine, though the exact
nature of this connection remains unclear.

\section{Conclusion}

This paper has presented \emph{nominal matching}, a new strategy that
general-purpose incremental computation can use to match a proposed
computation against a prior, memoized one. With nominal matching,
programmers may explicitly associate a \emph{name} with a memoized
computation, and matching is done by name equality. Nominal matching
overcomes the conservativity of \emph{structural matching}, the most
commonly employed approach, which compares computations based on their
structure and thus may fail to reuse prior results when it should
(i.e., those that are not structurally identical but require little
work to bring up-to-date). We have implemented nominal matching as
part of \NominalAdapton, an extension to the \Adapton general-purpose
system for incremental computation, and endowed it with
\emph{namespaces} for more flexible management of names in practical
programs. We have formalized \NominalAdapton's (and \Adapton's)
algorithms and proved them correct. We have implemented a variety of
data structures and benchmark programs in \NominalAdapton. Performance
experiments show that compared to \Adapton (which employs structural
matching) \NominalAdapton enjoys uniformly better performance,
sometimes achieving many orders-of-magnitude speedups over
from-scratch computation when \Adapton would suffer significant
slowdowns.

\section*{Acknowledgments}

We thank James Parker, Khoo Yit Phang, Robert Harper, as well as the
anonymous reviewers of the program committees (including the AEC) for
their questions, insights and feedback.
We also thank James Parker for his contributions to an earlier
prototype of the proposed design.
This research was supported in part by NSF CCF-1116740, NSF CCF-1319666, DARPA
FA8750-12-2-0106, and a gift from Mozilla Research.

\balance
\bibliographystyle{abbrvnat}
\bibliography{main,paper,hammer-thesis1,hammer-thesis2}

\iftr
\clearpage

\onecolumn
\appendix

\section*{Supplement to
   ``\mytitle''}

This supplementary material contains additional details about how
\NominalAdapton manages space, in Appendix~\ref{sec:spaceman}.
Our full from-scratch consistency result, along with the definitions and
lemmas it uses, appears in Appendix~\ref{apx:metatheory}.

\section{Space management}
\label{sec:spaceman}
\label{sec:spacemanagement}

In a long-running program, memo tables could grow without
bound. \NominalAdapton helps reduce table sizes, as we have already
seen, but we still need a mechanism to clean out the tables
when space becomes limited.
A natural idea is to implement a memo table's mapping from name to DCG
node with a \emph{weak reference}, so that if the table is the only
reference to the node, the node can be garbage collected when the
system is short on space.
This is not quite enough, though, because to implement dirtying, DCG
edges are bidirectional. To avoid space leaks, it is critical to also
make these backedges weak.

Unfortunately, using weak references for both memo tables and back
edges (as implemented in the original \Adapton) is generally
unsound. \Adapton supports an interactive pattern called
\emph{swapping}, wherein DCG components can be swapped in and out of
the active DCG.  Pathologically, during the time that a
sub-computation is swapped out, the garbage collector could remove
\emph{some} of this DCG structure, but not all of it. In particular,
it could null out some of its weak back edges, because these nodes are
only reachable by weak references. But if this swapped out
sub-computation is later swapped back into the DCG, these (weak) back
edges will be gone, and we will potentially fail to dirty nodes that
ought to be dirty, as future changes occur.

To fix the GC problem, we still use weak references for back edges,
but use strong references for memo table entries, so that from the
GC's point of view, all DCG nodes are always reachable. To implement
safe space reclamation, we also implement reference counting of DCG
nodes, where the counts reflect the number of strong edges reaching a
node. When DCG edges are deleted, the reference counts of target nodes
are decremented. Nodes that reach zero are not immediately collected;
this allows thunks to be ``resurrected'' by the swapping
pattern. Instead, we provide a \cod{flush} operation for memo tables
that deletes the strong mapping edge for all nodes with a count of
zero, which means they are no longer reachable by the main
program. Deletion is transitive: removing the node decrements the
counts of nodes it points to, which may cause them to be deleted.

An interesting question is how to decide when to invoke \cod{flush};
this is the system's \emph{eviction policy}. One obvious choice is to
flush when the system starts to run short of memory (based on a signal
from the GC), which matches the intended effect of the unsound weak
reference-based approach. But placing the eviction policy under the
program's control opens other possibilities, e.g., the programmer
could invoke \ml{flush} when it is semantically clear that thunks
cannot be restored.
We leave to future work a further exploration of sensible
eviction policies.

\section{Metatheory of \calc}
\Label{apx:metatheory}

\subsection{Overview}

Our main formal result in this paper, \emph{from-scratch consistency},
states that given an evaluation derivation
corresponding to an incremental computation, we can construct a derivation
corresponding to a non-incremental computation that returns the same result
and a corresponding graph.  That is, the incremental computation is consistent
with a computation in the simpler non-incremental system.

To properly state the consistency result, we need to define what it means for a graph to
be well-formed (\Sectionref{sec:graph-wf}),
relate the incremental computation's DCG
to the non-incremental computation's store (\Sectionref{sec:restriction}),
describe the sets of nodes read and written by an evaluation derivation
(\Sectionref{sec:rw-sets}),
and prove a store weakening lemma~(\Sectionref{sec:weakening}).
The consistency result itself is stated and proved in \Sectionref{sec:fsc}.

\subsection{Graph Well-Formedness}
\Label{sec:graph-wf}

\begin{figure*}
\small
\ottdefnsGrwf
\vspace{-3ex}
\caption{Graph well-formedness rules}
\FLabel{fig:graph-wf}
\end{figure*}

The judgment $ G   \vdash   H \; {\txtsf{wf} } $ is read ``$H$ is a well-formed subset of $G$''.
It implies that $H$ is a linearization of a subset of $G$: within $H$,
information sources appear to the left of information sinks,
dependency edges point to the left, and information flows to the right.
Consequently, $H$ is a dag.

For brevity, we write $ G \; {\txtsf{wf} } $ for $ G   \vdash   G \; {\txtsf{wf} } $, threading
the entire graph $G$ through the rules deriving $ G   \vdash   H \; {\txtsf{wf} } $.
The well-formedness rules in \Figureref{fig:graph-wf} decompose the right-hand
graph, and work as follows:

\begin{itemize}
\item
  For values (Grwf-val) and thunks with no cached result (Grwf-thunk),
  the rules only check the correspondence between $G$ (the entire graph)
  and $H$ (the subgraph).

\item For thunks with a cached result, Grwf-thunkCache
  examines the outgoing edges.  If they are all clean, then it checks
  that evaluating $e$ \emph{again}
  would not change the graph at all: $\D \derives  G   \vdash ^{ p }_{ \omega }  e   \Downarrow   G  ;  \mathrm{t} $.
  If one or more edges are dirty, it checks that all incoming edges are dirty.

\item
  Grwf-dirtyEdge checks that
  all edges flowing into a dirty edge are dirty:
  given edges from $p$ to $q$ and from $q$ to $r$, where
  the edge from $q$ to $r$ is dirty (meaning that $q$ depends on $r$ and $r$ needs
  to be recomputed), the edge from $p$ to $q$ should
  also be dirty.  Otherwise, we would think we could reuse the result in $p$,
  even though $p$ (transitively) depends on $r$.

  Conversely, if an edge from $p$ to $q$ is clean (Grwf-cleanEdge),
  then all edges out from $q$ must
  be clean (otherwise we would contradict the ``transitive dirtiness'' just described).
  Moreover, the action $a$ stored in the edge from $p$ to $q$ must be consistent
  with the contents of $q$.
\end{itemize}

\subsection{From Graphs to Stores: the Restriction Function}
\Label{sec:restriction}

To relate the graph associated with an incremental evaluation to
the store associated with a non-incremental evaluation, we define
a function~$ \lfloor  G  \rfloor_{ P } $ that restricts $G$ to a set of pointers
$P$, drops cached thunk results (*), and erases all edges (**):
\begin{defn}[Restriction]
\Label{def:restrict}
\[
  \begin{array}{cr@{~~}c@{~~}lll}
    & \lfloor   \varepsilon   \rfloor_{ P }  &=&  \varepsilon 
    \\[1ex]
    & \lfloor  G  \ottsym{,}   p {:} v   \rfloor_{ P }  &=&  \lfloor  G  \rfloor_{ P }   \ottsym{,}   p {:} v  & \text{if $p \in P$}
    \\
    & \lfloor  G  \ottsym{,}   p {:} v   \rfloor_{ P }  &=&  \lfloor  G  \rfloor_{ P }  & \text{if $p \notin P$}
    \\[1ex]
    & \lfloor  G  \ottsym{,}   p {:} e   \rfloor_{ P }  &=&  \lfloor  G  \rfloor_{ P }   \ottsym{,}   p {:} e  & \text{if $p \in P$}
    \\
    & \lfloor  G  \ottsym{,}   p {:} e   \rfloor_{ P }  &=&  \lfloor  G  \rfloor_{ P }  & \text{if $p \notin P$}
    \\
    \text{(*)}& \lfloor  G  \ottsym{,}   p {:}( e , \mathrm{t} )   \rfloor_{ P }  &=&  \lfloor  G  \rfloor_{ P }   \ottsym{,}   p {:} e  & \text{if $p \in P$}
    \\
    & \lfloor  G  \ottsym{,}   p {:}( e , \mathrm{t} )   \rfloor_{ P }  &=&  \lfloor  G  \rfloor_{ P }  & \text{if $p \notin P$}
    \\[1ex]
    \text{(**)}& \lfloor  G  \ottsym{,}  \ottsym{(}  p  \ottsym{,}  a  \ottsym{,}  b  \ottsym{,}  q  \ottsym{)}  \rfloor_{ P }  &=&  \lfloor  G  \rfloor_{ P } 
  \end{array}
\]
\end{defn}

\subsection{Read and Write Sets}
\Label{sec:rw-sets}

\paragraph{Join and Merge Operations.}
To specify the read and write sets, we
use a \emph{separating join} $ H_{{\mathrm{1}}}  \joinsym  H_{{\mathrm{2}}} $ on graphs:
$ H_{{\mathrm{1}}}  \joinsym  H_{{\mathrm{2}}}   \ottsym{=}  \ottsym{(}  H_{{\mathrm{1}}}  \ottsym{,}  H_{{\mathrm{2}}}  \ottsym{)}$ if $  \txtsf{dom}( H_{{\mathrm{1}}} )   \mathrel{\bot}   \txtsf{dom}( H_{{\mathrm{2}}} )  $,
and undefined otherwise.

We also define a \emph{merge} $ H_{{\mathrm{1}}}  \mergesym  H_{{\mathrm{2}}} $ that \emph{is} defined for
subgraphs with overlapping domains, provided $H_1$ and $H_2$ are consistent
with each other.  That is, if $ p  \in   \txtsf{dom}( H_{{\mathrm{1}}} )  $ and $ p  \in   \txtsf{dom}( H_{{\mathrm{2}}} )  $, then
$H_1(p) = H_2(p)$.

\begin{figure*}[thbp]
  \centering

    \newcommand{\SEP}{\\[1ex]}
    \begin{array}[t]{rl@{~~}lll}
       \D  \txtsf{~by~}    \txtsf{Eval-term}   (  )   \txtsf{~reads~}   \varepsilon   \txtsf{~writes~}   \varepsilon  
    \SEP
       \D  \txtsf{~by~}    \txtsf{Eval-app}   ( \D_{{\mathrm{1}}}  \ottsym{,}  \D_{{\mathrm{2}}} )   \txtsf{~reads~}   R_{{\mathrm{1}}}  \mergesym  \ottsym{(}  R_{{\mathrm{2}}}  \ottsym{-}  W_{{\mathrm{1}}}  \ottsym{)}   \txtsf{~writes~}   W_{{\mathrm{1}}}  \joinsym  W_{{\mathrm{2}}}  
      &\text{if}&
       \D_{{\mathrm{1}}}  \txtsf{~reads~}  R_{{\mathrm{1}}}  \txtsf{~writes~}  W_{{\mathrm{1}}} 
      \\ &\text{and}&
       \D_{{\mathrm{2}}}  \txtsf{~reads~}  R_{{\mathrm{2}}}  \txtsf{~writes~}  W_{{\mathrm{2}}} 
    \SEP
       \D  \txtsf{~by~}    \txtsf{Eval-bind}   ( \D_{{\mathrm{1}}}  \ottsym{,}  \D_{{\mathrm{2}}} )   \txtsf{~reads~}   R_{{\mathrm{1}}}  \mergesym  \ottsym{(}  R_{{\mathrm{2}}}  \ottsym{-}  W_{{\mathrm{1}}}  \ottsym{)}   \txtsf{~writes~}   W_{{\mathrm{1}}}  \joinsym  W_{{\mathrm{2}}}  
      &\text{if}&
       \D_{{\mathrm{1}}}  \txtsf{~reads~}  R_{{\mathrm{1}}}  \txtsf{~writes~}  W_{{\mathrm{1}}} 
      \\ &\text{and}&
       \D_{{\mathrm{2}}}  \txtsf{~reads~}  R_{{\mathrm{2}}}  \txtsf{~writes~}  W_{{\mathrm{2}}} 
    \SEP
       \D  \txtsf{~by~}    \txtsf{Eval-nest}   ( \D_{{\mathrm{1}}}  \ottsym{,}  \D_{{\mathrm{2}}} )   \txtsf{~reads~}   R_{{\mathrm{1}}}  \mergesym  \ottsym{(}  R_{{\mathrm{2}}}  \ottsym{-}  W_{{\mathrm{1}}}  \ottsym{)}   \txtsf{~writes~}   W_{{\mathrm{1}}}  \joinsym  W_{{\mathrm{2}}}  
      &\text{if}&
       \D_{{\mathrm{1}}}  \txtsf{~reads~}  R_{{\mathrm{1}}}  \txtsf{~writes~}  W_{{\mathrm{1}}} 
      \\ &\text{and}&
       \D_{{\mathrm{2}}}  \txtsf{~reads~}  R_{{\mathrm{2}}}  \txtsf{~writes~}  W_{{\mathrm{2}}} 
    \SEP
       \D  \txtsf{~by~}    \txtsf{Eval-computeDep}   ( \D_{{\mathrm{1}}}  \ottsym{,}  \D_{{\mathrm{2}}} )   \txtsf{~reads~}   R_{{\mathrm{1}}}  \mergesym  R_{{\mathrm{2}}}   \txtsf{~writes~}  W_{{\mathrm{2}}} 
          &\text{if}&
           \D_{{\mathrm{1}}}  \txtsf{~reads~}  R_{{\mathrm{1}}}  \txtsf{~writes~}  W_{{\mathrm{1}}} 
          \\ &\text{and}&
           \D_{{\mathrm{2}}}  \txtsf{~reads~}  R_{{\mathrm{2}}}  \txtsf{~writes~}  W_{{\mathrm{2}}} 
          \\ &\text{and}&
            \txtsf{dom}( W_{{\mathrm{1}}} )   \subseteq   \txtsf{dom}( W_{{\mathrm{2}}} )  
    \end{array}

    \medskip

    \begin{array}[t]{rl@{~~}lll}
       \D  \txtsf{~by~}    \txtsf{Eval-fix}   ( \D_{{\mathrm{0}}} )   \txtsf{~reads~}  R  \txtsf{~writes~}  W 
      &\text{if}&  \D_{{\mathrm{0}}}  \txtsf{~reads~}  R  \txtsf{~writes~}  W 
    \\
       \D  \txtsf{~by~}    \txtsf{Eval-case}   ( \D_{{\mathrm{0}}} )   \txtsf{~reads~}  R  \txtsf{~writes~}  W 
      &\text{if}&  \D_{{\mathrm{0}}}  \txtsf{~reads~}  R  \txtsf{~writes~}  W 
    \\
       \D  \txtsf{~by~}    \txtsf{Eval-split}   ( \D_{{\mathrm{0}}} )   \txtsf{~reads~}  R  \txtsf{~writes~}  W 
      &\text{if}&  \D_{{\mathrm{0}}}  \txtsf{~reads~}  R  \txtsf{~writes~}  W 
    \\
       \D  \txtsf{~by~}    \txtsf{Eval-namespace}   ( \D_{{\mathrm{0}}} )   \txtsf{~reads~}  R  \txtsf{~writes~}  W 
      &\text{if}&  \D_{{\mathrm{0}}}  \txtsf{~reads~}  R  \txtsf{~writes~}  W 
    \SEP
       \D  \txtsf{~by~}    \txtsf{Eval-fork}   (  )   \txtsf{~reads~}   \varepsilon   \txtsf{~writes~}   \varepsilon  
    \SEP
       \D  \txtsf{~by~}    \txtsf{Eval-refDirty}   (  )   \txtsf{~reads~}   \varepsilon   \txtsf{~writes~}   q {:} v  
      &\text{where}&
      e  \ottsym{=}   \textbf{ref}(\Grn{  \Grn{ \ottkw{nm} \, k }  }, v )  \AND q  \ottsym{=}   k  @  \omega 
    \SEP
       \D  \txtsf{~by~}    \txtsf{Eval-refClean}   (  )   \txtsf{~reads~}   \varepsilon   \txtsf{~writes~}   q {:} v  
      &\text{where}&
      e  \ottsym{=}   \textbf{ref}(\Grn{  \Grn{ \ottkw{nm} \, k }  }, v )  \AND q  \ottsym{=}   k  @  \omega 
    \SEP
       \D  \txtsf{~by~}    \txtsf{Eval-thunkDirty}   (  )   \txtsf{~reads~}   \varepsilon   \txtsf{~writes~}   q {:} e_{{\mathrm{0}}}  
      &\text{where}&
      e  \ottsym{=}   \textbf{thunk}(\Grn{  \Grn{ \ottkw{nm} \, k }  }, e_{{\mathrm{0}}} ) 
    \SEP
       \D  \txtsf{~by~}    \txtsf{Eval-thunkClean}   (  )   \txtsf{~reads~}   \varepsilon   \txtsf{~writes~}   q {:} G ( q )  
      &\text{where}&
      e  \ottsym{=}   \textbf{thunk}(\Grn{  \Grn{ \ottkw{nm} \, k }  }, e_{{\mathrm{0}}} ) 
    \SEP
       \D  \txtsf{~by~}    \txtsf{Eval-getClean}   (  )   \txtsf{~reads~}   q {:} v   \txtsf{~writes~}   \varepsilon  
      &\text{where}&
      e  \ottsym{=}  \ottkw{get} \, \ottsym{(}   \ottkw{ref} \,\Grn{ q }   \ottsym{)} \AND q  \ottsym{=}   k  @  \omega  \AND  G ( q )   \ottsym{=}  v
    \SEP
       \D  \txtsf{~by~}    \txtsf{Eval-forceClean}   (  )   \txtsf{~reads~}  R'  \ottsym{,}   q {:}( e , \mathrm{t} )   \txtsf{~writes~}  W' 
      &\text{where}&
      e  \ottsym{=}  \ottkw{force} \, \ottsym{(}   \ottkw{thk} \,\Grn{ q }   \ottsym{)}
      \\ &\text{and}&
       \D'  \txtsf{~reads~}  R'  \txtsf{~writes~}  W' 
      \\ &\text{where}&
      \text{$\D'$ is the derivation that computed $t$ (see text)}
    \SEP
       \D  \txtsf{~by~}    \txtsf{Eval-scrubEdge}   ( \D_{{\mathrm{0}}} )   \txtsf{~reads~}  R  \txtsf{~writes~}  W 
      &\text{if}&  \D_{{\mathrm{0}}}  \txtsf{~reads~}  R  \txtsf{~writes~}  W 
    \end{array}

  \caption{Read- and write-sets of a derivation}
  \FLabel{fig:readswrites}
\end{figure*}

\begin{defn}[Reads/writes]
\Label{def:rw}
  The effect of an evaluation derived by $\D$,
  written $ \D  \txtsf{~reads~}  R  \txtsf{~writes~}  W $,
  is defined in \Figureref{fig:readswrites}.
\end{defn}

This is a function over derivations.
We write ``$ \D  \txtsf{~by~}   \mathcal{R}  ( \vec{\D} )   \txtsf{~reads~}  R  \txtsf{~writes~}  W $''
to mean that rule $\mathcal{R}$ concludes $\D$ and has
subderivations $\vec{\D}$.  For example, $ \D  \txtsf{~by~}    \txtsf{Eval-fix}   ( \D_{{\mathrm{0}}} )   \txtsf{~reads~}  R  \txtsf{~writes~}  W $
provided that $ \D_{{\mathrm{0}}}  \txtsf{~reads~}  R  \txtsf{~writes~}  W $ where $\D_0$ derives the only premise of
Eval-fix.

In the Eval-forceClean case, we refer back to the derivation that (most recently) computed
the thunk being forced (that is, the first subderivation of Eval-computeDep).  A completely formal
definition would take as input a mapping from pointers $q$ to sets $R'$ and $W'$,
return this mapping as output, and modify the mapping in the Eval-computeDep case.

\paragraph{Agreement.}  We want to express a result
 (\Lemmaref{lem:respect-write}) that evaluation only affects pointers in the write set
$W$, leaving the contents of other pointers alone, so we define what it means
to leave pointers alone:

\begin{defn}[Agreement on a pointer]
  Graphs $G_{{\mathrm{1}}}$ and $G_{{\mathrm{2}}}$ \emph{agree on}
  $p$ iff either $ \txtsf{exp}( G_{{\mathrm{1}}} ,  p )   \ottsym{=}   \txtsf{exp}( G_{{\mathrm{2}}} ,  p ) $,
  or $ p {:} v  \in G_1$ and $ p {:} v  \in G_2$.
\end{defn}

\begin{defn}[Agreement on a set of pointers]
  Graphs $G_{{\mathrm{1}}}$ and $G_{{\mathrm{2}}}$ \emph{agree on}
  a set $P$ of pointers iff $G_{{\mathrm{1}}}$ and $G_{{\mathrm{2}}}$
  agree on each $p \in P$.
\end{defn}

\begin{lem}[Respect for write-set]
\Label{lem:respect-write}
~\\
   If $ \D  \derives   G   \vdash ^{ p }_{ \omega }  e   \Downarrow   G'  ;  \mathrm{t}  $ where $\D$ is incremental
   and $ \D  \txtsf{~reads~}  R  \txtsf{~writes~}  W $
   \\
   then $G'$ agrees with $G$ on $ \txtsf{dom}( G )   \ottsym{-}   \txtsf{dom}( W ) $.
\end{lem}
\begin{proof}
  By a straightforward induction on $\D$, referring to \Definitionref{def:rw}.

  In the Eval-forceClean case, use the fact that $G'$ differs
  from $G$ only in the addition of an edge, which does not affect agreement.
\end{proof}

\subsection{Satisfactory derivations}
\Label{sec:satisfactory-derivations}

In our main result~(\Theoremref{thm:fsc}), we will assume that all input
and output graphs appearing within a derivation are well-formed,
and that the read- and write-sets are defined:

\begin{defn}[Locally satisfactory]
\Label{def:local-sat}
~\\
     A derivation
     $\D \derives  G_{{\mathrm{1}}}   \vdash ^{ p }_{ \omega }  e   \Downarrow   G_{{\mathrm{2}}}  ;  \mathrm{t} $
     is \emph{locally satisfactory}
     if and only if

     \begin{enumerate}[(1)]
    \item  %
           $ G_{{\mathrm{1}}} \; {\txtsf{wf} } $ and $ G_{{\mathrm{2}}} \; {\txtsf{wf} } $

     \item  %
           $ \D  \txtsf{~reads~}  R  \txtsf{~writes~}  W $ is defined

     \end{enumerate}
\end{defn}

\begin{defn}[Globally satisfactory]
\Label{def:global-sat}
~\\
    An evaluation derivation
    $\D \derives  H_{{\mathrm{1}}}   \vdash ^{ p }_{ \omega }  e   \Downarrow   H_{{\mathrm{2}}}  ;  \mathrm{t} $
    is \emph{globally satisfactory}, written~$\D \satisfactory$,
    if and only if
    $\D$ is locally satisfactory
    and all its subderivations are locally satisfactory.
\end{defn}

\subsection{Weakening}
\Label{sec:weakening}

The main result needs to construct a non-incremental derivation
in a different order from the given incremental derivation.
In particular, the first evaluation $\D_1$ done in Eval-computeDep
will be done ``later'' in the non-incremental derivation.
Since it
is done later, the graph may have new material $G'$,
and we need a weakening lemma to move from a non-incremental
evaluation of $\D_1$ (obtained through the induction hypothesis)
to a non-incremental evaluation over the larger graph.

\begin{lem}[Weakening (non-incremental)]
\Label{lem:weakening-ni}
~\\
   If $ \D  \derives   G_{{\mathrm{1}}}   \vdash ^{ p }_{ \omega }  e   \Downarrow   G_{{\mathrm{2}}}  ;  \mathrm{t}  $
   and $\D$ is non-incremental
   \\
   and $G'$ is disjoint from $G_{{\mathrm{1}}}$ and $G_{{\mathrm{2}}}$
   \\
   then $ \D'  \derives   G_{{\mathrm{1}}}  \ottsym{,}  G'   \vdash ^{ p }_{ \omega }  e   \Downarrow   G_{{\mathrm{2}}}  \ottsym{,}  G'  ;  \mathrm{t}  $
   where $\D'$ is non-incremental.
\end{lem}
\begin{proof}
  By induction on $\D$.

  \begin{itemize}
      \ProofCasesRules{Eval-term, Eval-app, Eval-fix, Eval-bind,
                       Eval-case, Eval-split, Eval-fork, Eval-namespace, Eval-nest}

          These rules do not manipulate the graph, so
          just use the i.h.\ on each subderivation, then apply the same rule.

      \ProofCaseRule{Eval-refPlain}
           Since $ G_{{\mathrm{1}}} \{ q {\mapsto} v \}   \ottsym{=}  G_{{\mathrm{2}}}$,
           we have $ q  \in   \txtsf{dom}( G_{{\mathrm{2}}} )  $.
           \\ It is given that $  \txtsf{dom}( G' )   \mathrel{\bot}   \txtsf{dom}( G_{{\mathrm{2}}} )  $.
           Therefore $q \notin  \txtsf{dom}( G_{{\mathrm{1}}}  \ottsym{,}  G' ) $.
           \\ By definition, $ \ottsym{(}  G_{{\mathrm{1}}}  \ottsym{,}  G'  \ottsym{)} \{ q {\mapsto} v \}   \ottsym{=}  \ottsym{(}  G_{{\mathrm{2}}}  \ottsym{,}  G'  \ottsym{)}$.
           \\ Apply Eval-refPlain.

      \ProofCaseRule{Eval-thunkPlain}
           Similar to the Eval-refPlain case.

      \ProofCaseRule{Eval-forcePlain}
         We have $ \txtsf{exp}( G_{{\mathrm{1}}} ,  q )   \ottsym{=}  e_{{\mathrm{0}}}$.  Therefore $ \txtsf{exp}( G_{{\mathrm{1}}}  \ottsym{,}  G' ,  q )   \ottsym{=}  e_{{\mathrm{0}}}$.
         Use the i.h.\ and apply Eval-forcePlain.

      \ProofCaseRule{Eval-getPlain}
         Similar to the Eval-forcePlain case.
  \qedhere
  \end{itemize}
\end{proof}

\subsection{Main result: From-scratch consistency}
\Label{sec:fsc}

At the highest level, the main result (\Theoremref{thm:fsc}) says:

\begin{quote}
\textbf{First approximation} \\
  If $ H_{{\mathrm{1}}}   \vdash ^{ p }_{ \omega }  e   \Downarrow   H_{{\mathrm{2}}}  ;  \mathrm{t} $ by an incremental derivation,
  \\
  then $ H'_{{\mathrm{1}}}   \vdash ^{ p }_{ \omega }  e   \Downarrow   H'_{{\mathrm{2}}}  ;  \mathrm{t} $,
  where $H'_{{\mathrm{1}}}$ is a non-incremental version of $H_{{\mathrm{1}}}$
  and $H'_{{\mathrm{2}}}$ is a non-incremental version of $H_{{\mathrm{2}}}$.
\end{quote}
Using the restriction function from \Sectionref{sec:restriction},
we can refine this statement:

\begin{quote}
\textbf{Second approximation} \\
  If  $ H_{{\mathrm{1}}}   \vdash ^{ p }_{ \omega }  e   \Downarrow   H_{{\mathrm{2}}}  ;  \mathrm{t} $ by an incremental derivation
  and $ P_{{\mathrm{1}}}  \subseteq   \txtsf{dom}( H_{{\mathrm{1}}} )  $,
  \\
  then $  \lfloor  H_{{\mathrm{1}}}  \rfloor_{ P_{{\mathrm{1}}} }    \vdash ^{ p }_{ \omega }  e   \Downarrow    \lfloor  H_{{\mathrm{2}}}  \rfloor_{ P_{{\mathrm{2}}} }   ;  \mathrm{t} $
  by a non-incremental derivation,
  for some $P_{{\mathrm{2}}}$ such that  $ P_{{\mathrm{1}}}  \subseteq  P_{{\mathrm{2}}} $.
\end{quote}
Here, the pointer set $P_1$ gives the scope of the non-incremental
input graph $ \lfloor  H_{{\mathrm{1}}}  \rfloor_{ P_{{\mathrm{1}}} } $, and we construct $P_2$
describing the non-incremental output graph
$ \lfloor  H_{{\mathrm{2}}}  \rfloor_{ P_{{\mathrm{2}}} } $.

We further refine this statement by involving the derivation's read- and
write-sets: the read set $R$ must be contained in $P_1$, the write set $W$
must be disjoint from $P_1$ (written $  \txtsf{dom}( W )   \mathrel{\bot}  P_{{\mathrm{1}}} $), and $P_2$ must be
exactly $P_1$ plus $W$.
In the non-incremental semantics, the store should grow monotonically, so
we will also show $  \lfloor  H_{{\mathrm{1}}}  \rfloor_{ P_{{\mathrm{1}}} }   \subseteq   \lfloor  H_{{\mathrm{2}}}  \rfloor_{ P_{{\mathrm{2}}} }  $:

\begin{quote}
\textbf{Third approximation} \\
  If  $ H_{{\mathrm{1}}}   \vdash ^{ p }_{ \omega }  e   \Downarrow   H_{{\mathrm{2}}}  ;  \mathrm{t} $ by an incremental derivation $\D$
  with $ \D  \txtsf{~reads~}  R  \txtsf{~writes~}  W $,
  \\
  and $ P_{{\mathrm{1}}}  \subseteq   \txtsf{dom}( H_{{\mathrm{1}}} )  $
  such that $  \txtsf{dom}( R )   \subseteq  P_{{\mathrm{1}}} $
  and $  \txtsf{dom}( W )   \mathrel{\bot}  P_{{\mathrm{1}}} $
  \\
  then $  \lfloor  H_{{\mathrm{1}}}  \rfloor_{ P_{{\mathrm{1}}} }    \vdash ^{ p }_{ \omega }  e   \Downarrow    \lfloor  H_{{\mathrm{2}}}  \rfloor_{ P_{{\mathrm{2}}} }   ;  \mathrm{t} $
  \\
  where $  \lfloor  H_{{\mathrm{1}}}  \rfloor_{ P_{{\mathrm{1}}} }   \subseteq   \lfloor  H_{{\mathrm{2}}}  \rfloor_{ P_{{\mathrm{2}}} }  $
  and $P_{{\mathrm{2}}}  \ottsym{=}  P_{{\mathrm{1}}} \, \cup \,  \txtsf{dom}( W ) $.
\end{quote}
Even this refinement is not quite enough, because the incremental
system can perform computations in a different order than the non-incremental
system.  Specifically, the Eval-computeDep rule carries out a subcomputation
first, then continues with a larger computation that depends on the subcomputation.
The subcomputation does not fit into the non-incremental derivation at that point;
non-incrementally, the subcomputation is performed when it is demanded by
the larger computation.  Thus, we can't just apply the induction hypothesis on the
subcomputation.

However, the subcomputation is ``saved'' in its (incremental) output graph.
So we incorporate an invariant that all thunks with cached results in the graph
``are consistent'', that is, they satisfy a property similar to the overall consistency
result.  When this is the case, we say that the graph is from-scratch consistent.
The main result, then, will assume that the input graph is from-scratch consistent,
and show that the output graph remains from-scratch consistent.
Since the graph can grow in the interval between the subcomputation of Eval-computeDep
and the point where the subcomputation is demanded, we require that the output
graph $G_{{\mathrm{2}}}$ of the saved derivation does not contradict the larger, newer graph
$H$.  This is part (4) in the next definition.  (We number the parts from (i) to (ii)
and then from (1) so that they mostly match similar parts in the main result.)

\begin{defn}[From-scratch consistency of a derivation]
\Label{def:fsc}
~\\
  A derivation $\D_i \derives  G_{{\mathrm{1}}}   \vdash ^{ p }_{ \omega }  e   \Downarrow   G_{{\mathrm{2}}}  ;  \mathrm{t} $
  is
  \\
  \emph{from-scratch consistent for}
  $ P_q  \subseteq   \txtsf{dom}( G_{{\mathrm{1}}} )  $
  \emph{up to}
  $H$  %
  if and only if
  \begin{enumerate}[(i)]
  \item %
       $\D_i \satisfactory$ where $ \D_{\ottmv{i}}  \txtsf{~reads~}  R  \txtsf{~writes~}  W $
  \item %
     $  \txtsf{dom}( R )   \subseteq  P_q $
     and
     $  \txtsf{dom}( W )   \mathrel{\bot}  P_q $

      \item[(1)] %
           there exists a non-incremental derivation
           $\D_{ni}
            \derives
              \lfloor  G_{{\mathrm{1}}}  \rfloor_{ P_q }    \vdash ^{ p }_{ \omega }  e   \Downarrow    \lfloor  G_{{\mathrm{2}}}  \rfloor_{ P_{{\mathrm{2}}} }   ;  \mathrm{t} $

      \item[(2)] %
            $  \lfloor  G_{{\mathrm{1}}}  \rfloor_{ P_q }   \subseteq   \lfloor  G_{{\mathrm{2}}}  \rfloor_{ P_{{\mathrm{2}}} }  $

      \item[(3)] %
           $P_{{\mathrm{2}}}  \ottsym{=}  P_q \, \cup \,  \txtsf{dom}( W ) $

      \item[(4)] %
            $  \lfloor  G_{{\mathrm{2}}}  \rfloor_{ P_{{\mathrm{2}}} }   \subseteq   \lfloor  H  \rfloor_{  \txtsf{dom}( H )  }  $
  \end{enumerate}
\end{defn}

\begin{defn}[From-scratch consistency of graphs]
\Label{def:fsc-graph}
~\\
A graph $H$ is
\emph{from-scratch consistent}
\\
if, for all $q \in  \txtsf{dom}( H ) $ such that $ H ( q ) = ( e , \mathrm{t} ) $,
     \\
     there exists
     $\D_q \derives  G_{{\mathrm{1}}}   \vdash ^{ p }_{ \omega }  e   \Downarrow   G_{{\mathrm{2}}}  ;  \mathrm{t} $
     \\
     and
     $ P_q  \subseteq   \txtsf{dom}( G_{{\mathrm{1}}} )  $
     \\
     such that $\D_q$ is from-scratch consistent (\Definitionref{def:fsc})
     for $P_q$ up to $H$.
\end{defn}

The proof of the main result must maintain that the graph is from-scratch
consistent as the graph becomes larger, for which \Lemmaref{lem:fsc-ext-graph}
is useful.

\begin{lem}[Consistent extension]
\Label{lem:fsc-ext}
~\\
  If  $\D_q$ is from-scratch consistent for $P_q$ up to $H$
  \\
  and $H \subseteq H'$
  \\
  then $\D_q$ is from-scratch consistent for $P_q$ up to $H'$.
\end{lem}
\begin{proof}
  Only part (4) of \Definitionref{def:fsc} involves the ``up to'' part
  of from-scratch consistency,
  so we already have (i)--(iii) and (1)--(3).

  We have (4) $  \lfloor  G_{{\mathrm{2}}}  \rfloor_{ P_{{\mathrm{2}}} }   \subseteq   \lfloor  H  \rfloor_{  \txtsf{dom}( H )  }  $.
  Using our assumptions, $  \lfloor  H  \rfloor_{  \txtsf{dom}( H )  }   \subseteq   \lfloor  H'  \rfloor_{  \txtsf{dom}( H' )  }  $.

  Therefore (4) $  \lfloor  G_{{\mathrm{2}}}  \rfloor_{ P_{{\mathrm{2}}} }   \subseteq   \lfloor  H'  \rfloor_{  \txtsf{dom}( H' )  }  $.
\end{proof}

\begin{lem}[Consistent graph extension]
\Label{lem:fsc-ext-graph}
~\\
   If $H$ is from-scratch consistent
   \\
   and $H \subseteq H'$
   \\
   and, for all $q \in  \txtsf{dom}( H' )   \ottsym{-}   \txtsf{dom}( H ) $ such that $H'(q) = (e,t)$,
   \\
   $~~$  there exists
     $\D_q \derives  G_{{\mathrm{1}}}   \vdash ^{ p }_{ \omega }  e   \Downarrow   G_{{\mathrm{2}}}  ;  \mathrm{t} $
     and
     $ P_q  \subseteq   \txtsf{dom}( G_{{\mathrm{1}}} )  $
     \\
     $~~$
     such that $\D_q$ is from-scratch consistent (\Definitionref{def:fsc})
     \\
     $~~$
     for $P_q$ up to $H'$,
   \\
   then
   $H'$ is from-scratch consistent.
\end{lem}
\begin{proof}
  Use \Lemmaref{lem:fsc-ext} on each ``old'' pointer with a cached result in $H$,
  then apply the definitions for each ``new'' pointer with a cached result in
  $ \txtsf{dom}( H' )   \ottsym{-}   \txtsf{dom}( H ) $.
\end{proof}

At last, we can state and prove the main result, which corresponds to
the ``third approximation'' above, plus the invariant that the graph is
from-scratch consistent (parts (iii) and (4)).

We present most of the proof in a line-by-line style, with the judgment or
proposition being derived in the left column, and its justification in the right column.
In each case, we need to show four different things (1)--(4), some of which are obtained
midway through the case, so we highlight these with ``(1) \Hand'', and so on.

\begin{thm}[From-scratch consistency]
\Label{thm:fsc}
~\\
  Given an incremental $\D_i \derives  H_{{\mathrm{1}}}   \vdash ^{ p }_{ \omega }  e   \Downarrow   H_{{\mathrm{2}}}  ;  \mathrm{t} $ where
  \begin{enumerate}[(i)]
  \item %
       $\D_i \satisfactory$ where $ \D  \txtsf{~reads~}  R  \txtsf{~writes~}  W $
  \item %
     a set of pointers $ P_{{\mathrm{1}}}  \subseteq   \txtsf{dom}( H_{{\mathrm{1}}} )  $ is such that $  \txtsf{dom}( R )   \subseteq  P_{{\mathrm{1}}} $ and $  \txtsf{dom}( W )   \mathrel{\bot}  P_{{\mathrm{1}}} $
  \item %
     $H_{{\mathrm{1}}}$ is from-scratch consistent (\Definitionref{def:fsc-graph})
  \end{enumerate}
  then
  \begin{enumerate}[(1)]
  \item %
    there exists a non-incremental
      $\D_{ni} \derives   \lfloor  H_{{\mathrm{1}}}  \rfloor_{ P_{{\mathrm{1}}} }    \vdash ^{ p }_{ \omega }  e   \Downarrow    \lfloor  H_{{\mathrm{2}}}  \rfloor_{ P_{{\mathrm{2}}} }   ;  \mathrm{t} $
  \item %
      $  \lfloor  H_{{\mathrm{1}}}  \rfloor_{ P_{{\mathrm{1}}} }   \subseteq   \lfloor  H_{{\mathrm{2}}}  \rfloor_{ P_{{\mathrm{2}}} }  $
  \item %
       $P_{{\mathrm{2}}}  \ottsym{=}  P_{{\mathrm{1}}} \, \cup \,  \txtsf{dom}( W ) $
  \item %
      $H_{{\mathrm{2}}}$ is from-scratch consistent (\Definitionref{def:fsc-graph}).
  \end{enumerate}
\end{thm}
\begin{proof}
  By induction on $\D_i \derives  H_{{\mathrm{1}}}   \vdash ^{ p }_{ \omega }  e   \Downarrow   H_{{\mathrm{2}}}  ;  \mathrm{t} $.

  \begin{itemize}
    \ProofCaseRule{Eval-term}
       By \Definitionref{def:rw}, $W  \ottsym{=}   \varepsilon $,
       so let $P_2 = P_1$.

       \begin{enumerate}[(1)]
       \item Apply Eval-term.
       \item We have $H_{{\mathrm{1}}}  \ottsym{=}  H_{{\mathrm{2}}}$ and $P_{{\mathrm{1}}}  \ottsym{=}  P_{{\mathrm{2}}}$
         so $ \lfloor  H_{{\mathrm{1}}}  \rfloor_{ P_{{\mathrm{1}}} }   \ottsym{=}   \lfloor  H_{{\mathrm{2}}}  \rfloor_{ P_{{\mathrm{2}}} } $.
       \item It follows from $ \txtsf{dom}( W )   \ottsym{=}   \emptyset $ and $P_2 = P_1$
         that $P_{{\mathrm{2}}}  \ottsym{=}  P_{{\mathrm{1}}} \, \cup \,  \txtsf{dom}( W ) $.
       \item It is given (iii) that $H_1$ is from-scratch consistent.
         \\
         We have $H_{{\mathrm{1}}}  \ottsym{=}  H_{{\mathrm{2}}}$, so $H_2$ is from-scratch consistent.
       \end{enumerate}

     \smallskip

     \ProofCaseRule{Eval-fork}
        Similar to the Eval-term case.

    \DerivationProofCase{Eval-getClean}{
       H_{{\mathrm{1}}} ( q )   \ottsym{=}  v
    }{
       H_{{\mathrm{1}}}   \vdash ^{ p }_{ \omega }  \ottkw{get} \, \ottsym{(}   \ottkw{ref} \,\Grn{ q }   \ottsym{)}   \Downarrow   H_{{\mathrm{1}}}  \ottsym{,}  \ottsym{(}  p  \ottsym{,}  \ottkw{obs} \, v  \ottsym{,}   \txtsf{clean}   \ottsym{,}  q  \ottsym{)}  ;  \ottkw{ret} \, v 
    }

    \begin{llproof}
      \eqPfParenR{ H_2 }{ H_{{\mathrm{1}}}  \ottsym{,}  \ottsym{(}  p  \ottsym{,}  \ottkw{obs} \, v  \ottsym{,}   \txtsf{clean}   \ottsym{,}  q  \ottsym{)} }{ Given }
      \eqPf{ e }{ \ottkw{get} \, \ottsym{(}   \ottkw{ref} \,\Grn{ q }   \ottsym{)} }{ \ditto }
      \Pf{}{}{  H_{{\mathrm{1}}} ( q )   \ottsym{=}  v }{ Premise }
      \eqPf{R}{ q {:} v }  {By \Definitionref{def:rw}}
      \subseteqPf{ \txtsf{dom}( R ) }{P_{{\mathrm{1}}}}  {Given (iii)}
      \subseteqPf{\{q\}}{P_{{\mathrm{1}}}}  {$R =  q {:} v $}
      \Pf{}{}{   \lfloor  H_{{\mathrm{1}}}  \rfloor_{ P_{{\mathrm{1}}} }  ( q )   \ottsym{=}  v }  {By \Definitionref{def:restrict}}
      \Pf{}{}{   \lfloor  H_{{\mathrm{1}}}  \rfloor_{ P_{{\mathrm{1}}} }    \vdash ^{ p }_{ \omega }  e   \Downarrow    \lfloor  H_{{\mathrm{1}}}  \rfloor_{ P_{{\mathrm{1}}} }   ;  \ottkw{ret} \, v }  {\byrule{Eval-getPlain}}
      \LetPf{P_{{\mathrm{2}}}}{P_{{\mathrm{1}}}} {}
      \eqPf{ \lfloor  H_{{\mathrm{1}}}  \rfloor_{ P_{{\mathrm{1}}} } }{ \lfloor  H_{{\mathrm{2}}}  \rfloor_{ P_{{\mathrm{1}}} } }  {By def.\ of restriction}
    (1)  \Hand
      \Pf{}{}{   \lfloor  H_{{\mathrm{1}}}  \rfloor_{ P_{{\mathrm{1}}} }    \vdash ^{ p }_{ \omega }  e   \Downarrow    \lfloor  H_{{\mathrm{2}}}  \rfloor_{ P_{{\mathrm{1}}} }   ;  \ottkw{ret} \, v  }{By above equality}
      \proofsep
      \eqPf{W}{ \varepsilon }  {By \Definitionref{def:rw}}
    (2) \Hand
      \subseteqPf{ \lfloor  H_{{\mathrm{1}}}  \rfloor_{ P_{{\mathrm{1}}} } }{ \lfloor  H_{{\mathrm{2}}}  \rfloor_{ P_{{\mathrm{1}}} } }  {If $=$ then $\subseteq$}
    (3)  \Hand
      \eqPf{P_{{\mathrm{2}}}}{P_{{\mathrm{1}}} \, \cup \,  \txtsf{dom}(  \varepsilon  ) }   {$P_{{\mathrm{2}}}  \ottsym{=}  P_{{\mathrm{1}}}$}
      \Pf{}{}{\text{$H_1$ from-scratch consistent}} {Given (iii)}
    (4)  \Hand
      \Pf{}{}{\text{$H_2$ from-scratch consistent}}
             {$H_2$ differs from $H_1$ only in its edges}
   \end{llproof}

\iffalse
\fi

    \DerivationProofCase{Eval-thunkDirty}
        {
              q  \ottsym{=}   k  @  \omega 
              \\
               H_{{\mathrm{1}}} \{ q {\mapsto} e_{{\mathrm{0}}} \}   \ottsym{=}  G_{{\mathrm{2}}}
              \\
               \txtsf{dirty-paths-in} ( G_{{\mathrm{2}}} , q )   \ottsym{=}  G_{{\mathrm{3}}}
        }
        {
           H_{{\mathrm{1}}}   \vdash ^{ p }_{ \omega }   \textbf{thunk}(\Grn{  \Grn{ \ottkw{nm} \, k }  }, e_{{\mathrm{0}}} )    \Downarrow   G_{{\mathrm{3}}}  \ottsym{,}  \ottsym{(}  p  \ottsym{,}  \ottkw{alloc} \, e_{{\mathrm{0}}}  \ottsym{,}   \txtsf{clean}   \ottsym{,}  q  \ottsym{)}  ;  \ottkw{ret} \, \ottsym{(}   \ottkw{thk} \,\Grn{ q }   \ottsym{)} 
        }

    \begin{llproof}
      \eqPfParenR{ H_2 }{ G_{{\mathrm{3}}}  \ottsym{,}  \ottsym{(}  p  \ottsym{,}  \ottkw{alloc} \, e_{{\mathrm{0}}}  \ottsym{,}   \txtsf{clean}   \ottsym{,}  q  \ottsym{)} }{ Given}
      \eqPf{ e }{  \textbf{thunk}(\Grn{  \Grn{ \ottkw{nm} \, k }  }, e_{{\mathrm{0}}} )  }{ Given}
      \eqPf{q}{  k  @  \omega }{ Given }
      \eqPf{R}{ \varepsilon } {By \Definitionref{def:rw}}
      \eqPf{W}{ q {:} e_{{\mathrm{0}}} } {\ditto}
      \disjointPf{ \txtsf{dom}( W ) }{P_{{\mathrm{1}}}}  {Given}
      \notinPf{q}{P_{{\mathrm{1}}}}  {}
      \notinPf{q}{ \txtsf{dom}(  \lfloor  H_{{\mathrm{1}}}  \rfloor_{ P_{{\mathrm{1}}} }  ) }  {From \Definitionref{def:restrict}}
      \decolumnizePf
      \Pf{}{}{  \lfloor  H_{{\mathrm{1}}}  \rfloor_{ P_{{\mathrm{1}}} }    \vdash ^{ p }_{ \omega }   \textbf{thunk}(\Grn{  \Grn{ \ottkw{nm} \, k }  }, e_{{\mathrm{0}}} )    \Downarrow     \lfloor  H_{{\mathrm{1}}}  \rfloor_{ P_{{\mathrm{1}}} }  \{ q {\mapsto} e_{{\mathrm{0}}} \}   ;  \ottkw{ret} \, \ottsym{(}   \ottkw{thk} \,\Grn{ q }   \ottsym{)} }{\byrule{Eval-thunkPlain}}
      \decolumnizePf
      (3) \Hand
      \LetPf{P_2}{P_1 \union \{q\}}  {}
      \eqPf{  \lfloor  H_{{\mathrm{1}}}  \rfloor_{ P_{{\mathrm{1}}} }  \{ q {\mapsto} e_{{\mathrm{0}}} \} }{ \lfloor   H_{{\mathrm{1}}} \{ q {\mapsto} e_{{\mathrm{0}}} \}   \rfloor_{ P_{{\mathrm{1}}} } }   {$q \notin P_1$}
      \continueeqPf{ \lfloor   \txtsf{dirty-paths-in} (  H_{{\mathrm{1}}} \{ q {\mapsto} e_{{\mathrm{0}}} \}  , q )   \rfloor_{ P_{{\mathrm{1}}} } }   {$q \notin P_1$}
      \continueeqPf{ \lfloor   \txtsf{dirty-paths-in} (  H_{{\mathrm{1}}} \{ q {\mapsto} e_{{\mathrm{0}}} \}  , q )   \ottsym{,}  \ottsym{(}  p  \ottsym{,}  \ottkw{alloc} \, e_{{\mathrm{0}}}  \ottsym{,}   \txtsf{clean}   \ottsym{,}  q  \ottsym{)}  \rfloor_{ P_{{\mathrm{1}}} } }   {From \Definitionref{def:restrict}}
      \decolumnizePf
      (1) \Hand
      \Pf{}{}{  \lfloor  H_{{\mathrm{1}}}  \rfloor_{ P_{{\mathrm{1}}} }    \vdash ^{ p }_{ \omega }   \textbf{thunk}(\Grn{  \Grn{ \ottkw{nm} \, k }  }, e_{{\mathrm{0}}} )    \Downarrow     \lfloor  H_{{\mathrm{2}}}  \rfloor_{ P_{{\mathrm{1}}} }  \{ q {\mapsto} e_{{\mathrm{0}}} \}   ;  \ottkw{ret} \, \ottsym{(}   \ottkw{thk} \,\Grn{ q }   \ottsym{)} }{By above equalities}
      \decolumnizePf
      \subseteqPf{ \lfloor  H_{{\mathrm{1}}}  \rfloor_{ P_{{\mathrm{1}}} } }{ \lfloor   H_{{\mathrm{1}}} \{ q {\mapsto} e_{{\mathrm{0}}} \}   \rfloor_{ P_{{\mathrm{2}}} } }   {Immediate}
      \subseteqPf{ \lfloor  H_{{\mathrm{1}}}  \rfloor_{ P_{{\mathrm{1}}} } }{ \lfloor  G_{{\mathrm{2}}}  \rfloor_{ P_{{\mathrm{2}}} } }   {By \Definitionref{def:restrict}}
      \subseteqPf{ \lfloor  H_{{\mathrm{1}}}  \rfloor_{ P_{{\mathrm{1}}} } }{ \lfloor  G_{{\mathrm{3}}}  \rfloor_{ P_{{\mathrm{2}}} } }   {Restriction ignores edges}
      (2) \Hand
      \subseteqPf{ \lfloor  H_{{\mathrm{1}}}  \rfloor_{ P_{{\mathrm{1}}} } }{ \lfloor  H_{{\mathrm{2}}}  \rfloor_{ P_{{\mathrm{2}}} } }   {Restriction ignores edges}
      \decolumnizePf
      \Pf{}{}{\text{$H_1$ from-scratch consistent}}
               {Given (iii)}
      (4) \Hand
      \Pf{}{}{\text{$H_2$ from-scratch consistent}}
             {By \Lemmaref{lem:fsc-ext-graph}}
    \end{llproof}

    \DerivationProofCase{Eval-thunkClean}
        {
              q  \ottsym{=}   k  @  \omega 
              \\
               \txtsf{exp}( H_{{\mathrm{1}}} ,  q )   \ottsym{=}  e_{{\mathrm{0}}}
        }
        {
           H_{{\mathrm{1}}}   \vdash ^{ p }_{ \omega }   \textbf{thunk}(\Grn{  \Grn{ \ottkw{nm} \, k }  }, e_{{\mathrm{0}}} )    \Downarrow   H_{{\mathrm{1}}}  \ottsym{,}  \ottsym{(}  p  \ottsym{,}  \ottkw{alloc} \, e_{{\mathrm{0}}}  \ottsym{,}   \txtsf{clean}   \ottsym{,}  q  \ottsym{)}  ;  \ottkw{ret} \, \ottsym{(}   \ottkw{thk} \,\Grn{ q }   \ottsym{)} 
        }

    \smallskip

    \begin{llproof}
      \eqPfParenR{ H_2 }{ H_{{\mathrm{1}}}  \ottsym{,}  \ottsym{(}  p  \ottsym{,}  \ottkw{alloc} \, e_{{\mathrm{0}}}  \ottsym{,}   \txtsf{clean}   \ottsym{,}  q  \ottsym{)} }{ Given}
      \eqPf{ e }{  \textbf{thunk}(\Grn{  \Grn{ \ottkw{nm} \, k }  }, e_{{\mathrm{0}}} )  }{ Given}
      \eqPf{q}{ k  @  \omega }{ Given }
      \eqPf{W}{ q {:} e_{{\mathrm{0}}} } {By \Definitionref{def:rw}}
      \disjointPf{ \txtsf{dom}( W ) }{P_{{\mathrm{1}}}}  {Given}
      \notinPf{q}{P_{{\mathrm{1}}}}  {}
      \notinPf{q}{ \txtsf{dom}(  \lfloor  H_{{\mathrm{1}}}  \rfloor_{ P_{{\mathrm{1}}} }  ) }  {From \Definitionref{def:restrict}}
      \decolumnizePf
      \Pf{}{}{  \lfloor  H_{{\mathrm{1}}}  \rfloor_{ P_{{\mathrm{1}}} }    \vdash ^{ p }_{ \omega }   \textbf{thunk}(\Grn{  \Grn{ \ottkw{nm} \, k }  }, e_{{\mathrm{0}}} )    \Downarrow     \lfloor  H_{{\mathrm{1}}}  \rfloor_{ P_{{\mathrm{1}}} }  \{ q {\mapsto} e_{{\mathrm{0}}} \}   ;  \ottkw{ret} \, \ottsym{(}   \ottkw{thk} \,\Grn{ q }   \ottsym{)} }{\byrule{Eval-thunkPlain}}
      \decolumnizePf
      \eqPf{  \lfloor  H_{{\mathrm{1}}}  \rfloor_{ P_{{\mathrm{1}}} }  \{ q {\mapsto} e_{{\mathrm{0}}} \} }
           { \lfloor  H_{{\mathrm{1}}}  \ottsym{,}  \ottsym{(}  p  \ottsym{,}  \ottkw{alloc} \, e_{{\mathrm{0}}}  \ottsym{,}   \txtsf{clean}   \ottsym{,}  q  \ottsym{)}  \rfloor_{ P_{{\mathrm{1}}} } }
           {From \Definitionref{def:restrict} and $ \txtsf{exp}( H_{{\mathrm{1}}} ,  q )   \ottsym{=}  e_{{\mathrm{0}}}$}
      \decolumnizePf
      \Hand
(1)      \Pf{}{}{  \lfloor  H_{{\mathrm{1}}}  \rfloor_{ P_{{\mathrm{1}}} }    \vdash ^{ p }_{ \omega }   \textbf{thunk}(\Grn{  \Grn{ \ottkw{nm} \, k }  }, e_{{\mathrm{0}}} )    \Downarrow     \lfloor  H_{{\mathrm{2}}}  \rfloor_{ P_{{\mathrm{1}}} }  \{ q {\mapsto} e_{{\mathrm{0}}} \}   ;  \ottkw{ret} \, \ottsym{(}   \ottkw{thk} \,\Grn{ q }   \ottsym{)} }{By above equalities}
\end{llproof}

    \smallskip

    Since $H_{{\mathrm{2}}}  \ottsym{=}  H_{{\mathrm{1}}}  \ottsym{,}  \ottsym{(}  p  \ottsym{,}  \ottkw{alloc} \, e_{{\mathrm{0}}}  \ottsym{,}   \txtsf{clean}   \ottsym{,}  q  \ottsym{)}$,
    parts (2)--(4) are straightforward.

    \medskip

    \ProofCaseRule{Eval-refDirty}
       Similar to the Eval-thunkDirty case.

    \smallskip

    \ProofCaseRule{Eval-refClean}
       Similar to the Eval-thunkClean case.

    \smallskip

    \DerivationProofCase{Eval-bind}
        {
           H_{{\mathrm{1}}}   \vdash ^{ p }_{ \omega }  e_{{\mathrm{1}}}   \Downarrow   H'  ;  \ottkw{ret} \, v 
          \\
           H'   \vdash ^{ p }_{ \omega }  \ottsym{[}  v  \ottsym{/}  x  \ottsym{]}  e_{{\mathrm{2}}}   \Downarrow   H_{{\mathrm{2}}}  ;  \mathrm{t} 
        }
        {
           H_{{\mathrm{1}}}   \vdash ^{ p }_{ \omega }   \textbf{let}\, x \,{\leftarrow}\, e_{{\mathrm{1}}} \, \ottkw{in} \, e_{{\mathrm{2}}}    \Downarrow   H_{{\mathrm{2}}}  ;  \mathrm{t} 
        }

        \begin{llproof}
            \Pf{}{\D_1 \derives}{ H_{{\mathrm{1}}}   \vdash ^{ p }_{ \omega }  e_{{\mathrm{1}}}   \Downarrow   H'  ;  \ottkw{ret} \, v }   {Subderivation}
        (i)~~\Pf{}{}{\D_1 \satisfactory}  {By \Definitionref{def:global-sat}}
             \eqPfParenR{R} { R_{{\mathrm{1}}}  \mergesym  \ottsym{(}  R_{{\mathrm{2}}}  \ottsym{-}  W_{{\mathrm{1}}}  \ottsym{)} }    {By \Definitionref{def:rw}}
             \eqPf{W} { W_{{\mathrm{1}}}  \joinsym  W_{{\mathrm{2}}} }    {\ditto}
             \decolumnizePf
             \subseteqPf{ \txtsf{dom}(  R_{{\mathrm{1}}}  \mergesym  \ottsym{(}  R_{{\mathrm{2}}}  \ottsym{-}  W_{{\mathrm{1}}}  \ottsym{)}  ) } {P_{{\mathrm{1}}}}    {Given}
             \disjointPf{ \txtsf{dom}(  W_{{\mathrm{1}}}  \joinsym  W_{{\mathrm{2}}}  ) } {P_{{\mathrm{1}}}}    {Given}
             \Pf{}{} { \D_{{\mathrm{1}}}  \txtsf{~reads~}  R_{{\mathrm{1}}}  \txtsf{~writes~}  W_{{\mathrm{1}}} }    {\ditto}
             \Pf{}{} { \D_{{\mathrm{2}}}  \txtsf{~reads~}  R_{{\mathrm{2}}}  \txtsf{~writes~}  W_{{\mathrm{2}}} }    {\ditto}
             \decolumnizePf
        (ii)~~\subseteqPf{ \txtsf{dom}( R_{{\mathrm{1}}} ) }{P_{{\mathrm{1}}}}   {$  \txtsf{dom}(  R_{{\mathrm{1}}}  \mergesym  \ottsym{(}  R_{{\mathrm{2}}}  \ottsym{-}  W_{{\mathrm{1}}}  \ottsym{)}  )   \subseteq  P_{{\mathrm{1}}} $}
        (ii)~~\disjointPf{ \txtsf{dom}( W_{{\mathrm{2}}} ) }{P_{{\mathrm{1}}}}  {$  \txtsf{dom}(  W_{{\mathrm{1}}}  \joinsym  W_{{\mathrm{2}}}  )   \mathrel{\bot}  P_{{\mathrm{1}}} $}
          \Pf{}{}{  \lfloor  H_{{\mathrm{1}}}  \rfloor_{ P_{{\mathrm{1}}} }    \vdash ^{ p }_{ \omega }  e_{{\mathrm{1}}}   \Downarrow    \lfloor  H'  \rfloor_{ P' }   ;  \ottkw{ret} \, v }
                 {By i.h.}
          \subseteqPf{ \lfloor  H_{{\mathrm{1}}}  \rfloor_{ P_{{\mathrm{1}}} } }{ \lfloor  H'  \rfloor_{ P' } }  {\ditto}
          \eqPf{P'}{P_{{\mathrm{1}}} \, \cup \,  \txtsf{dom}( W_{{\mathrm{1}}} ) }   {\ditto}
          \Pf{}{}{\text{$H'$ from-scratch consistent}}
                  {\ditto}
          \proofsep
          \Pf{}{\D_2 \derives}{ H'   \vdash ^{ p }_{ \omega }  \ottsym{[}  v  \ottsym{/}  x  \ottsym{]}  e_{{\mathrm{2}}}   \Downarrow   H_{{\mathrm{2}}}  ;  \mathrm{t} }   {Subderivation}
        (i) \Pf{}{}{\D_2 \satisfactory}  {By \Definitionref{def:global-sat}}
             \decolumnizePf
            \subseteqPf{ \txtsf{dom}(  R_{{\mathrm{1}}}  \mergesym  \ottsym{(}  R_{{\mathrm{2}}}  \ottsym{-}  W_{{\mathrm{1}}}  \ottsym{)}  ) }{P_{{\mathrm{1}}}}   {Given}
            \subseteqPf{ \txtsf{dom}( R_{{\mathrm{2}}}  \ottsym{-}  W_{{\mathrm{1}}} ) }{P_{{\mathrm{1}}}}   {}
            \subseteqPf{ \txtsf{dom}( R_{{\mathrm{2}}} ) }{P_{{\mathrm{1}}} \, \cup \,  \txtsf{dom}( W_{{\mathrm{1}}} ) }   {}
        (ii)  \subseteqPf{ \txtsf{dom}( R_{{\mathrm{2}}} ) }{P'}   {$P'  \ottsym{=}  P_{{\mathrm{1}}} \, \cup \,  \txtsf{dom}( W_{{\mathrm{1}}} ) $}
            \decolumnizePf
            \disjointPf{ \txtsf{dom}(  W_{{\mathrm{1}}}  \joinsym  W_{{\mathrm{2}}}  ) }{P_{{\mathrm{1}}}}   {Given}
            \disjointPf{ \txtsf{dom}( W_{{\mathrm{1}}} ) }{ \txtsf{dom}( W_{{\mathrm{2}}} ) }   {From def.\ of $\joinsym$}
            \disjointPf{ \txtsf{dom}( W_{{\mathrm{2}}} ) }{P_{{\mathrm{1}}} \, \cup \,  \txtsf{dom}( W_{{\mathrm{1}}} ) }   {}
        (ii) \disjointPf{ \txtsf{dom}( W_{{\mathrm{2}}} ) }{P'}   {$P'  \ottsym{=}  P_{{\mathrm{1}}} \, \cup \,  \txtsf{dom}( W_{{\mathrm{1}}} ) $}
            \decolumnizePf
            \Pf{}{}{  \lfloor  H'  \rfloor_{ P' }    \vdash ^{ p }_{ \omega }  \ottsym{[}  v  \ottsym{/}  x  \ottsym{]}  e_{{\mathrm{2}}}   \Downarrow    \lfloor  H_{{\mathrm{2}}}  \rfloor_{ P_{{\mathrm{2}}} }   ;  \mathrm{t} }
                   {By i.h.}
            \subseteqPf{ \lfloor  H'  \rfloor_{ P' } }{ \lfloor  H_{{\mathrm{2}}}  \rfloor_{ P_{{\mathrm{2}}} } }  {\ditto}
            \eqPf{P_{{\mathrm{2}}}}{P' \, \cup \,  \txtsf{dom}( W_{{\mathrm{2}}} ) }
                   {\ditto}
        (4) \Hand
            \Pf{}{}{\text{$H_2$ from-scratch consistent}}
                   {\ditto}
        (2) \Hand
            \subseteqPf{ \lfloor  H_{{\mathrm{1}}}  \rfloor_{ P_{{\mathrm{1}}} } }{ \lfloor  H_{{\mathrm{2}}}  \rfloor_{ P_{{\mathrm{2}}} } }  {By transitivity of $\subseteq$}
           \eqPf{P_{{\mathrm{2}}}}{P_{{\mathrm{1}}} \, \cup \,  \txtsf{dom}( W_{{\mathrm{1}}} )  \, \cup \,  \txtsf{dom}( W_{{\mathrm{2}}} ) }
                 {By $P'  \ottsym{=}  P_{{\mathrm{1}}} \, \cup \,  \txtsf{dom}( W_{{\mathrm{1}}} ) $}
        (3) \Hand
           \eqPf{P_{{\mathrm{2}}}}{P_{{\mathrm{1}}} \, \cup \,  \txtsf{dom}( W ) }
                 {By $W  \ottsym{=}   W_{{\mathrm{1}}}  \joinsym  W_{{\mathrm{2}}} $}
          \proofsep
          \Pf{}{}{  \lfloor  H_{{\mathrm{1}}}  \rfloor_{ P_{{\mathrm{1}}} }    \vdash ^{ p }_{ \omega }  e_{{\mathrm{1}}}   \Downarrow    \lfloor  H'  \rfloor_{ P' }   ;  \ottkw{ret} \, v }
                 {Above}
          \Pf{}{}{  \lfloor  H'  \rfloor_{ P' }    \vdash ^{ p }_{ \omega }  \ottsym{[}  v  \ottsym{/}  x  \ottsym{]}  e_{{\mathrm{2}}}   \Downarrow    \lfloor  H_{{\mathrm{2}}}  \rfloor_{ P_{{\mathrm{2}}} }   ;  \mathrm{t} }
                 {Above}
          \proofsep
      (1) \Hand
          \Pf{}{}{  \lfloor  H_{{\mathrm{1}}}  \rfloor_{ P_{{\mathrm{1}}} }    \vdash ^{ p }_{ \omega }   \textbf{let}\, x \,{\leftarrow}\, e_{{\mathrm{1}}} \, \ottkw{in} \, e_{{\mathrm{2}}}    \Downarrow    \lfloor  H_{{\mathrm{2}}}  \rfloor_{ P_{{\mathrm{2}}} }   ;  \mathrm{t} }
                 {\byrule{Eval-bind}}
        \end{llproof}

        \medskip

        \ProofCaseRule{Eval-app}
           Similar to the Eval-bind case.

        \ProofCaseRule{Eval-nest}
             Similar to the Eval-bind case.

        \ProofCaseRule{Eval-fix}
             The input and output graphs of the subderivation match those of the conclusion,
             as do the read and write sets according to \Definitionref{def:rw}.
             Thus, we can just use the i.h.\ and apply Eval-fix.

        \ProofCaseRule{Eval-case}
             Similar to the Eval-fix case.

        \ProofCaseRule{Eval-split}
             Similar to the Eval-case case.

        \ProofCaseRule{Eval-namespace}
             Similar to the Eval-fix case.

    \DerivationProofCase{Eval-computeDep}
        {
          \arrayenvbl{
               \txtsf{exp}( H_{{\mathrm{1}}} ,  q )   \ottsym{=}  e'  \\
               \txtsf{del-edges-out} (  H_{{\mathrm{1}}} \{ q {\mapsto} e' \}  , q )   \ottsym{=}  G'_{{\mathrm{1}}}  \\
               \D_{{\mathrm{1}}}  \derives   G'_{{\mathrm{1}}}   \vdash ^{ q }_{  \txtsf{namespace}(  q  )  }  e'   \Downarrow   G_{{\mathrm{2}}}  ;  \mathrm{t}'    \\
          }
          \\
          \arrayenvbl{
             G_{{\mathrm{2}}} \{ q {\mapsto}( e' , \mathrm{t}' )\}   \ottsym{=}  G'_{{\mathrm{2}}} \\
             \txtsf{all-clean-out}( G'_{{\mathrm{2}}} , q )  \\
             \D_{{\mathrm{2}}}  \derives   G'_{{\mathrm{2}}}   \vdash ^{ p }_{ \omega }  \ottkw{force} \, \ottsym{(}   \ottkw{thk} \,\Grn{ p_{{\mathrm{0}}} }   \ottsym{)}   \Downarrow   H_{{\mathrm{2}}}  ;  \mathrm{t}  
          }
        }
        {
           H_{{\mathrm{1}}}   \vdash ^{ p }_{ \omega }  \ottkw{force} \, \ottsym{(}   \ottkw{thk} \,\Grn{ p_{{\mathrm{0}}} }   \ottsym{)}   \Downarrow   H_{{\mathrm{2}}}  ;  \mathrm{t} 
        }

    \begin{llproof}
      \eqPf{e} {\ottkw{force} \, \ottsym{(}   \ottkw{thk} \,\Grn{ p_{{\mathrm{0}}} }   \ottsym{)} } {Given}
      \eqPfParenR{R} { R_{{\mathrm{1}}}  \mergesym  R_{{\mathrm{2}}} }   {By \Definitionref{def:rw}}
      \eqPf{W} {W_{{\mathrm{2}}}}   {\ditto}
      \Pf{}{} { \D_{{\mathrm{1}}}  \txtsf{~reads~}  R_{{\mathrm{1}}}  \txtsf{~writes~}  W_{{\mathrm{1}}} }    {\ditto}
      \Pf{}{} { \D_{{\mathrm{2}}}  \txtsf{~reads~}  R_{{\mathrm{2}}}  \txtsf{~writes~}  W }    {\ditto}
      \subseteqPf{ \txtsf{dom}( W_{{\mathrm{1}}} ) }{ \txtsf{dom}( W_{{\mathrm{2}}} ) }    {\ditto}
    \end{llproof}

    \smallskip

    We don't immediately need to apply the i.h.\ to $\D_1$, because that computation will be done later
    in the reference derivation.
    But we do need to apply the i.h.\ to
    $\D_2 \derives  G'_{{\mathrm{2}}}   \vdash ^{ p }_{ \omega }  e   \Downarrow   H_{{\mathrm{2}}}  ;  \mathrm{t} $.
    So we need to show part (iii) of the statement, which says that each \emph{cached} computation
    in the input graph is consistent with respect to an earlier version of the graph.

    Since we're adding such a computation $e'$ in $G'_{{\mathrm{2}}}$, which is the input
    graph of $\D_2$, we have to show that the computation of $e'$ (by $\D_1$) is consistent,
    which means applying the i.h.\ to $\D_1$.

    \smallskip

    \begin{llproof}
      (iii) \subseteqPf{ \txtsf{dom}( R_{{\mathrm{1}}} ) }{P_{{\mathrm{1}}}}   {$  \txtsf{dom}( R )   \subseteq  P_{{\mathrm{1}}} $}
      (iii) \disjointPf{ \txtsf{dom}( W_{{\mathrm{1}}} ) }{P_{{\mathrm{1}}}}   {$  \txtsf{dom}( W )   \mathrel{\bot}  P_{{\mathrm{1}}} $}
      \Pf{}{}{\text{$H_1$ from-scratch consistent}}
             {Given (iii)}
      \Pf{}{}{\text{$G_1'$ from-scratch consistent}}
             {$G_1'$ differs from $H_1$ only in its edges}
             \trailingjust{(note that $q$ points to $e'$ but has no cached result)}
      \Pf{}{\D_1 \derives}{ G'_{{\mathrm{1}}}   \vdash ^{ q }_{  \txtsf{namespace}(  q  )  }  e'   \Downarrow   G_{{\mathrm{2}}}  ;  \mathrm{t}' }  {Subderivation}
      (1) \Pf{}{}{  \lfloor  G'_{{\mathrm{1}}}  \rfloor_{ P_{{\mathrm{1}}} }    \vdash ^{ q }_{  \txtsf{namespace}(  q  )  }  e'   \Downarrow    \lfloor  G_{{\mathrm{2}}}  \rfloor_{ P'_{{\mathrm{1}}} }   ;  \mathrm{t}' }  {By i.h.}
      (2) \subseteqPf{ \lfloor  G'_{{\mathrm{1}}}  \rfloor_{ P_{{\mathrm{1}}} } }{ \lfloor  G_{{\mathrm{2}}}  \rfloor_{ P'_{{\mathrm{1}}} } } {\ditto}
      (3) \eqPf{P'_{{\mathrm{1}}}}{P_{{\mathrm{1}}} \, \cup \,  \txtsf{dom}( W_{{\mathrm{1}}} ) }   {\ditto}
      (4) \Pf{}{}{\text{$G_2$ from-scratch consistent}}
             {\ditto}
      \subseteqPf{ \lfloor  G_{{\mathrm{2}}}  \rfloor_{ P'_{{\mathrm{1}}} } }{ \lfloor  G_{{\mathrm{2}}}  \rfloor_{ P'_{{\mathrm{1}}} } }  {If $=$ then $\subseteq$}
      \proofsep
      \Pf{}{}{\text{$G_2'$ from-scratch consistent}}
             {By \Lemmaref{lem:fsc-ext-graph}}
    \end{llproof}

    \smallskip

    Having ``stowed away'' the consistency of $e'$, we can move on to $\D_2$.

    \smallskip

    \begin{llproof}
      \proofsep
           \Pf{}{\D_2 \derives}{  G'_{{\mathrm{2}}}   \vdash ^{ p }_{ \omega }  e   \Downarrow   H_{{\mathrm{2}}}  ;  \mathrm{t}  }{ Subderivation }
      (i)~~~\Pf{}{}{ \D_2 \satisfactory }{ $\D_2$ is a subderivation of $\D_{i}$ }
           \subseteqPf{ \txtsf{dom}( R ) }{P_{{\mathrm{1}}} \AND   \txtsf{dom}( W )   \mathrel{\bot}  P_{{\mathrm{1}}} }   {Given}
      (ii)~~~\subseteqPf{ \txtsf{dom}( R_{{\mathrm{2}}} ) }{P_{{\mathrm{1}}} \AND   \txtsf{dom}( W_{{\mathrm{2}}} )   \mathrel{\bot}  P_{{\mathrm{1}}} }   {Using above equalities}
      (iii)~~~\Pf{}{}{\text{$H_1$ from-scratch consistent}}
             {\ditto}
      (1)~~~\Pf{}{}{   \lfloor  G'_{{\mathrm{2}}}  \rfloor_{ P_{{\mathrm{1}}} }    \vdash ^{ p }_{ \omega }  e   \Downarrow    \lfloor  H_{{\mathrm{2}}}  \rfloor_{ P_{{\mathrm{2}}} }   ;  \mathrm{t}  }{ By i.h.}
      (2)~~~\subseteqPf{ \lfloor  G'_{{\mathrm{2}}}  \rfloor_{ P_{{\mathrm{1}}} } }{ \lfloor  H_{{\mathrm{2}}}  \rfloor_{ P_{{\mathrm{2}}} } }   {\ditto}
      \Hand
      (3)~~~\eqPf{P_{{\mathrm{2}}}}{P_{{\mathrm{1}}} \, \cup \,  \txtsf{dom}( W ) }   {\ditto}
      \Hand
      (4)~~~\Pf{}{}{\text{$H_2$ from-scratch consistent}}
             {\ditto}
      \proofsep
      \eqPf{ \lfloor  H_{{\mathrm{1}}}  \rfloor_{ P_{{\mathrm{1}}} } }{ \lfloor   H_{{\mathrm{1}}} \{ q {\mapsto} e' \}   \rfloor_{ P_{{\mathrm{1}}} } }  {Follows from $ \txtsf{exp}( H_{{\mathrm{1}}} ,  q )   \ottsym{=}  e'$}
      \continueeqPf{ \lfloor   \txtsf{del-edges-out} (  H_{{\mathrm{1}}} \{ q {\mapsto} e' \}  , q )   \rfloor_{ P_{{\mathrm{1}}} } }  {\Definitionref{def:restrict} ignores edges}
      \continueeqPf{ \lfloor  G'_{{\mathrm{1}}}  \rfloor_{ P_{{\mathrm{1}}} } }  {By above equality}
    \end{llproof}

    \smallskip

    We need to show $ \lfloor  G'_{{\mathrm{1}}}  \rfloor_{ P_{{\mathrm{1}}} }   \ottsym{=}   \lfloor  G_{{\mathrm{2}}}  \rfloor_{ P_{{\mathrm{1}}} } $.  That is,
    evaluating $e'$---which will be done \emph{inside} the reference derivation's version
    of $\D_2$---doesn't change anything in $P_{{\mathrm{1}}}$.

    Fortunately, we know that $  \txtsf{dom}( W )   \mathrel{\bot}  P_{{\mathrm{1}}} $ and $  \txtsf{dom}( W_{{\mathrm{1}}} )   \subseteq   \txtsf{dom}( W )  $.
    Therefore $  \txtsf{dom}( W_{{\mathrm{1}}} )   \mathrel{\bot}  P_{{\mathrm{1}}} $.

    By \Lemmaref{lem:respect-write}, $G_{{\mathrm{2}}}$ agrees with $G'_{{\mathrm{1}}}$ on
    $ \txtsf{dom}( G'_{{\mathrm{1}}} )   \ottsym{-}   \txtsf{dom}( W_{{\mathrm{1}}} ) $.  Since $W_1$ is disjoint from $P_1$,
    we have that $G_{{\mathrm{2}}}$ agrees with $G'_{{\mathrm{1}}}$ on $P_1$.

    Therefore $ \lfloor  G'_{{\mathrm{1}}}  \rfloor_{ P_{{\mathrm{1}}} }   \ottsym{=}   \lfloor  G_{{\mathrm{2}}}  \rfloor_{ P_{{\mathrm{1}}} } $.

    Now we'll show that $ \lfloor  G_{{\mathrm{2}}}  \rfloor_{ P_{{\mathrm{1}}} }   \ottsym{=}   \lfloor  G'_{{\mathrm{2}}}  \rfloor_{ P_{{\mathrm{1}}} } $, that is,
    $ \lfloor  G_{{\mathrm{2}}}  \rfloor_{ P_{{\mathrm{1}}} }   \ottsym{=}   \lfloor   G_{{\mathrm{2}}} \{ q {\mapsto}( e' , \mathrm{t}' )\}   \rfloor_{ P_{{\mathrm{1}}} } $.

    \begin{itemize}
    \item     If $q \notin P_1$ then this follows easily from \Definitionref{def:restrict}.
    \item    Otherwise, $q \in P_1$.  We have $ \txtsf{exp}( G'_{{\mathrm{1}}} ,  q )   \ottsym{=}  e'$
      and therefore $ \txtsf{exp}( G_{{\mathrm{2}}} ,  q )   \ottsym{=}  e'$, so updating $G_{{\mathrm{2}}}$ with $q$ pointing
      to $e'$ doesn't change the restriction.
    \end{itemize}

    \smallskip

    \begin{llproof}
      \eqPf{ \lfloor  H_{{\mathrm{1}}}  \rfloor_{ P_{{\mathrm{1}}} } }{ \lfloor  G'_{{\mathrm{2}}}  \rfloor_{ P_{{\mathrm{1}}} } } {Shown above}
      \Hand
      (1)~~\Pf{}{}{   \lfloor  H_{{\mathrm{1}}}  \rfloor_{ P_{{\mathrm{1}}} }    \vdash ^{ p }_{ \omega }  e   \Downarrow    \lfloor  H_{{\mathrm{2}}}  \rfloor_{ P_{{\mathrm{2}}} }   ;  \mathrm{t} }  {By above equality}
      \Hand
      (2)~~\subseteqPf{ \lfloor  H_{{\mathrm{1}}}  \rfloor_{ P_{{\mathrm{1}}} } }{ \lfloor  H_{{\mathrm{2}}}  \rfloor_{ P_{{\mathrm{2}}} } }   {\ditto}
    \end{llproof}

    \smallskip

    \DerivationProofCase{Eval-forceClean}
        {
           H_{{\mathrm{1}}} ( q ) = ( e , \mathrm{t} ) 
          \\
           \txtsf{all-clean-out}( H_{{\mathrm{1}}} , q ) 
        }
        {
           H_{{\mathrm{1}}}   \vdash ^{ p }_{ \omega }  \ottkw{force} \, \ottsym{(}   \ottkw{thk} \,\Grn{ q }   \ottsym{)}   \Downarrow   H_{{\mathrm{1}}}  \ottsym{,}  \ottsym{(}  p  \ottsym{,}  \ottkw{obs} \, \mathrm{t}  \ottsym{,}   \txtsf{clean}   \ottsym{,}  q  \ottsym{)}  ;  \mathrm{t} 
        }

        \begin{llproof}
          \eqPfParenR{H_{{\mathrm{2}}}}{H_{{\mathrm{1}}}  \ottsym{,}  \ottsym{(}  p  \ottsym{,}  \ottkw{obs} \, \mathrm{t}  \ottsym{,}   \txtsf{clean}   \ottsym{,}  q  \ottsym{)}}  {Given}
          \proofsep
          \eqPf{ H_{{\mathrm{1}}} ( q ) }{(e,t)}  {Premise}
          \Pf{}{}{\text{$H_1$ from-scratch consistent over $P_1$}} {Given (iii)}
          \Pf{}{\D_q}{\satisfactory}  {\Definitionref{def:fsc}~(i)}
          \subseteqPf{ \txtsf{dom}( R_q ) }{P_q}  {\ditto~(ii)}
          \disjointPf{ \txtsf{dom}( W_q ) }{P_q}  {\ditto~(ii)}
          \eqPfParenR{R}{R_q  \ottsym{,}   q {:}( e , \mathrm{t} ) }   {By \Definitionref{def:rw} for Eval-forceClean ($\D' = \D_q$)}
          \eqPf{W}{W_q}   {\ditto}
        \end{llproof}

        \smallskip

        To show that $\D_q$ is consistent, we use
        assumption (iii) that $H_1$ is from-scratch consistent.
        We have $ q {:}( e , \mathrm{t} ) $ in $H_{{\mathrm{1}}}$.  By \Definitionref{def:fsc-graph}),
        $\D_q \derives  G_q   \vdash ^{ p_q }_{ \omega_q }  e   \Downarrow   G_q'  ;  \mathrm{t} $
        is from-scratch consistent for some $ P_q  \subseteq   \txtsf{dom}( G_q )  $
        up to $H_{{\mathrm{1}}}$.

        Now we turn to \Definitionref{def:fsc}.

        \smallskip

        \begin{llproof}
          \Pf{}{}{  \lfloor  G_q  \rfloor_{ P_q }    \vdash ^{ p_q }_{ \omega_q }  e   \Downarrow    \lfloor  G_q'  \rfloor_{ P_q \, \cup \,  \txtsf{dom}( W_q )  }   ;  \mathrm{t} }  {By \Definitionref{def:fsc} (1)}
          \subseteqPf{ \lfloor  G_q  \rfloor_{ P_q } }{ \lfloor  G_q'  \rfloor_{ P_q \, \cup \,  \txtsf{dom}( W_q )  } }  {By \Definitionref{def:fsc} (2)}
          \subseteqPf{ \lfloor  G_q'  \rfloor_{ P_q \, \cup \,  \txtsf{dom}( W_q )  } }{ \lfloor  H_{{\mathrm{1}}}  \rfloor_{ P_{{\mathrm{1}}} } }  {By \Definitionref{def:fsc} (4)}
          \proofsep
          \Pf{}{}{  \lfloor  H_{{\mathrm{1}}}  \rfloor_{ P_{{\mathrm{1}}} }    \vdash ^{ p_q }_{ \omega_q }  e   \Downarrow    \lfloor  H_{{\mathrm{1}}}  \rfloor_{ P_{{\mathrm{1}}} \, \cup \,  \txtsf{dom}( W_q )  }   ;  \mathrm{t} }  {By \Lemmaref{lem:weakening-ni}}
          \proofsep
          \LetPf{P_{{\mathrm{2}}}}{P_{{\mathrm{1}}} \, \cup \,  \txtsf{dom}( W_q ) }   {}
          \decolumnizePf
          \Pf{}{}{  \lfloor  H_{{\mathrm{1}}}  \rfloor_{ P_{{\mathrm{1}}} }    \vdash ^{ p }_{ \omega }  e   \Downarrow    \lfloor  H_{{\mathrm{1}}}  \rfloor_{ P_{{\mathrm{2}}} }   ;  \mathrm{t} }  {By above equality}
          \Pf{}{}{  \lfloor  H_{{\mathrm{1}}}  \rfloor_{ P_{{\mathrm{1}}} }    \vdash ^{ p }_{ \omega }  \ottkw{force} \, \ottsym{(}   \ottkw{thk} \,\Grn{ q }   \ottsym{)}   \Downarrow    \lfloor  H_{{\mathrm{1}}}  \rfloor_{ P_{{\mathrm{2}}} }   ;  \mathrm{t} }  {\byrule{Eval-forcePlain}}
          \Pf{}{}{  \lfloor  H_{{\mathrm{1}}}  \rfloor_{ P_{{\mathrm{1}}} }    \vdash ^{ p }_{ \omega }  \ottkw{force} \, \ottsym{(}   \ottkw{thk} \,\Grn{ q }   \ottsym{)}   \Downarrow    \lfloor  H_{{\mathrm{1}}}  \ottsym{,}  \ottsym{(}  p  \ottsym{,}  \ottkw{obs} \, \mathrm{t}  \ottsym{,}   \txtsf{clean}   \ottsym{,}  q  \ottsym{)}  \rfloor_{ P_{{\mathrm{2}}} }   ;  \mathrm{t} }  {By \Definitionref{def:restrict}}
        (1) \Hand
          \Pf{}{}{  \lfloor  H_{{\mathrm{1}}}  \rfloor_{ P_{{\mathrm{1}}} }    \vdash ^{ p }_{ \omega }  \ottkw{force} \, \ottsym{(}   \ottkw{thk} \,\Grn{ q }   \ottsym{)}   \Downarrow    \lfloor  H_{{\mathrm{2}}}  \rfloor_{ P_{{\mathrm{2}}} }   ;  \mathrm{t} }  {By above equality}
          \decolumnizePf
           \subseteqPf{ \lfloor  H_{{\mathrm{1}}}  \rfloor_{ P_{{\mathrm{1}}} } }{ \lfloor  H_{{\mathrm{1}}}  \rfloor_{ P_{{\mathrm{2}}} } }   {By a property of \Definitionref{def:restrict}}
        (2) \Hand
           \subseteqPf{ \lfloor  H_{{\mathrm{1}}}  \rfloor_{ P_{{\mathrm{1}}} } }{ \lfloor  H_{{\mathrm{2}}}  \rfloor_{ P_{{\mathrm{2}}} } }   {By a property of \Definitionref{def:restrict}}
        (3) \Hand
          \eqPf{P_{{\mathrm{2}}}}{P_{{\mathrm{1}}} \, \cup \,  \txtsf{dom}( W ) }   {$W = W_q$}
        (4) \Hand
          \Pf{}{}{\text{$H_2$ from-scratch consistent}}
             {$H_2$ differs from $H_1$ only in its edges}
        \end{llproof}

        \smallskip

    \DerivationProofCase{Eval-scrubEdge}
    {
      \arrayenvbl{
         \txtsf{all-clean-out}( \ottsym{(}  G_{{\mathrm{1}}}  \ottsym{,}  G_{{\mathrm{2}}}  \ottsym{)} , q_{{\mathrm{2}}} ) 
        \\
         \txtsf{consistent-action} ( \ottsym{(}  G_{{\mathrm{1}}}  \ottsym{,}  G_{{\mathrm{2}}}  \ottsym{)} ,  a ,  q_{{\mathrm{2}}}  ) 
        \\
         G_{{\mathrm{1}}}  \ottsym{,}  \ottsym{(}  q_{{\mathrm{1}}}  \ottsym{,}  a  \ottsym{,}   \txtsf{clean}   \ottsym{,}  q_{{\mathrm{2}}}  \ottsym{)}  \ottsym{,}  G_{{\mathrm{2}}}   \vdash ^{ p }_{ \omega }  \ottkw{force} \, \ottsym{(}   \ottkw{thk} \,\Grn{ p_{{\mathrm{0}}} }   \ottsym{)}   \Downarrow   H_{{\mathrm{2}}}  ;  \mathrm{t} 
      }
    }
    {
       G_{{\mathrm{1}}}  \ottsym{,}  \ottsym{(}  q_{{\mathrm{1}}}  \ottsym{,}  a  \ottsym{,}   \txtsf{dirty}   \ottsym{,}  q_{{\mathrm{2}}}  \ottsym{)}  \ottsym{,}  G_{{\mathrm{2}}}   \vdash ^{ p }_{ \omega }  \ottkw{force} \, \ottsym{(}   \ottkw{thk} \,\Grn{ p_{{\mathrm{0}}} }   \ottsym{)}   \Downarrow   H_{{\mathrm{2}}}  ;  \mathrm{t} 
    }

    \begin{llproof}
      \eqPf{e}{\ottkw{force} \, \ottsym{(}   \ottkw{thk} \,\Grn{ p_{{\mathrm{0}}} }   \ottsym{)}} {Given}
      \eqPfParenR{H_{{\mathrm{1}}}}{G_{{\mathrm{1}}}  \ottsym{,}  \ottsym{(}  q_{{\mathrm{1}}}  \ottsym{,}  a  \ottsym{,}   \txtsf{dirty}   \ottsym{,}  q_{{\mathrm{2}}}  \ottsym{)}  \ottsym{,}  G_{{\mathrm{2}}} }{Given}
      \LetPf{G'}{G_{{\mathrm{1}}}  \ottsym{,}  \ottsym{(}  q_{{\mathrm{1}}}  \ottsym{,}  a  \ottsym{,}   \txtsf{clean}   \ottsym{,}  q_{{\mathrm{2}}}  \ottsym{)}  \ottsym{,}  G_{{\mathrm{2}}}}{}
      \Pf{}{}{   \lfloor  G'_{{\mathrm{1}}}  \rfloor_{ P_{{\mathrm{1}}} }    \vdash ^{ p }_{ \omega }  e   \Downarrow    \lfloor  H_{{\mathrm{2}}}  \rfloor_{ P_{{\mathrm{2}}} }   ;  \mathrm{t}  }{By i.h.}
      \subseteqPf{ \lfloor  G'_{{\mathrm{1}}}  \rfloor_{ P_{{\mathrm{1}}} } }{ \lfloor  H_{{\mathrm{2}}}  \rfloor_{ P_{{\mathrm{2}}} } }  {\ditto}
      (3)
      \Hand
      \eqPf{P_{{\mathrm{2}}}}{P_{{\mathrm{1}}} \, \cup \,  \txtsf{dom}( W ) }{ \ditto }
      (4)
      \Hand
      \Pf{}{}{\text{$H_2$ from-scratch consistent}}
          {\ditto}
      \eqPf{ \lfloor  H_{{\mathrm{1}}}  \rfloor_{ P_{{\mathrm{1}}} } }{ \lfloor  G'  \rfloor_{ P_{{\mathrm{1}}} } }  {\Definitionref{def:restrict} ignores edges}
      (1)
      \Hand
      \Pf{}{}{   \lfloor  H_{{\mathrm{1}}}  \rfloor_{ P_{{\mathrm{1}}} }    \vdash ^{ p }_{ \omega }  e   \Downarrow    \lfloor  H_{{\mathrm{2}}}  \rfloor_{ P_{{\mathrm{2}}} }   ;  \mathrm{t}  }{By above equality}
      (2)
      \Hand
      \subseteqPf{ \lfloor  H_{{\mathrm{1}}}  \rfloor_{ P_{{\mathrm{1}}} } }{ \lfloor  H_{{\mathrm{2}}}  \rfloor_{ P_{{\mathrm{2}}} } }  {By above equality}
    \end{llproof}
  \qedhere
  \end{itemize}
\end{proof}

\fi
\end{document}